\documentclass{article}
\usepackage{ijcai26}

\usepackage{times}
\usepackage{soul}
\usepackage{url}
\usepackage[hidelinks]{hyperref}
\usepackage[utf8]{inputenc}
\usepackage[small]{caption}
\usepackage{graphicx}
\usepackage{amsmath}
\usepackage{amsthm}
\usepackage{mathtools}
\usepackage{amssymb}
\usepackage{booktabs}
\usepackage[linesnumbered,ruled,vlined]{algorithm2e}
\usepackage[switch]{lineno}
\usepackage{xspace}
\usepackage{subcaption}
\usepackage{multirow}
\usepackage{xcolor}

\usepackage{cleveref}

\newtheorem{theorem}{Theorem}

\newtheorem{lemma}{Lemma}
\newtheorem{definition}{Definition}

\DeclareMathOperator*{\argmax}{arg\,max}

\DeclarePairedDelimiter{\floor}{\lfloor}{\rfloor}
\DeclareMathOperator*{\omegaV}{\boldsymbol{\omega}}
\DeclareMathOperator*{\omegaVlb}{\boldsymbol{\omega_{lb}}}
\DeclareMathOperator*{\omegaVub}{\boldsymbol{\omega_{ub}}}

\newcommand{\warehouselargeW}{\texttt{warehouse-33-36}\xspace}
\newcommand{\randomSmall}{\texttt{random-32-32-20}\xspace}
\newcommand{\roomLarge}{\texttt{room-64-64-8}\xspace}

\newcommand{\emptyMid}{\texttt{empty-48-48}\xspace}

\newcommand{\denSmall}{\texttt{den312d}\xspace}

\newcommand{\citet}[1]{\citeauthor{#1} [\citeyear{#1}]}

\newcommand{\ERS}{ERS\xspace}

\title{Optimization of Edge Directions and Weights for Mixed Guidance Graphs in Lifelong Multi-Agent Path Finding}

\author{
Yulun Zhang$^1$
\and
Varun Bhatt$^2$\and
Matthew C. Fontaine$^{3}$\footnote{Work done before joining Lila Sciences.}\and
Stefanos Nikolaidis$^2$\And
Jiaoyang Li$^1$
\affiliations
$^1$Robotics Institute, Carnegie Mellon University
\hspace{1cm}$^3$Lila Sciences\\
$^2$Thomas Lord Department of Computer Science, University of Southern California
\emails
yulunzhang@cmu.edu,
vsbhatt@usc.edu,
mfontaine@lila.ai,
nikolaid@usc.edu,
jiaoyangli@cmu.edu
}

\begin{document}

\maketitle

% why use rotation model
% Application: differential drive has rotation

% why we want direction
% Reduce search space of search-based mapf algo
% Greedy: help to prune bad action.

% Key insight: with MAPF-rotation model, we incorporate edge direction into ggo via three ways:
% 1. Directly optimize edge directions via EA while preserving the strong connectivity of a graph
% 2. Incoroporate into observation/output of update model
% 2.1. Into observation: CMA-ES-PU to optimize for edge weights
% 2.2. Into output and use weighted sum objective (tp + w * edge_sim) to optimize for both edge direction and edge weights

% Outline
% lifelong mapf is important with many applications
% a prior work has attempted to improve throughput via optimizing guidance graphs, providing preferences on movement between (u, v) and (v, u)
% however, the prior work tested GGO with the pebble motion model, simplifying the movement of the agents
% while considering the rotational model, it is sometime desirable to block some edges because the agents will have less opportunities to make turns, which can slow down the movement of them. But there is a trade-off because blocking too many edges can make agents take unnecessary detours.
% In this paper we introduce the concept of mixed Guidance Graph Optimization (MGGO), which incorporate edge directions into the guidance graphs optimization.

\begin{abstract}
% We study how to incorporate edge direction into Guidance Graph Optimization (GGO) for lifelong Multi-Agent Path Finding (MAPF) to improve throughput.
Multi-Agent Path Finding (MAPF) aims to move agents from their start to goal vertices on a graph. Lifelong MAPF (LMAPF) continuously assigns new goals to agents as they complete current ones.
To guide agents' movement in LMAPF, prior works have proposed Guidance Graph Optimization (GGO) methods to optimize a guidance graph, which is a bidirected weighted graph whose directed edges represent moving and waiting actions with edge weights being action costs. Higher edge weights represent higher action costs.
However, edge weights only provide \emph{soft} guidance. An edge with a high weight only \emph{discourages} agents from using it, instead of \emph{prohibiting} agents from traversing it.
In this paper, we explore the need to incorporate edge directions optimization into GGO, providing \emph{strict} guidance. We generalize GGO to Mixed Guidance Graph Optimization (MGGO), presenting two MGGO methods capable of optimizing both edge weights and directions. The first optimizes edge directions and edge weights in two phases separately. The second applies Quality Diversity algorithms to optimize a neural network capable of generating edge directions and weights. We also incorporate traffic patterns relevant to edge directions into a GGO method, making it capable of generating edge-direction-aware guidance graphs.

% Prior work on GGO focuses on optimizing edge weights to guide agents' movement and improve throughput, where a higher edge weight represents higher action cost during planning. However, guidance graphs only provide \emph{soft} guidance. 

% However, modifying edge weights alone provides only \emph{soft} guidance: even edges with high costs remain traversable, allowing agents to select suboptimal routes. Moreover, LMAPF is commonly solved by decomposing it into a sequence of MAPF instances, which leads to inherently myopic planning decisions. As a result, agents may still traverse high-cost edges, causing congestion and limiting throughput.
% To address these limitations, we generalize GGO to explicitly optimize edge directions, introducing Mixed Guidance Graph Optimization (MGGO). We propose two MGGO methods. The first optimizes edge directions and edge weights in two phases separately. The second jointly optimizes a neural network capable of generating edge directions and weights based on Quality Diversity (QD) algorithms. In addition, we introduce an improved edge-weight optimization strategy that enhances the effectiveness of prior GGO methods.
\end{abstract}

\section{Introduction}

We study the problem of leveraging a mixed guidance graph with optimized edge weights and directions to guide agents' movement, improving throughput of LMAPF algorithms. Multi-Agent Path Finding (MAPF)~\cite{SternSoCS19} aims to move agents from their corresponding start to goal vertices on a given graph $G$ without collisions, while lifelong MAPF (LMAPF)~\cite{Li2020LifelongMP} continuously assigns new goal vertices to agents while they finish the current ones. Applications of LMAPF include autonomous warehouses~\cite{Li2020LifelongMP}, robotic sorting systems~\cite{zhang2025tmo}, and video games~\cite{MaAIIDE17,Jansen2008DirectionMF}.

LMAPF problems are solved by decomposing them into a sequence of MAPF instances and solving them in order. Inevitably, planned paths are myopic, which could result in traffic congestion. To mitigate this issue, \citet{zhang2024ggo} propose a guidance graph, a bidirected weighted graph in which each directed edge $(u, v)$ represents a moving or waiting action from vertex $u$ to $v$ and its weight represents the action cost, where higher weights represent higher costs. Agents are asked to plan paths that minimize the sum of action costs on the guidance graph.
Two Guidance Graph Optimization (GGO) techniques are presented to optimize the edge weights with the objective of maximizing throughput, defined as the number of goals reached per timestep.
% A well-optimized guidance graph provides global guidance by manipulating the edge weights, where an edge with higher weights indicates higher costs, discouraging agents to move through it. 
A well-optimized guidance graph reduces traffic congestion by encouraging agents to move in the same direction on an edge.
% reduces the number of head-on collisions need to be resolved during planning, which occur when two agents move through the same edge from opposite directions.

% Rewrite: prior GGO works considers the pebble motion agent model. If we additionally consider rotation (in discretized timesteps), only optimizing edge weight might create issue. Three types of actions: move forward, rotate in place, wait. One might consider that wait is the least desirable 
% not optimizing edge direction can create issue because it can increase the number of rotation

% However, prior GGO methods are only tested assuming that agents move with the pebble motion model~\cite{SternSoCS19}, where in each discretized timestep, they can move to their adjacent vertices on $G$ in an omni-directional manner or wait in the current vertex. When agents move with the rotation motion model~\cite{chan2024lorr}, where they can only move forward and rotate in place, the guidance graph with edge weights alone cannot provide sufficient guidance. Under LMAPF settings, MAPF planner can make myopic decisions to traverse through the high-cost edges, making traffic worse 

\begin{figure}
    \centering
    \includegraphics[width=1\linewidth]{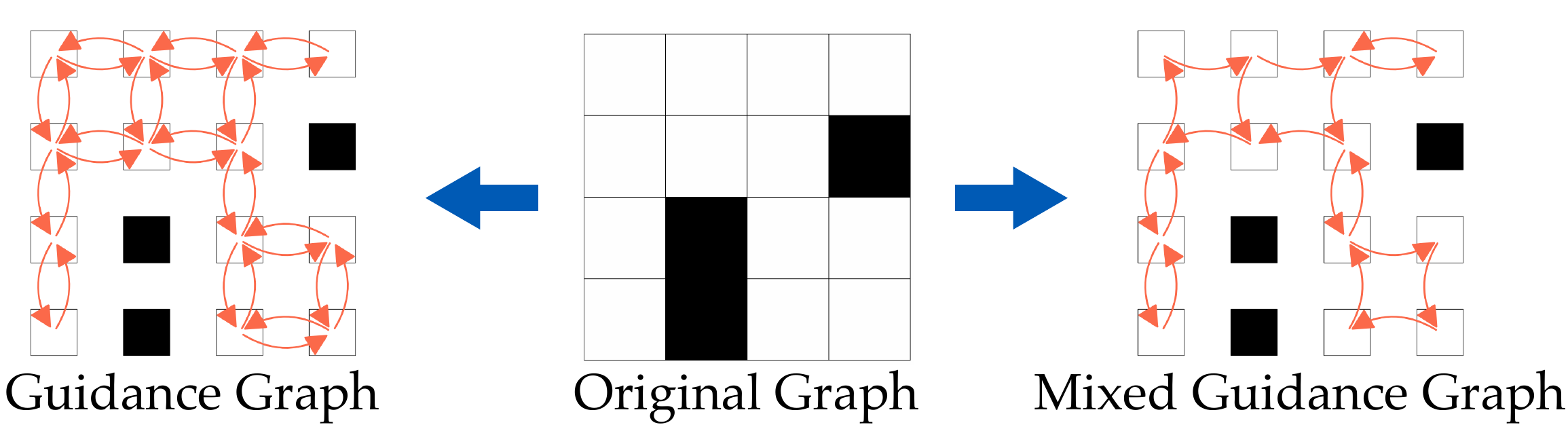}
    \caption{A guidance graph and a mixed guidance graph defined for a grid graph. Edge weights are all 1 for simplicity.}
    \label{fig:front-fig}
\end{figure}

However,
% optimizing edge weights only \emph{reduces} the number of head-on collisions, not \emph{eliminates} them. 
for a high-weight edge, agents can use it if it results in a smaller overall cost, resulting in head-on collisions, which happens when agents travel in opposite directions on the same edge.
Resolving a head-on collision in MAPF is challenging, as it requires at least one of the agents to take a detour.
% This slows agent movement and creates more traffic congestion. 
% In addition, if agents apply the rotational motion model~\cite{chan2024lorr}, where they can only move forward and rotate in place, such negative impacts of head-on collisions can be amplified.
This issue can be amplified if agents move in more realistic motion models. For example, if agents move under the rotation motion model~\cite{Jiang2024Competition}, in which agents can only move forward and rotate in place, taking a detour can take longer than with the pebble motion model~\cite{SternSoCS19}, where agents move omni-directionally, because of the additional rotation actions.
% This issue can be amplified under the rotational motion model
% The rotational motion model can model robots more realistically in real-world applications of LMAPF, such as autonomous warehouses~\cite{yan2025mass} and robotic sorting systems~\cite{KouAAAI20}. Taking detours with the rotational model is slower because it requires additional rotation actions. 
Since the rotational model more accurately reflects the motion of differentiable drives commonly used in real-world LMAPF applications, including autonomous warehouses~\cite{yan2025mass} and robotic sorting systems~\cite{KouAAAI20}, it is desirable to optimize edge directions, determining which edges to block.
% Although the rotational model more accurately reflects the motion of differentiable drives commonly used in real-world LMAPF applications, including autonomous warehouses~\cite{yan2025mass} and robotic sorting systems~\cite{KouAAAI20}, making detours is particularly costly due to additional rotation actions.
% To further reduce the number of head-on collisions, it is desirable to optimize edge directions to determine which edges to block.

Optimizing edge directions is a non-trivial problem. First, we need to guarantee strong connectivity. All pairs of vertices of the graph must be reachable from each other. Second, blocking too many edges might force agents to take unnecessary detours. 
% For example, a strongly connected directed graph can eliminate all head-on collisions between agents because no edges can be traversed from opposite directions. However, it might force agents to take long detours, reducing system throughput.
We need a technique that considers the trade-off between reducing head-on collisions and shortening agents' travel distances. 
% Third, optimizing direction alone does not provide sufficient guidance either. When agents can traverse in opposite directions on an edge, we still need to optimize edge weights to provide more fine-grained guidance to the agents' movement.
Therefore, in this paper, we explore whether the incorporation of edge direction optimization into GGO can help improve throughput. We generalize GGO to Mixed Guidance Graph Optimization (MGGO), which optimizes both edge weights and directions, and present two MGGO methods. \Cref{fig:front-fig} shows an example of a mixed guidance graph compared to a guidance graph.
% . First, we propose a two-phase method that optimizes edge directions and edge weights separately. In phase one, we use an Evolutionary Algorithm (EA) to optimize the edge directions directly. 
% \Cref{fig:random-32-32-20-gd} shows an optimized guidance graph with most of the edges unidirectional. 
To ensure strong connectivity, we introduce a simple algorithm to repair the graph. 
% In phase two, we apply GGO~\cite{zhang2024ggo} to optimize the edge weights. 
% Second, we propose a co-optimization technique based on Quality Diversity (QD) algorithms~\cite{justin_qd_2016} that jointly optimizes the direction and weights of the edges. 
In addition, we incorporate edge-direction-related information into the optimization process of a GGO method, making it optimize \emph{edge-direction-aware} guidance graphs.
% We also incorporate edge directions into a prior GGO method, making it optimize edge-direction-aware guidance graphs.

We make the following contributions: (1) generalizing the guidance graph to the mixed guidance graph as a more general representation of the guidance for LMAPF, (2) proposing two MGGO methods, showing their superior performance compared to state-of-the-art GGO methods, and (3) providing insights into the trade-off of edge weights and edge directions as different means of representing guidance for LMAPF.

\section{Background}

\subsection{Lifelong Multi-Agent Path Finding}

Since the rotational motion model~\cite{Jiang2024Competition} can model robots of real-world applications of LMAPF more realistically than the pebble motion model~\cite{SternSoCS19}, this paper focuses primarily on the MAPF and LMAPF problems defined on the rotational motion model.

\begin{definition}[MAPF]
Given a four-connected 2D grid graph $G(V,E)$ and $k$ agents with their start and goal vertices in $G$, MAPF aims to find collision-free paths from the start to the goal vertices. The state of each agent is determined by its current vertex $v \in V$ and its orientation $o \in \{ \text{North, South, East, West}\}$. In each discretized timestep, agents can move forward to an adjacent vertex, rotate 90$^\circ$, or wait. Two agents collide if they arrive at the same vertex or move through the same edge from opposite directions at the same timestep. MAPF minimizes the sum-of-cost, the sum of path length of all agents.
\end{definition}

\begin{definition} [LMAPF]
Lifelong MAPF (LMAPF) is a variant of MAPF in which agents are continuously assigned new goal vertices when they arrive at the current ones. LMAPF maximizes throughput, the average number of goal vertices reached by all agents at each timestep.
\end{definition}

% \subsubsection{Lifelong MAPF Algorithms}

% Solving MAPF optimally is known to be NP-hard~\cite{YuAAAI13}, and LMAPF poses a greater challenge, as it is usually solved by decomposing it into a series of MAPF problem instances and solving each in sequence. 
To solve LMAPF, \citet{SilverAIIDE05} proposes WHCA$^*$, which uses priority planning~\cite{Erdmann1987PP} to find collision-free paths for all agents within a pre-defined planning window. RHCR~\cite{Li2020LifelongMP} advances this idea to using any MAPF algorithm to search for collision-free paths within a planning window $w$ every $h$ timesteps ($w \geq h)$. 
% Although replanning all agents can provide better solutions, it is computationally expensive. Therefore, some works propose only replanning agents that have reached their goals at every timestep while avoiding agents that are still moving towards their goals~\cite{CapICAPS15,MaAAMAS17,LiuAAMAS19}.
Instead of explicitly searching for collision-free paths, rule-based algorithms~\citet{WangICAPS08} move agents step by step by following a set of pre-defined rules. PIBT~\cite{okumura2019priority} is a ground-breaking work in this category that moves by selecting their most preferred actions in each step. If collisions occur, a pre-defined rule based on priority inheritance and backtracking is applied to resolve the collisions. \citet{YukhnevichEPIBT2025} extends PIBT by planning multiple steps.
To show that our proposed MGGO methods work for different types of LMAPF algorithms, we select RHCR~\cite{Li2020LifelongMP} and PIBT~\cite{okumura2019priority,Jiang2024Competition} to conduct the experiments.

\subsection{Edge Weights as Guidance} \label{sec:ggo}

% Many prior works have represented guidance as edge weights in LMAPF~\cite{Jansen2008DirectionMF,liron_highway16,Yu2023,ChenAAAI24}.

Many prior works have represented guidance as edge weights in MAPF. However, they mostly manually design or use handcrafted equations to generate the edge weights. \citet{Cohen2015FeasibilitySU} proposes the concept of highways, which selects a set of preferred edges as highways and assigns a lower edge weight to them, encouraging agents to traverse the highways. \citet{liron_highway16} generates highways based on handcrafted equations, but does not outperform Crisscross~\cite{lironPhDthesis,li2023study}, which manually designs highways by alternating directions in even and odd rows and columns. Meanwhile, \citet{Jansen2008DirectionMF} assigns a direction vector to every vertex in $G$, and computes a movement cost based on past traffic using a handcrafted equation.
More recent works applies a similar idea to collect traffic patterns and use handcrafted equations to compute edge weights~\cite{han2022spaceutil,Yu2023,ChenAAAI24}.
Attempting to unify the prior works, \citet{zhang2024ggo} propose the guidance graph. 

\paragraph{Notation.}
For a pair of vertices $u$ and $v$, we use $\{u,v\}$ and $(u,v)$ to denote undirected and directed edges, respectively. The connection between $u$ and $v$ is bidirected if both $(u,v)$ and $(v,u)$ exist. The connection is unidirectional if exactly one of $(u,v)$ and $(v,u)$ exists. 

\begin{definition} [Guidance Graph]
    Given a connected graph $G(V, E)$ for LMAPF, a guidance graph $G_g(V, E_g, \omegaV)$ is a bidirected weighted graph where each directed edge $(u,v) \in E_g$, $u,v \in V$, corresponds to an action that agents can take in $G$ to move from $u$ to $v$. $E_g = E_{wr} \cup E_{move}$, where $E_{wr}$ contains self-loops where $u = v$, representing waiting or rotating. $E_{move}$ contains edges with $u \neq v$, representing moving forward. 
    % $\forall \{u,v\} \in E$, $E_{move}$ contains both $(u,v)$ and $(v,u)$.
    All edge weights are represented as a vector $\omegaV \in \mathbb{R}^{|E_g|}_{>0}$.
\end{definition}

\begin{definition} [Guidance Graph Optimization]
    Given a graph $G(V, E)$, an objective function $f: \mathbb{R}^{|E_g|} \rightarrow \mathbb{R}$, and lower and upper bounds $\omegaVlb$ and $\omegaVub$ ($0 < \omegaVlb \le \omegaVub$) for edge weights, GGO searches for the best guidance graph $G_g^*(V_g,E_g,\boldsymbol{\omega^*})$ with
    \begin{equation}
        \boldsymbol{\omega^*} = \argmax_{\omegaVlb \leq \omegaV \leq \omegaVub} f(\omegaV)
    \end{equation}
\end{definition}

% Note that an edge weight of positive infinity in a guidance graph does not prohibit agents from using it.

\citet{zhang2024ggo} propose two GGO methods. They first propose GGO Direct Search (GGO-DS), which uses the Covariance Matrix Adaptation Evolutionary Strategy (CMA-ES)~\cite{hansen2016cmaes}, a black-box optimizer, to directly search for all edge weights with the objective of maximizing throughput. 
CMA-ES maintains a multi-variate Gaussian distribution to model the edge weights. Starting from a standard normal Gaussian, GGO-DS iteratively samples new batches of edge weights from the Gaussian. 
To ensure that the edge weights satisfies the pre-defined range, GGO-DS uses min-max normalization.
After evaluation, the mean and covariance matrix of the Gaussian is updated to the high-throughput region of the search space.
However, CMA-ES scales poorly as the number of edges in the guidance graph increases. Therefore, they propose GGO Parameterized Iterative Update (GGO-PIU), which uses a neural-network-based \emph{update model} to generate edge weights in an iterative manner. Starting from an unweighted guidance graph where all edge weights are 1, PIU runs LMAPF simulations to collect the frequency of waiting at each vertex and uses the update model to update edge weights. PIU runs for $N_p$ iterations. The model is optimized by CMA-ES to maximize throughput.

GGO-PIU has several weaknesses.
% First, it is only tested with the pebble motion model, ignoring rotation actions of the agents. 
First, it requires running at least one LMAPF simulation in each iteration, which is computationally expensive. For example, PIU does not work for RHCR~\cite{zhang2024ggo}.
% Third, the rationale for having a iterative update procedure is to allow the model to generate a higher-quality guidance graph in each iteration. However, an iterative procedure also makes it more difficult to optimize the update model because it has to understand how to update a variety of guidance graphs. In fact, if we provide a better initial guidance graph to the update model, the model can also generate better guidance graphs.
Second, the edge weights of the initial guidance graph are 1. Intuitively, a better initial guidance graph can provide more hints to the update model, making it easier for the model to generate better guidance graphs. To tackle these weaknesses, we propose to use Parameterized Update (PU), a variant of PIU with $N_p = 1$. 
% We justify the advantage of PU in \Cref{sec:exp}.

\subsection{Edge Directions as Guidance}

Although many works have studied edge weights as guidance, similar studies on edge directions are relatively limited. In MAPF, 
% while it is beneficial to design a graph with unidirectional edges to reduce head-on collisions, 
no prior works have studied edge directions optimization. The directed Crisscross pattern~\cite{WangICAPS08,li2023study} is a common manually designed pattern of edge directions.
% Crisscross~\cite{lironPhDthesis} selects a set of preferred edges by alternating edge directions in even and odd rows and columns of a grid graph. It assigns a smaller edge weight to the preferred edges than the rest. 
It is a variant of Crisscross, in which only the preferred edges are kept in the graph. However, directed Crisscross can only be generated from biconnected graphs, which are undirected graphs that do not have bridges. A bridge is an undirected edge in an undirected graph whose removal disconnects the graph. 
% In addition, no previous work in MAPF has explored how to design effective edge directions.

In urban traffic network design and navigation of multiple Automated Guided Vehicles, a similar problem of Strong Network Orientation Problem (SNOP)~\cite{duhamel_strong_2024} is studied. Given an undirected graph, SNOP determines a direction for each undirected edge in the graph so that the resulting directed graph is strongly connected and the total distance between each pair of vertices is optimized. \citet{robbins_theorem_1939} has established that directions can be assigned to the edges of a biconnected graph to make it strongly connected. 
As an NP-hard problem~\cite{burkard_minimum-cost_1999}, existing works solve SNOP by using Mixed Integer Linear Programming (MILP)~\cite{gaskins_flow_1987,duhamel_strong_2024}, Branch-and-Bound~\cite{kaspi_optimal_1990}, and neighborhood search~\cite{gallo_meta-heuristic_2010}. 
However, it is non-trivial to adapt these methods to optimize edge directions for MAPF. First, they do not scale to large graphs with thousands of vertices. Second, they do not optimize throughput. Third, they enforce strong connectivity using MILP~\cite{duhamel_strong_2024}, which scales poorly, or naive rejection sampling~\cite{gallo_meta-heuristic_2010}, which is not sample efficient.

% \citet{robbins_theorem_1939} has established that directions can be assigned to the edges of an undirected graph to make it strongly connected if and only if the graph does not have bridges. A bridge is an edge whose removal disconnects the graph. An undirected graph with no bridges is considered biconnected. \citet{birkhoff_graph_1978} has shown that running the Depth First Search (DFS) can generate a strongly connected directed graph on a biconnected graph.

\subsection{Quality Diversity Algorithm}

% Although single-objective optimizers, such as CMA-ES, have been used to optimize black-box objectives, we are also interested in 
QD algorithms optimize an objective and diversify a set of diversity measures simultaneously. A measure space is defined based on the given diversity measures and it is evenly discretized into cells, referred to as an \emph{archive}. 
% Each cell corresponds to a range of values of diversity measures. 
QD attempts to find the best solution in each cell, optimizing the QD-score, defined as the sum of objective values of all solutions in the archive. Covariance Matrix Adaptation MAP-Annealing (CMA-MAE)~\cite{Fontaine2022CovarianceMA} is the state-of-the-art QD algorithm for continuous search spaces. Similar to CMA-ES, it maintains, samples, and updates a multi-variate Gaussian distribution to optimize QD-score. 
Although QD methods are used primarily to generate diverse high-quality solutions~\cite{fontaine2020illuminating,Bhatt2022DeepSA}, prior work has shown both theoretically~\cite{Qian2024qdhelpful} and empirically~\cite{ZhangNCA2023} that diversity measures can assist objective optimization.

\section{Problem Definition}

\begin{definition}[Mixed Guidance Graph]
    Given a connected undirected graph $G(V,E)$, a mixed guidance graph $G_{mg}(V,E_{mg},\omegaV_{mg})$ is a directed weighted graph. For each non-bridge undirected edge $\{u,v\} \in E$, we add $(u,v)$ or $(v,u)$ or both $(u,v)$ and $(v,u)$ to $E_{mg}$. For each bridge $\{u,v\}$, we add both $(u,v)$ and $(v,u)$ to $E_{mg}$. $\forall v \in V$, we also add $(v,v)$ to $E_{mg}$. All edge weights are represented as a vector $\omegaV_{mg} \in \mathbb{R}^{|E_{mg}|}_{>0}$.
\end{definition}

% Intuitively, a mixed guidance graph is a guidance graph in which a set of unidirectional edges $E_{move} \setminus E_{mg\_move}$ are removed and their corresponding actions are prohibited in planning. Agents can still wait or rotate on all vertices.
To plan with a mixed guidance graph without compromising feasibility, we change the objective of MAPF from sum-of-costs to sum of action costs on the mixed guidance graph $G_{mg}$, similar to \citet{zhang2024ggo}.

\begin{definition} [Mixed Guidance Graph Optimization]
Given a connected unweighted graph $G(V, E)$, an objective function $f: \mathcal{G}_{mg} \rightarrow \mathbb{R}$, where $\mathcal{G}_{mg}$ is the space of all mixed guidance graphs, and lower and upper bounds $\omegaVlb$ and $\omegaVub$ ($0 < \omegaVlb \le \omegaVub$) for edge weights, the Mixed Guidance Graph Optimization (MGGO) problem searches for the optimal mixed guidance graph $G^*_{mg}(V, E_{mg}^*, \boldsymbol{\omega^*_{mg}})$ with:
\begin{align}
    &G^*_{mg} = \argmax_{G_{mg} \in \mathcal{G}_{mg}} f(G_{mg}) \\ 
    \textrm{s.t.} \quad 
    &G_{mg}\textrm{ is strongly connected}\\
    &\omegaVlb \leq \boldsymbol{\omega_{mg}} \leq \omegaVub
\end{align}
% $E_g = E_{wr} \cup E_{move}$ contains directed edges that agents can traverse derived from $E$.
\end{definition}

% In this paper, our objective function $f$ is a LMAPF simulator that evaluate a given mixed guidance graph by running a given LMAPF algorithm and returning the throughput.

% \section{From PIU to PU}

\section{Approach}

% We first present a two-phase MGGO method that searches for a mixed guidance graph by optimizing edge weights and edge directions in two separate phases. We then present a joint optimization method that optimizes both simultaneously. Finally, we present an enhanced edge-direction-aware GGO-PU to optimize the edge weights. 
We first introduce a simple algorithm to repair a given graph for strong connectivity.
We then present three approaches to incorporate edge directions.
The first two are MGGO methods: one searches for a mixed guidance graph by optimizing edge weights and directions in two separate phases, and the other optimizes both simultaneously.
The third approach only optimizes edge weights via GGO-PU, but is guided by information from manually designed mixed guidance graphs with unidirectional edges.

\subsection{Edge Reversal Search}

\begin{algorithm}[!t]
\caption{Edge Reversal Search (\ERS)}\label{alg:edge-dir-flip-search}
\LinesNumbered
\small
\SetAlgoNlRelativeSize{0}

\SetKwInput{KwInput}{Input}
\SetKwInput{KwOutput}{Output}
\SetKwProg{Fn}{Function}{:}{}

% Helper functions
% \SetKwFunction{isStronglyConnected}{is\_strongly\_connected}
\SetKwFunction{TarjanSCC}{tarjanAlgo}
\SetKwFunction{BuildCond}{buildCondenseGraph}
\SetKwFunction{OutCut}{outgoingEdges}
\SetKwFunction{Reverse}{ReverseEdge}

\KwInput{$G_{in}=(V_{in},E_{in})$: a mixed guidance graph.\\
$E_b$: bridge edges.\\
\TarjanSCC: Function that runs Tarjan algorithm and return the SCCs and the number of SCCs.\\
\BuildCond: Function to build a condensation graph on the SCCs of $G_{in}$.\\
\OutCut: Function to find the outgoing edges of vertices in $G_{in}$ in the SCC corresponding to a meta vertex.\\
\Reverse: Function to reverse the edge direction.
}
\KwOutput{A strongly connected mixed graph obtained by reversing edge directions.}

% $G_{in}(V',E') \gets G(V,E)$\nllabel{efs:create-G'}\;

\If(\nllabel{efs:if-no-bridge}){$E_b \not\subseteq E_{in}$}{
    $E_{in} \gets E_{in} \cup \bigcup_{\{u,v\}\in E_b} \{(u,v),(v,u)\}$\nllabel{efs:add-bridge}\;
}

$\mathsf{scc},C \gets \TarjanSCC(G_{in})$\nllabel{efs:init-tarjan}\;

\While(\nllabel{efs:while}){$C > 1$}{
    $\mathcal{D(V,E)} \gets \BuildCond(G_{in}, \mathsf{scc}, C)$\nllabel{efs:build-metag}\;
    
    % \If{$\mathcal{S}=\emptyset$}{
    %     \Return{\textbf{fail}} \tcp*[f]{should not happen: condensation is a DAG}
    % }
    $\eta \gets$ a meta vertex  s.t. $\eta \in \mathcal{V} \And deg^-(\eta) = 0$\nllabel{efs:find-source}\;

    $E_{\text{out}} \gets \OutCut(G_{in}, \eta)$\nllabel{efs:find-out-edge}\;
    
    % \If{$|E_{\text{out}}| \le 1$}{
    %     \Return{\textbf{fail}} \tcp*[f]{no sufficient cut edges to flip}
    % }

    $E_{rev} \gets$ sample $\floor{\frac{|E_{\text{out}}|}{2}}$ distinct edges from $E_{\text{out}}$\nllabel{efs:sample-rev-edge}\;
    
    \For(\nllabel{efs:rev-edge-for}){$(u,v) \in E_{rev}$}{
        \Reverse($G_{in}$,$(u,v)$)\nllabel{efs:rev-edge}\;
        % \tcp{flip direction while preserving weight}
        % $w \gets \mathrm{weight}(u,v)$\;
        % remove $(u,v)$ from $E'$\;
        % add $(v,u)$ to $E'$ with weight $w$\;
    }

    $\mathsf{scc},C \gets \TarjanSCC(G_{in})$\nllabel{efs:end-while-tarjan}\;
}
\Return $G_{in}$\nllabel{efs:return}\;

% \Fn{\BuildCond($G_{in}$,$\mathsf{scc}$,$C$)}{
%     $\mathcal{V} \gets \{0,1,...,C-1\}$, $\mathcal{E} \gets \{\}$\nllabel{efs:init-meta-ve}\;
%     % Initialize $\mathcal{D(V,E)}$\;
%     \For(\nllabel{efs:meta-add-e-for}){$(u,v)\in E_{in}$}{
%         \If(\nllabel{efs:meta-meta-add-e-check}){$\mathsf{scc}(u) \neq \mathsf{scc}(v) \And (\mathsf{scc}(u),\mathsf{scc}(v)) \not\in \mathcal{E})$}{
%             $\mathcal{E} \gets \mathcal{E} \cup {(\mathsf{scc}(u),\mathsf{scc}(v))}$\nllabel{efs:meta-add-e}\;
%         }
%     }
%     \Return $\mathcal{D(V,E)}$\;
% }

\end{algorithm}

We propose \textbf{E}dge \textbf{R}eversal \textbf{S}earch (\ERS), a greedy algorithm to repair a given mixed guidance graph for strong connectivity by reversing edge directions. \Cref{alg:edge-dir-flip-search} shows the pseudocode. Given an input mixed graph $G_{in}(V_{in},E_{in})$ and all bridges $E_b$ in the original undirected graph $G$ corresponding to $G_{in}$, we first add $(u,v)$ and $(v,u)$ for all bridges $\{u,v\} \in E_b$ to $E_{in}$ (\cref{efs:if-no-bridge,efs:add-bridge}) because these edges must be present to ensure strong connectivity. Bridges can be easily found in $G$ by using a linear-time algorithm based on DFS~\cite{tarjan_note_1974}.
% To enforce strong connectivity after initialization and mutation, it is necessary to keep bridges as bidirected edges. Therefore, we add bridges to the initialized and mutated mixed guidance graph, and bridges are not involved in the optimization process. We then apply \ERS to repair the graphs.
We then run Tarjan algorithm~\cite{tarjan1972scc} to find all strongly connected components (SCC) of $G_{in}$ (\cref{efs:init-tarjan}) and $C$, the number of SCCs. 
% We store the SCCs as a hashmap $\mathsf{scc}$ from the vertices in $V_{in}$ to unique indices of the SCCs to which the vertices belong. 
If there is more than 1 SCC (\cref{efs:while}), then $G_{in}$ is not strongly connected, and we perform an iterative edge reversal procedure. We build a condensation graph $\mathcal{D(V,E)}$ (\cref{efs:build-metag}) following \citet{cormen2009introalgo} by contracting each SCC into a vertex in $\mathcal{V}$.
% where each vertex in $\mathcal{V}$ corresponds to an SCC of $G_{in}$ (\cref{efs:init-meta-ve}) 
% For each edge $(u,v) \in E_{in}$, we add a meta edge $(\mathsf{scc}(u),\mathsf{scc}(v))$ if $u$ and $v$ are in different SCCs in $G_{in}$ and the meta edge is not yet added (\cref{efs:meta-add-e-for,efs:meta-meta-add-e-check,efs:meta-add-e}). 
Since $\mathcal{D}$ is a DAG~\cite{cormen2009introalgo}, we can find a source meta vertex $\eta \in \mathcal{V}$ s.t. $\eta$ has no incoming edges (\cref{efs:find-source}). We then find all outgoing edges $E_{out} \subseteq E_{in}$ of the vertices in the SCC of $G_{in}$ corresponding to $\eta$ (\cref{efs:find-out-edge}) and sample half of them as $E_{rev}$ (\cref{efs:sample-rev-edge}). If $|E_{out}| = 1$, there are no edges in $|E_{rev}|$. We prove that $|E_{out}| \geq 2$ and thus $|E_{rev}| \geq 1$ in \Cref{appen:efs-proof}. We reverse the direction of the edges in $E_{rev}$ while maintaining their edge weights (\cref{efs:rev-edge-for,efs:rev-edge}) and run the Tarjan algorithm to check for updated SCCs (\cref{efs:end-while-tarjan}). We return $G_{in}$ when it has exactly 1 SCC (\cref{efs:return}).
Intuitively, reversing directions of the edges in $E_{rev}$ allows $\eta$ to have incoming edges in $\mathcal{D}$, potentially reducing the number of SCCs in $G_{in}$. Although \ERS is not a complete algorithm, it runs very fast with each iteration taking $O(|E_{in}| + |V_{in}|)$. Empirically, \ERS always returns a valid solution. We provide experimental evaluations of \ERS compared to an optimal MILP solver~\cite{duhamel_strong_2024} in \Cref{appen:efs-exp}.
% We prove the completeness of \ERS in \Cref{appen:efs-proof}.

% we enter a while loop that checks if $G_{in}$ is not strongly connected (\cref{efs:while}), which can be done by running DFS twice~\cite{}

\subsection{Two-Phase MGGO-DS} \label{subsec:two-phase-mggo}

% We propose a two-phase MGGO Direct Search (MGGO-DS) method, where we sequentially optimize edge directions and edge weights in two phases.

\subsubsection{Phase One: Optimizing Edge Directions}
In phase one, we optimize the edge directions of the mixed guidance graph by assigning a direction to \emph{all} non-bridge edges in $G(V,E)$.
% Concretely, we solve a variant of the MGGO problem in which the edge weights $\omegaV_{mg} = \boldsymbol{1}^{|E_{mg}|}$.
We define a binary decision variable for each undirected non-bridge edge $\{u, v\} \in E$. The binary variable determines whether $\{u,v\}$ is changed to $(u,v)$ or $(v,u)$. We do not include the option to make $\{u,v\}$ bidirected in the decision variable because an unweighted bidirected edge does not provide any guidance and, therefore, would be abandoned by the optimizer. 
% We justify our choice of design in \Cref{sec:exp}.

We optimize the edge directions using $(1+\lambda)$ Evolutionary Algorithm (EA)~\cite{thomas1997EA}. \Cref{fig:two-phase-mggo} shows the overview. After initializing a batch of $\lambda$ mixed guidance graphs, we repair the graphs for strong connectivity and evaluate each of them by running LMAPF simulations $N_e$ times for $T$ timesteps. We select the one with the highest throughput and mutate it, generating $\lambda$ new graphs. We then repair them with \ERS and evaluate the mutated graphs in the simulator. We stop when the number of evaluations reaches $N_{eval}$.

\begin{figure}[t]
    \centering
    \includegraphics[width=1\linewidth]{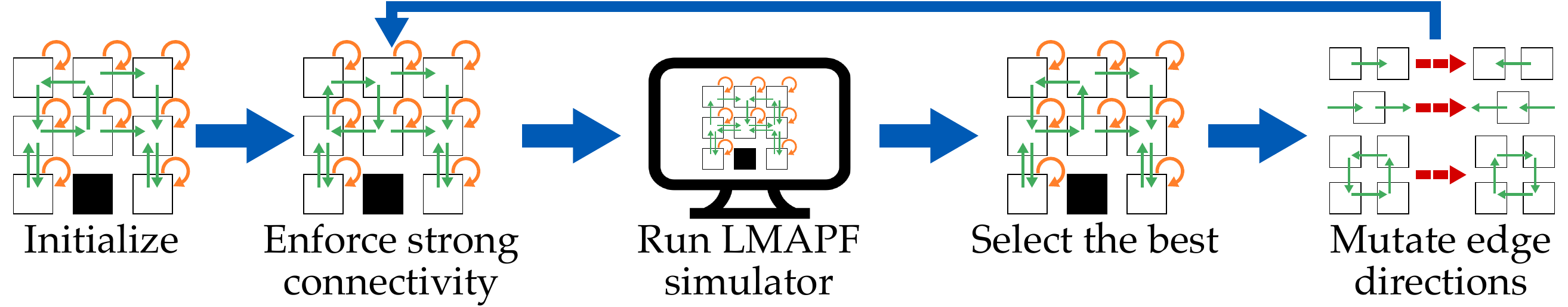}
    \caption{Optimize edge directions in phase one of MGGO-DS.}
    \label{fig:two-phase-mggo}
\end{figure}

\paragraph{Initialization.} We use three initialization methods. (1) \emph{directed Crisscross}: the original Crisscross alternates directions every 1 column or row. We generalize it to alternate directions every $p$ rows and columns $\forall p \in \{1,2,\cdots,\floor{\frac{\min(H,W)}{2}}\}$, where $H$ and $W$ are the height and width of the grid graph, respectively. (2) \emph{DFS}~\cite{birkhoff_graph_1978}: we use DFS to generate different strongly connected mixed guidance graphs. (3) \emph{random}: we generate random directed graphs. With a batch size $\lambda$, we perform each initialization on one-third of the batch.
% we randomly assign a value to each decision variable, deciding a direction for each pair of adjacent vertices.
% For directed Crisscross and random generated graphs, we repair them for strongly connectivity by using \ERS.

\paragraph{Mutation.} 
% To mutate the edge directions of a parent mixed guidance graph, we first construct the parent mixed guidance graph corresponding to the decision variables. Then, 
Inspired by previous works~\cite{gallo_meta-heuristic_2010,fontaine:gecco19}, we use three mutation operators. (1) \emph{$k$-edges}: we randomly sample $k$ distinct \emph{edges} and change their directions. (2) \emph{$k$-vertices}: we randomly sample $k$ distinct \emph{vertices} and change the directions of all their incoming and outgoing edges. In these two operators, $k$ is sampled from a geometric distribution $P(X=k) = (1-p)^{k-1}p$ with $p=\frac{1}{2}$. (3) \emph{random-cycle}: we first select a random vertex and perform a DFS from it to find a random cycle. Since the graph before mutation is guaranteed to be strongly connected, at least one cycle exists in the graph. We then change the directions of all the edges in the cycle. With a batch size $\lambda$, we perform each mutation on one-third of the batch.
% $\floor{\frac{\lambda}{3}}$ $k$-edges mutation, $\floor{\frac{\lambda}{3}}$ $k$-vertices mutations, and $\lambda - 2 \cdot \floor{\frac{\lambda}{3}}$ random-cycle mutations. 

\subsubsection{Phase Two: Optimizing Edge Weights}

After optimizing the edge directions in phase one, we obtain an unweighted mixed guidance graph $G_{mg}(V,E_{mg}, \boldsymbol{1}^{|E_{mg}|})$. We then use GGO-DS~\cite{zhang2024ggo} to optimize the edge weights, forming $G_{mg}(V,E_{mg}, \omegaV_{mg})$.

\subsection{QD-based Joint MGGO-PU} \label{sec:qd-mggo}

\begin{figure}[t]
    \centering
    \includegraphics[width=1\linewidth]{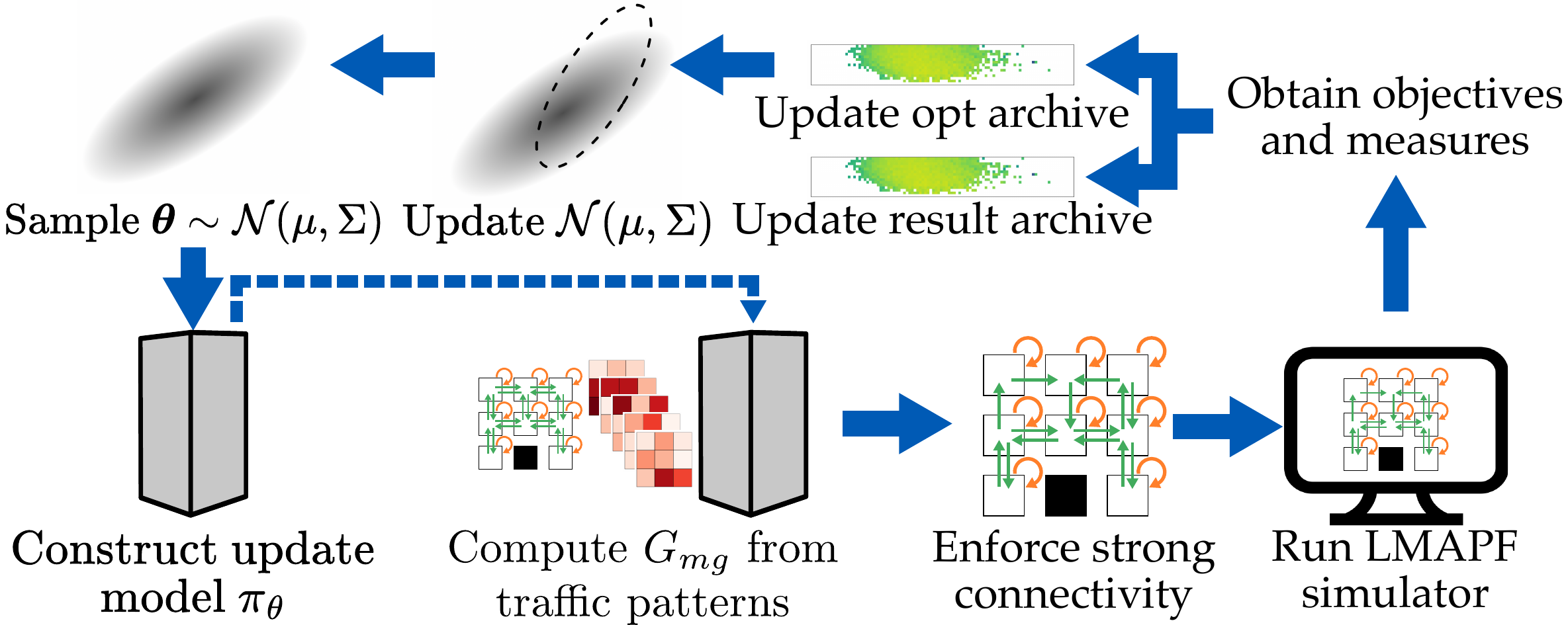}
    \caption{QD-based Joint MGGO-PU.}
    \label{fig:mggo-pu}
\end{figure}

Although the two-phase search can optimize both the edge directions and the edge weights, it maximizes the number of unidirectional edges. Agents can take unnecessary detours with too many unidirectional edges. Therefore, we present joint MGGO-PU, where edge directions and edge weights are optimized jointly. \Cref{fig:mggo-pu} shows the overview. 
We use CMA-MAE~\cite{Fontaine2022CovarianceMA} to optimize the parameters of an update model, which generates mixed guidance graphs given precomputed traffic patterns.
We first sample from the multi-variate Gaussian of CMA-MAE for a batch of $\lambda$ parameter vectors, forming update models. Based on precomputed traffic patterns, we use update models to generate the mixed guidance graphs. We then repair the graphs with \ERS, evaluate them in the LMAPF simulator, and compute the objectives and measures. We then add the evaluated guidance graphs and the corresponding update models to an optimization archive and a result archive. Finally, we update the parameters of the Gaussian. We run CMA-MAE until the total number of evaluations reaches $N_{eval}$.

\subsubsection{Precomputed Traffic Patterns}
To obtain the traffic pattern, we run $N_{obs}$ LMAPF simulations in an unweighted guidance graph and a directed Crisscross mixed guidance graph generated based on directed Crisscross highways~\cite{li2023study}. We run \ERS on the generated directed Crisscross to ensure strong connectivity.
% For each simulation, we collect a tensor of size $[H,W,7]$, where the first four channels are the frequency of agents moving to neighbor vertices in four directions, and the last three channels are the frequency of agents waiting in place, rotating clockwise, and rotating counterclockwise, respectively. 
Each simulation produces a $[H,W,7]$ tensor that encodes the frequencies of move actions in four directions and the frequencies for waiting and clockwise/counterclockwise rotations.
We feed the traffic pattern along with the initial guidance graphs to the update model to generate the resulting mixed guidance graphs.
We use directed Crisscross because it is the state-of-the-art human-designed guidance graph~\cite{zhang2024ggo}, and we make it directed to give the model more information on how edge directions provide guidance.

\subsubsection{Update Model and Graph Representation}
We need to define an update model and represent the edge weights and edge directions of a mixed guidance graph for a grid graph of size $H\times W$.
We follow \citet{zhang2024ggo} to represent the edge weights as a tensor of size $[H,W,5]$, where the five channels represent the edge weights of the four outgoing edges at each vertex and the cost of waiting. We represent the edge directions differently in the input and output. In the input, we use an \emph{independent representation}, where the edge directions are represented as a $[H,W,4]$ tensor. The four channels are binary variables that indicate whether an outgoing edge exists from the corresponding vertices.
% If an edge $(u,v)$ exists in the graph, then the corresponding value in the tensor at vertex $u$ is 1. Otherwise, it is 0. 
However, this cannot be used in the output. For a pair of vertices $u$ and $v$, if the model outputs 0 for both $(u,v)$ and $(v,u)$, there are no edges between $u$ and $v$.
Therefore, we use a \emph{dependent representation} in the output with a tensor of size $[H,W,6]$. The first 3 and the last 3 channels represent a one-hot encoding of the outgoing edges to the east and south, respectively. For a vertex $u$ and its eastern neighbor $v$, for example, the first 3 values indicate the probability of having edge $(u,v)$, $(v,u)$, or both $(u,v)$ and $(v,u)$, respectively. Since CMA-MAE is sensitive to the dimension of the search space, we use the independent representation in the input to keep the update model small and the dependent representation in the output to ensure that the model always gives a valid decision. Therefore, our update model is a CNN that takes in a tensor of size $[H,W,2 \cdot 7 \cdot N_{obs} + 2 \cdot 9]$ and outputs a tensor of size $[H,W,11]$.
% We justify our choice of representation in \Cref{appen:mgg-represent}.

\begin{figure}[t]
    \centering
    \includegraphics[width=1\linewidth]{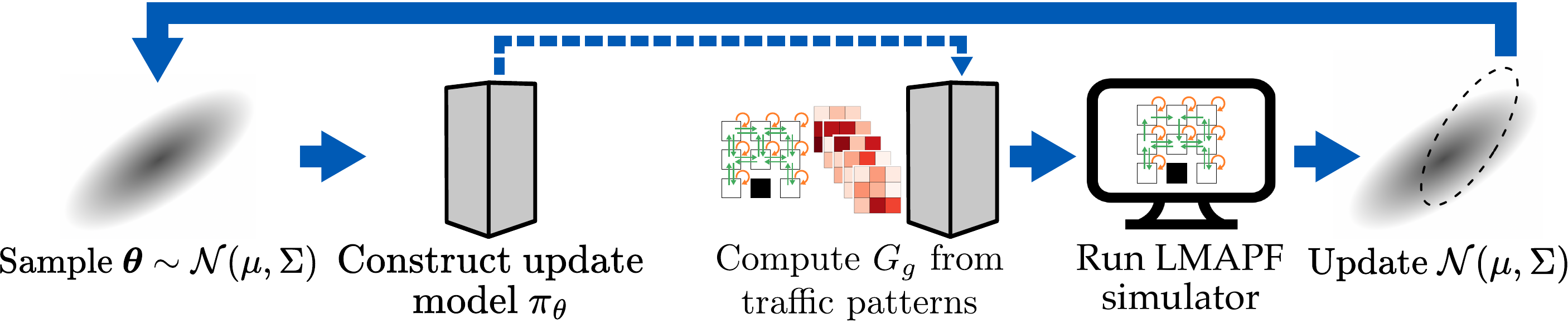}
    \caption{Edge-Direction-Aware GGO-PU.}
    \label{fig:ggo-pu}
\end{figure}

\begin{table}[!t]
    \centering
    \small
    \resizebox{\linewidth}{!}{
    \begin{tabular}{cccrrr}
    \toprule
    Setup & MAPF & Map & $|E_{mg\_max}|$ & $|X|$ & $N_a$\\
    \midrule
    1 & \multirow{3}{*}[-0.5em]{RHCR} & \randomSmall     & 3,359  & 1,258 & 200                     \\
    \cmidrule(lr){3-6}
    2 &                               & \warehouselargeW & 4,074  & 1,563 & 220                     \\    
    \cmidrule(lr){3-6}
    3 &                               & \emptyMid        & 11,328 & 4,512 & 800                     \\
    \midrule
    4 & \multirow{5}{*}[-0.5em]{PIBT} & \randomSmall     & 3,359  & 1,258 & \multirow{2}{*}{400}    \\
    5 &                               & \warehouselargeW & 4,074  & 1,563 &                         \\
    \cmidrule(lr){3-6}
    6 &                               & \emptyMid        & 11,328 & 4,512 & 1,000                   \\
    \cmidrule(lr){3-6}
    7 &                               & \roomLarge       & 14,340 & 5,535 & 1,500                   \\
    \cmidrule(lr){3-6}
    8 &                               & \denSmall        & 11,227 & 4,369 & 1,200                   \\
    \bottomrule
    \end{tabular}
    }
    \caption{
    Summary of the experiment setup.
    % Map Size is the dimension of the grid map, 
    $|E_{mg\_max}|$ is the maximum number of edges the mixed guidance graphs can have, respectively. $|X|$ is the number of decision variables in phase-one of two-phase MGGO-DS.
    $N_a$ is the number of agents.
    }
    \label{tab:exp-setup}
\end{table}

\subsubsection{Objectives and Diversity Measure}
We have two objectives, $f_{opt}$ and $f_{res}$, corresponding to the optimization archive and the result archive. $f_{res}$ runs $N_e$ LMAPF simulations and computes the average throughput, while $f_{opt} = f_{res} + \alpha \cdot \Delta$, computing a weighted sum of $f_{res}$ and an edge direction similarity score $\Delta$. $\Delta$ quantifies the ratio of edges that are reversed by \ERS, and $\alpha$ is a hyperparameter.
We use CMA-MAE to maximize $f_{opt}$ in the optimization archive while keeping the high-throughput guidance graphs in the result archive. We incorporate $\Delta$ to $f_{opt}$ as a regularization term to encourage the model to directly generate strongly connected mixed guidance graphs, reducing the dependency on \ERS. \citet{ZhangNCA2023} used a similar regularization technique to improve the results of QD optimization. 
% We also show in \Cref{sec:exp} that adding $\Delta$ can improve the throughput of optimized mixed guidance graphs.
We use one diversity measure, namely the ratio of unidirectional edges, to control the trade-off between reducing head-on collisions and taking fewer detours in the graph.

\subsection{Edge-Direction-Aware GGO-PU}

We attempt to enhance GGO-PU by making it \emph{edge-direction-aware}, meaning that the update model observes and understands the effectiveness of unidirectional edges in providing guidance and learns to generate guidance graphs. \Cref{fig:ggo-pu} shows the overview of GGO-PU. We use CMA-ES to search for the parameters of an update model capable of generating high-throughput guidance graphs. The key to making GGO-PU edge-direction-aware is using the same traffic pattern from \Cref{sec:qd-mggo} as it includes the traffic pattern observed from a directed Crisscross guidance graph. 
% We show in \Cref{sec:exp} that simply incorporating the new precomputed traffic patterns significantly improves the performance of GGO-PU.

\section{Experimental Evaluation} \label{sec:exp}

% In this section, we compare our proposed MGGO and GGO methods with various baselines. We also show ablation experiments on design choices of our proposed methods.

\subsection{Experimental Setup}

\subsubsection{General Setup}
\Cref{tab:exp-setup} shows the experimental setups. 
Column 2 shows the LMAPF algorithms, namely PIBT~\cite{Jiang2024Competition} and RHCR~\cite{Li2020LifelongMP} with $w=5$, $h=2$, and a runtime limit of 2 seconds for each MAPF run. During optimization, we use $N_e=10$, $T=2,000$ for PIBT and $N_e=5$, $T=1,000$ for RHCR to obtain throughput.
Column 3 shows the maps, including one commonly used warehouse map~\cite{Li2020LifelongMP} and four chosen from the MAPF benchmark~\cite{SternSoCS19}. 
Column 4 shows the maximum number of edges $|E_{mg\_max}|$ that the optimized mixed guidance graph can have. Column 5 shows the number of decision variables in phase one of the two-phase MGGO-DS. Column 6 shows the number of agents $N_a$ used to optimize the mixed guidance graphs.
For the QD-based Joint MGGO-PU, we use $\alpha = 3$. We run all algorithms with $N_{eval} = 20,000$. For the two-phase MGGO-DS, we assign $5,000$ and $15,000$ evaluations to phase one and phase two, respectively. We show the implementation details and compute resources in \Cref{appen:imple-compute}.

\subsubsection{Baselines}

We compare with three baselines: (1) CMA-ES + GGO-DS~\cite{zhang2024ggo}, the state-of-the-art method for GGO, (2) Directed Crisscross~\cite{lironPhDthesis}, the state-of-the-art human-designed guidance graph, and (3) No Guidance, where all edge weights are 1. We do not include GGO-PIU because it does not outperform GGO-DS~\cite{zhang2024ggo}.

\subsubsection{Precomputed Traffic Patterns and Update Models}
We use $N_{obs} = 2$ to collect traffic patterns, resulting in an input of size $[H,W,46]$ to the update model. Our update model is a CNN with 3 convolutional layers with kernel sizes $3\times3$, $1\times1$, and $1\times1$, respectively. All convolutional layers have 8 hidden channels and are followed by a ReLU activation and a batch norm layer, resulting in 3,479 parameters.

\subsection{Experimental Result}

\begin{figure}[!t]
    \centering
    \includegraphics[width=1\linewidth]{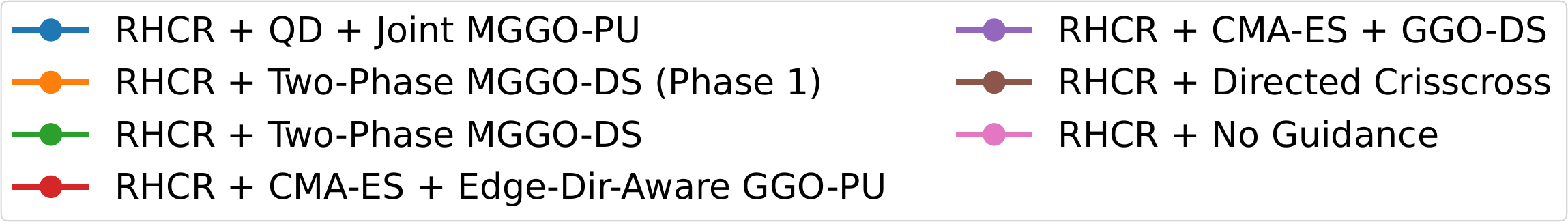}\par\medskip

    % -------- Row 1 --------
    \begin{subfigure}[t]{0.16\textwidth}
        \centering
        \includegraphics[width=\linewidth]{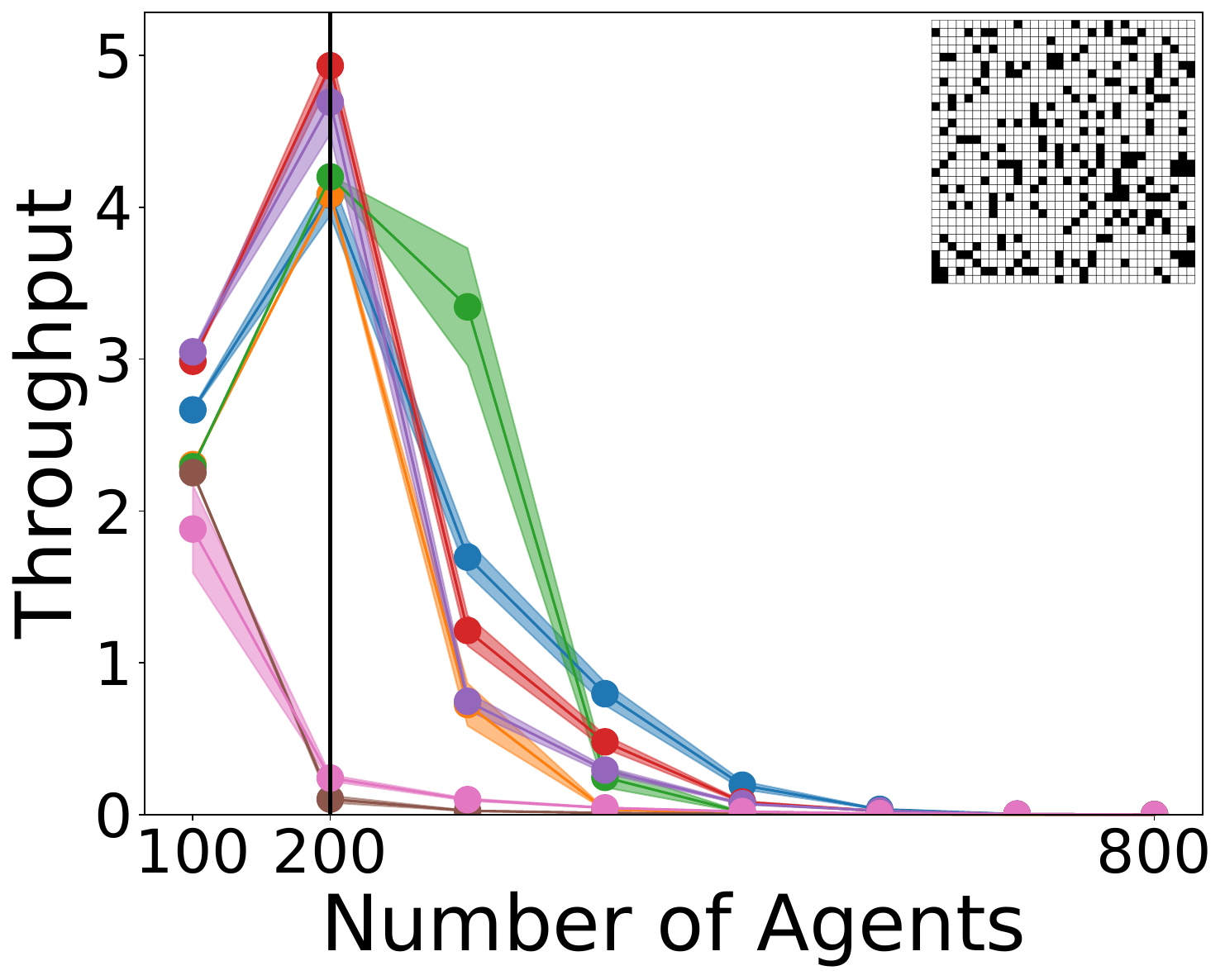}
    \end{subfigure}\hfill
    \begin{subfigure}[t]{0.16\textwidth}
        \centering
        \includegraphics[width=\linewidth]{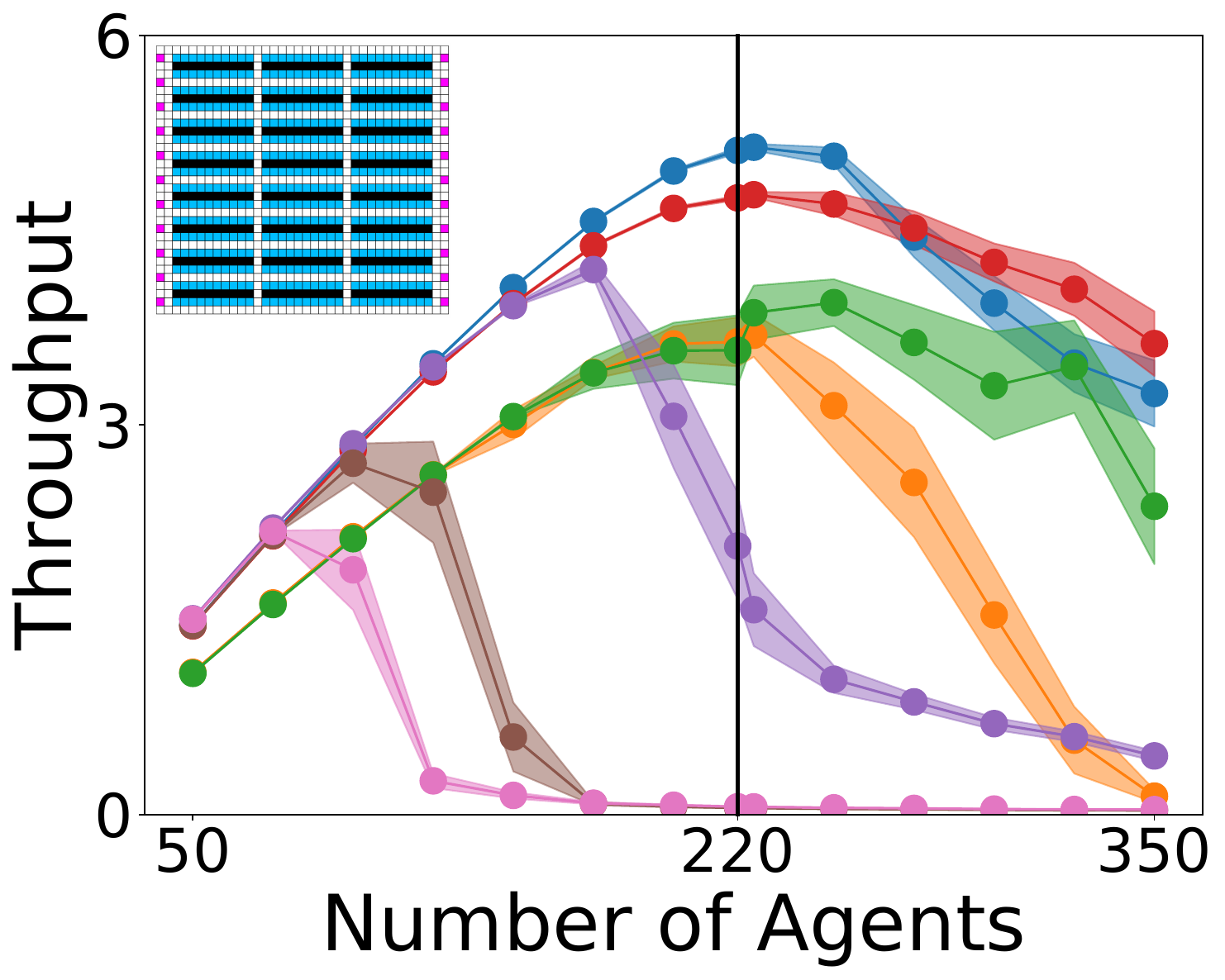}
    \end{subfigure}\hfill
    \begin{subfigure}[t]{0.16\textwidth}
        \centering
        \includegraphics[width=\linewidth]{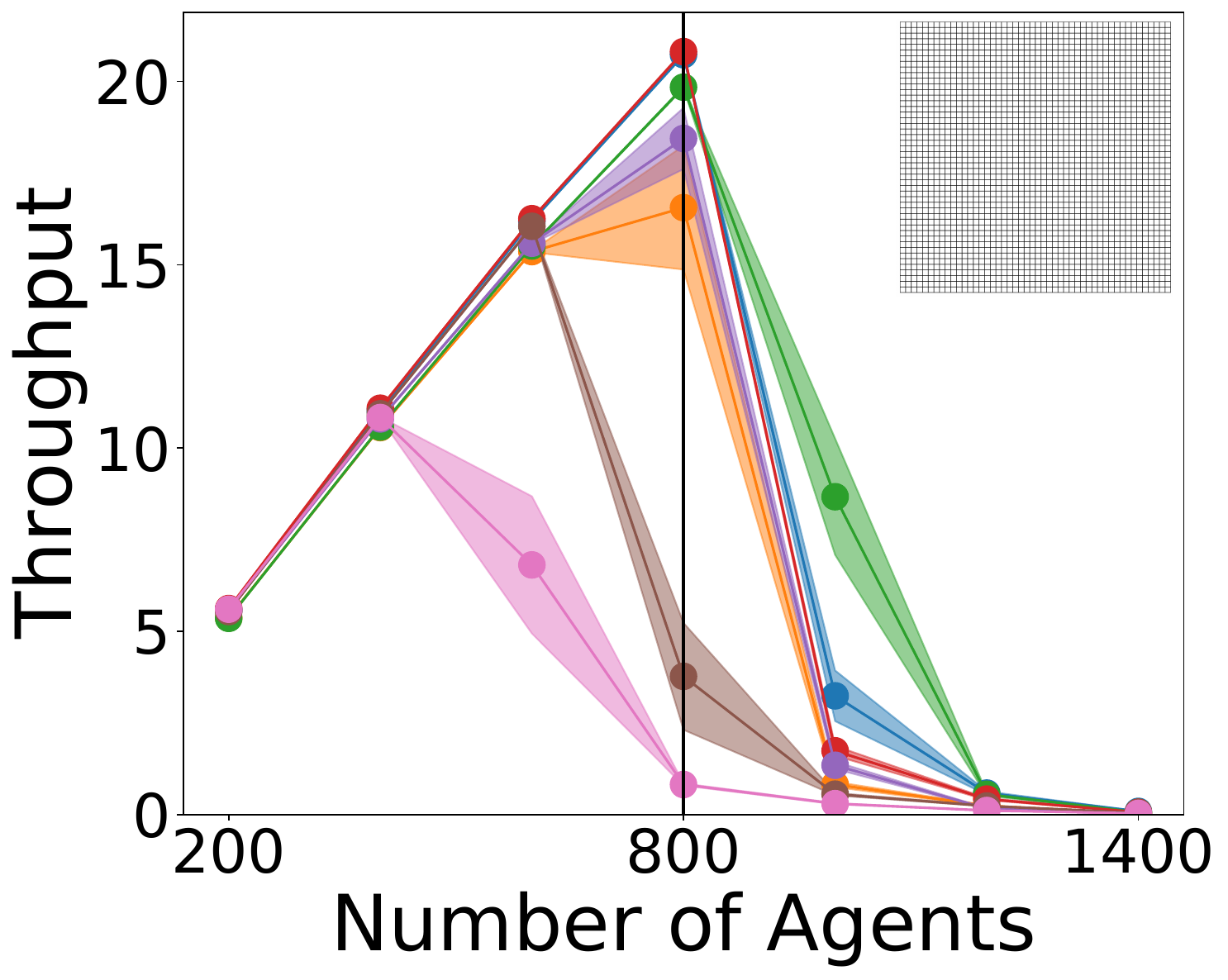}
    \end{subfigure}

    % -------- Row 2 --------
    \begin{subfigure}[t]{0.16\textwidth}
        \centering
        \includegraphics[width=\linewidth]{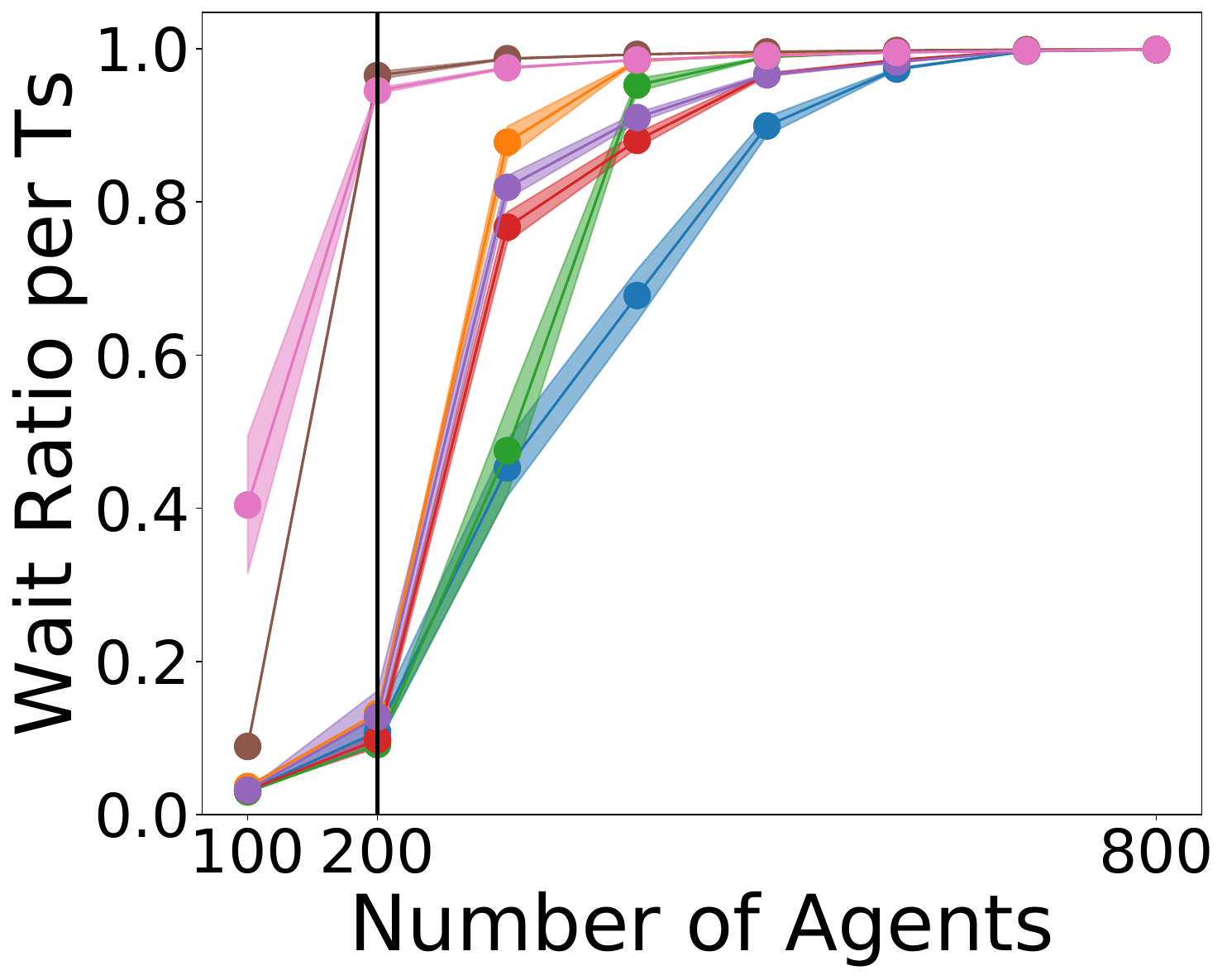}
    \end{subfigure}\hfill
    \begin{subfigure}[t]{0.16\textwidth}
        \centering
        \includegraphics[width=\linewidth]{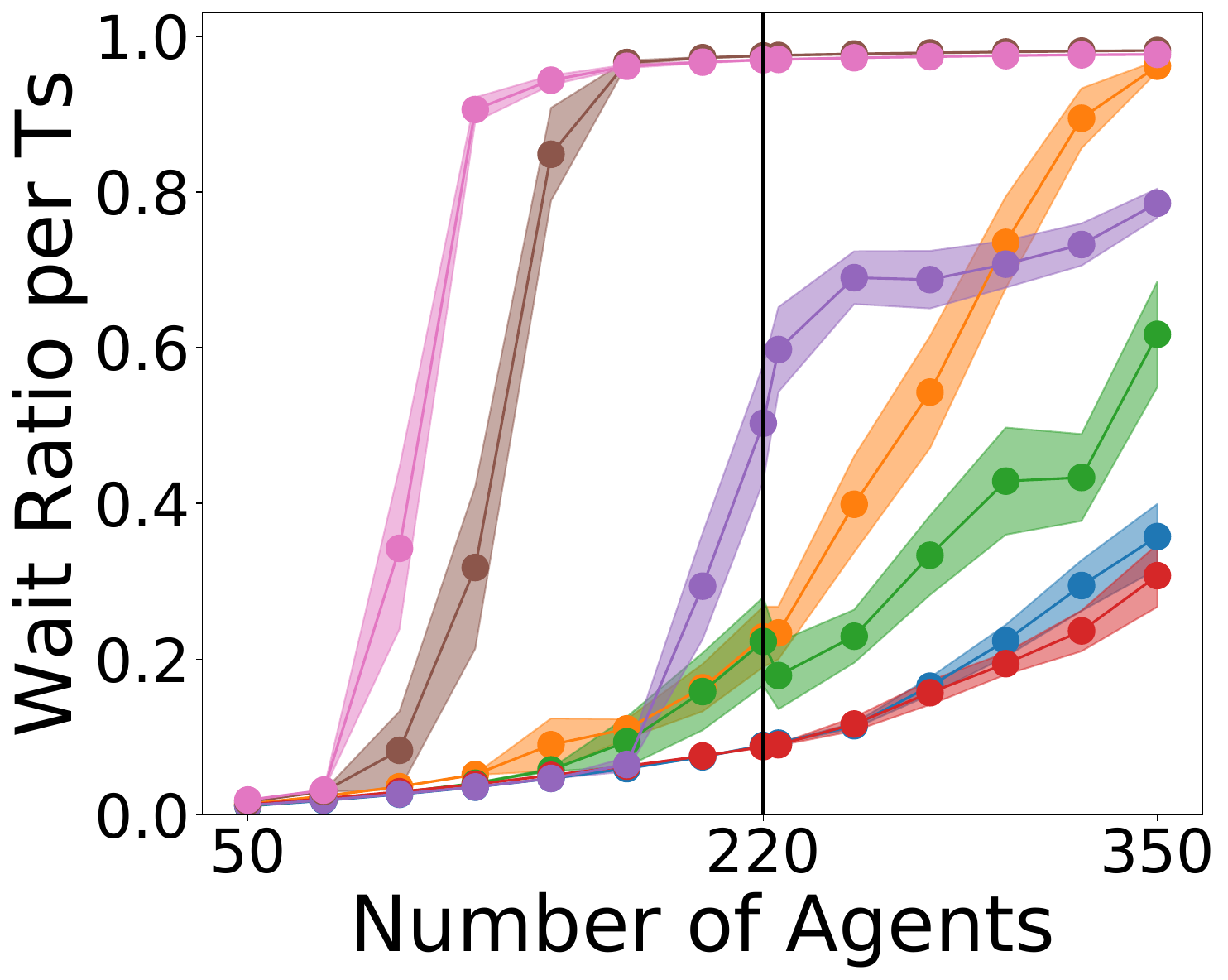}
    \end{subfigure}\hfill
    \begin{subfigure}[t]{0.16\textwidth}
        \centering
        \includegraphics[width=\linewidth]{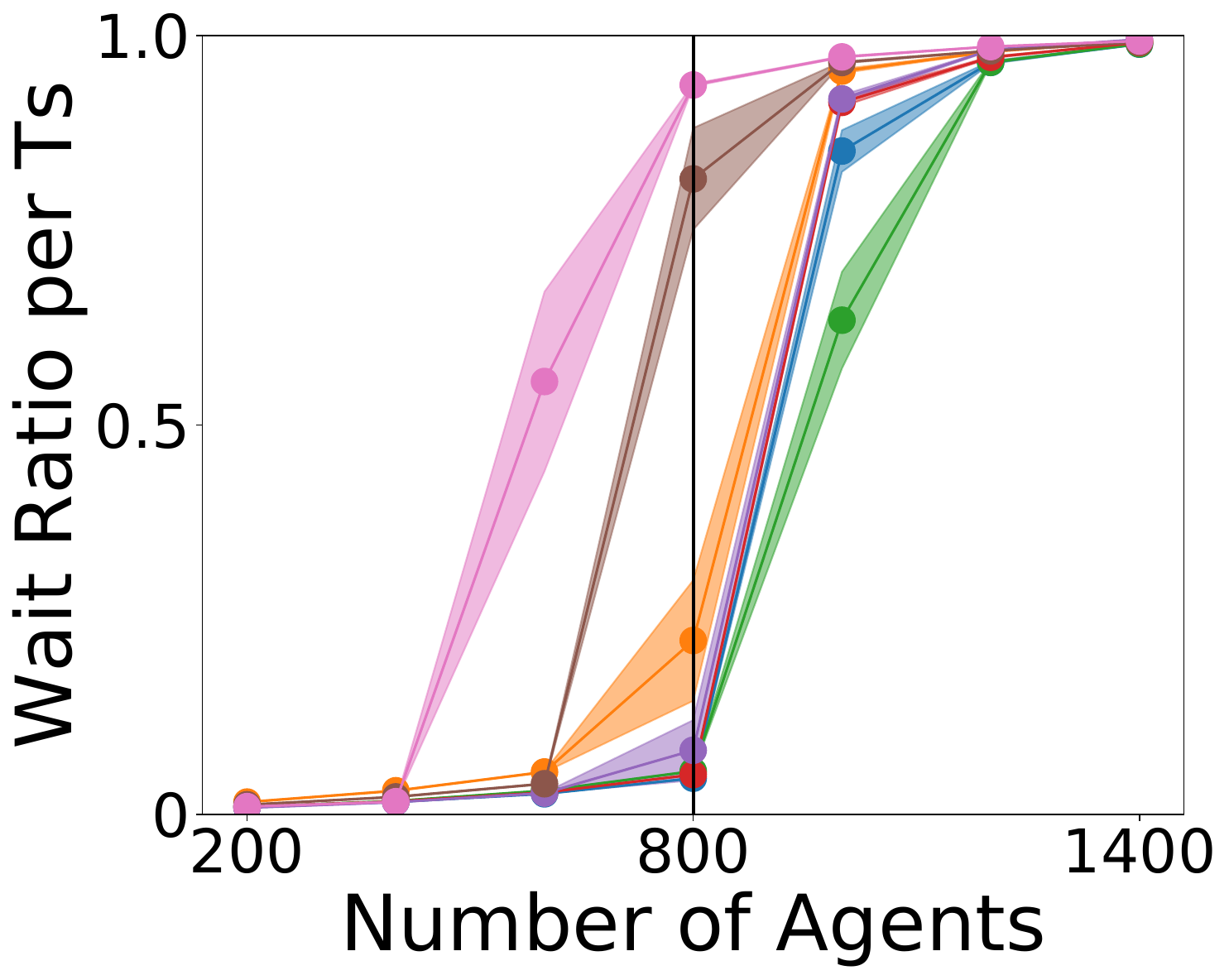}
    \end{subfigure}

    % -------- Row 3 --------
    \begin{subfigure}[t]{0.16\textwidth}
        \centering
        \includegraphics[width=\linewidth]{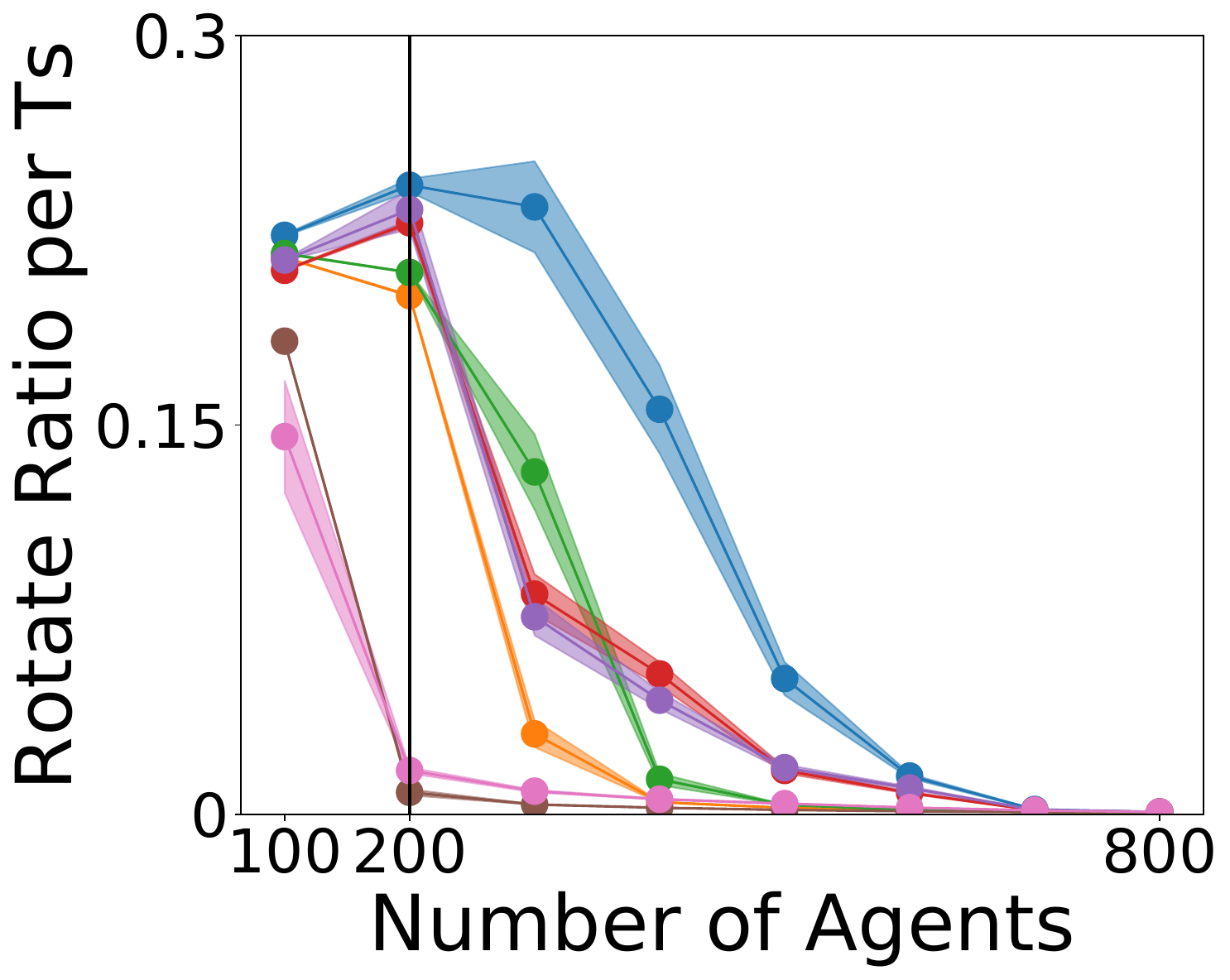}
    \end{subfigure}\hfill
    \begin{subfigure}[t]{0.16\textwidth}
        \centering
        \includegraphics[width=\linewidth]{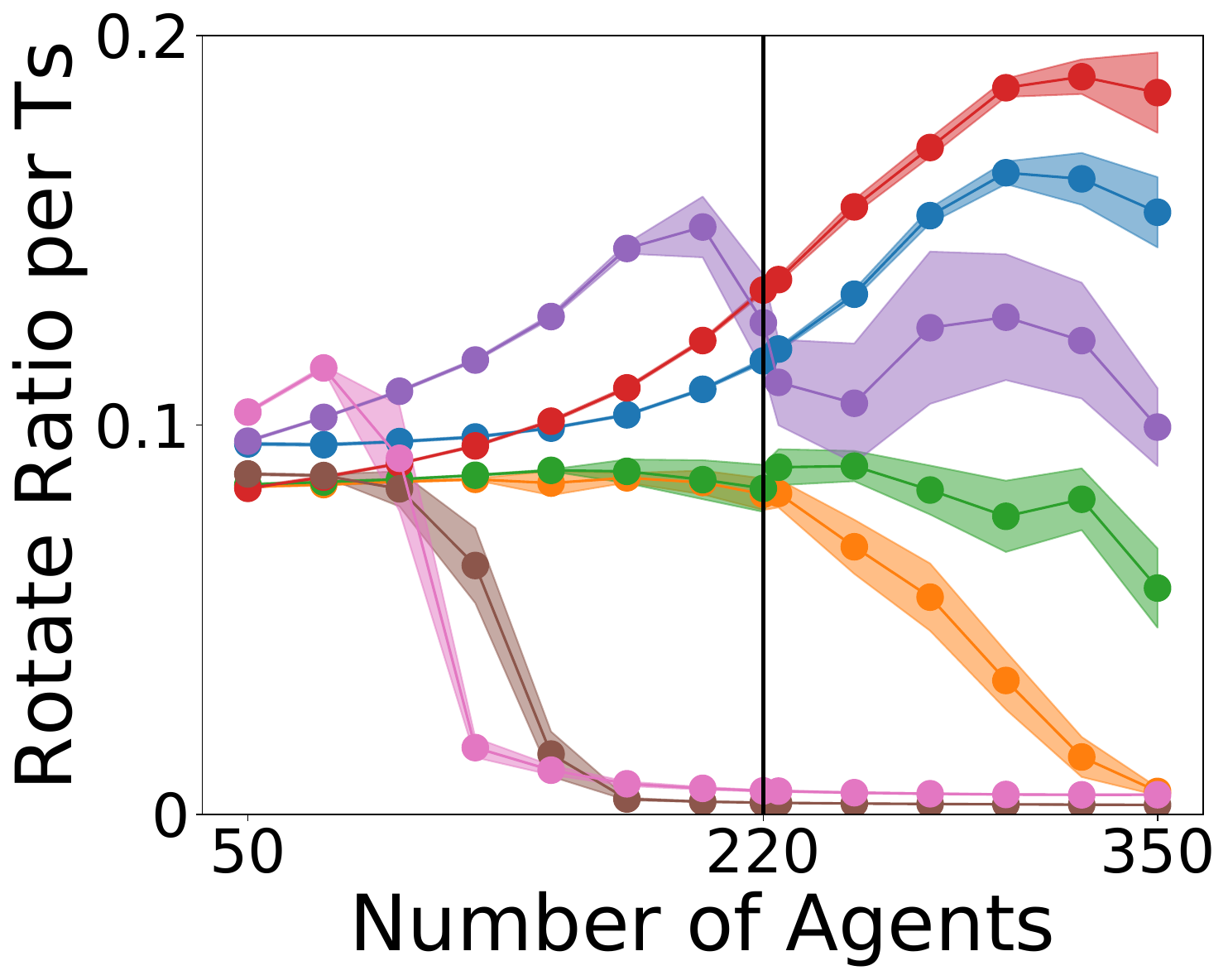}
    \end{subfigure}\hfill
    \begin{subfigure}[t]{0.16\textwidth}
        \centering
        \includegraphics[width=\linewidth]{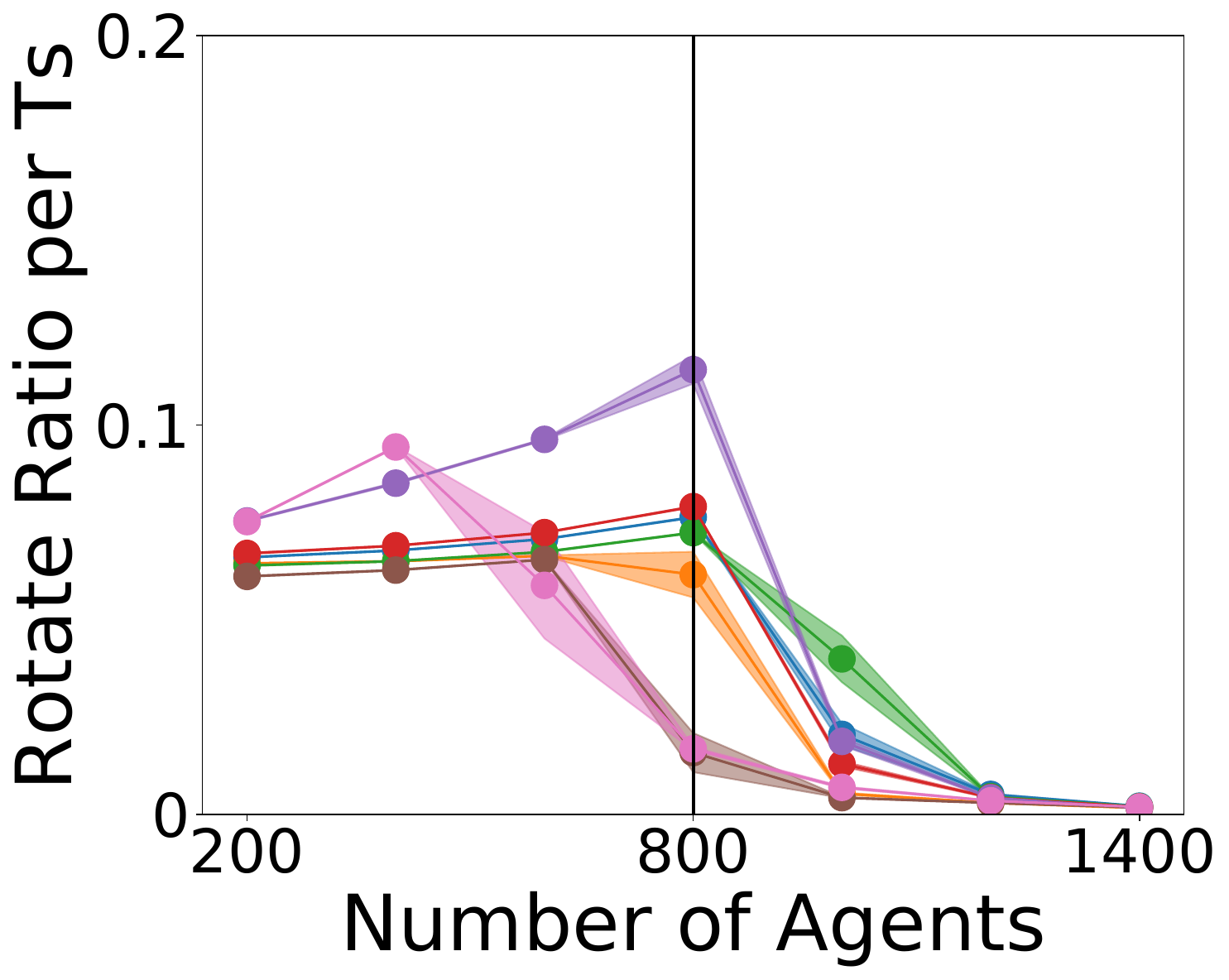}
    \end{subfigure}
    \par\vspace{-\abovecaptionskip}
    \subcaptionbox{Setup 1\label{fig:rhcr-r-random-32-32-20}}
        {\phantom{\rule{0.16\textwidth}{0pt}}}\hfill
    \subcaptionbox{Setup 2\label{fig:rhcr-r-warehouse-33-36}}
        {\phantom{\rule{0.16\textwidth}{0pt}}}\hfill
    \subcaptionbox{Setup 3\label{fig:rhcr-r-empty-48-48}}
        {\phantom{\rule{0.16\textwidth}{0pt}}}\hfill

    \caption{Throughput, wait action ratio, and rotation action ratio with different numbers of agents for RHCR in setups 1 to 3.}
    \label{fig:major-result-rhcr-r}
\end{figure}

\begin{figure*}[!t]
    \centering
    \includegraphics[width=1\linewidth]{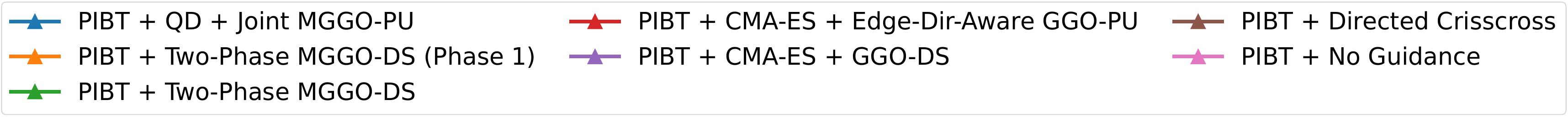}\par\medskip

    % -------- Row 1 --------
    \begin{subfigure}[t]{0.19\textwidth}
        \centering
        \includegraphics[width=\linewidth]{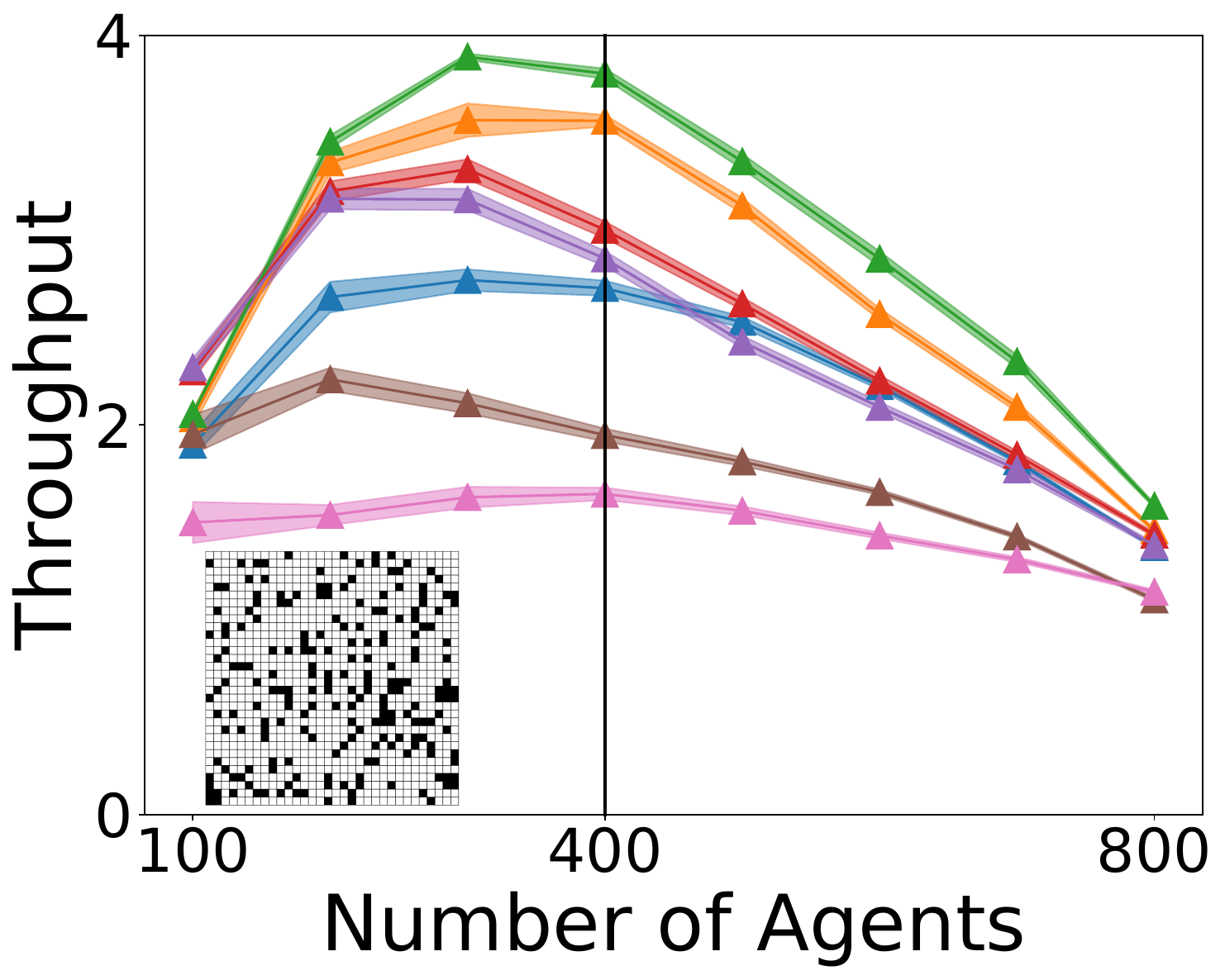}
    \end{subfigure}\hfill
    \begin{subfigure}[t]{0.19\textwidth}
        \centering
        \includegraphics[width=\linewidth]{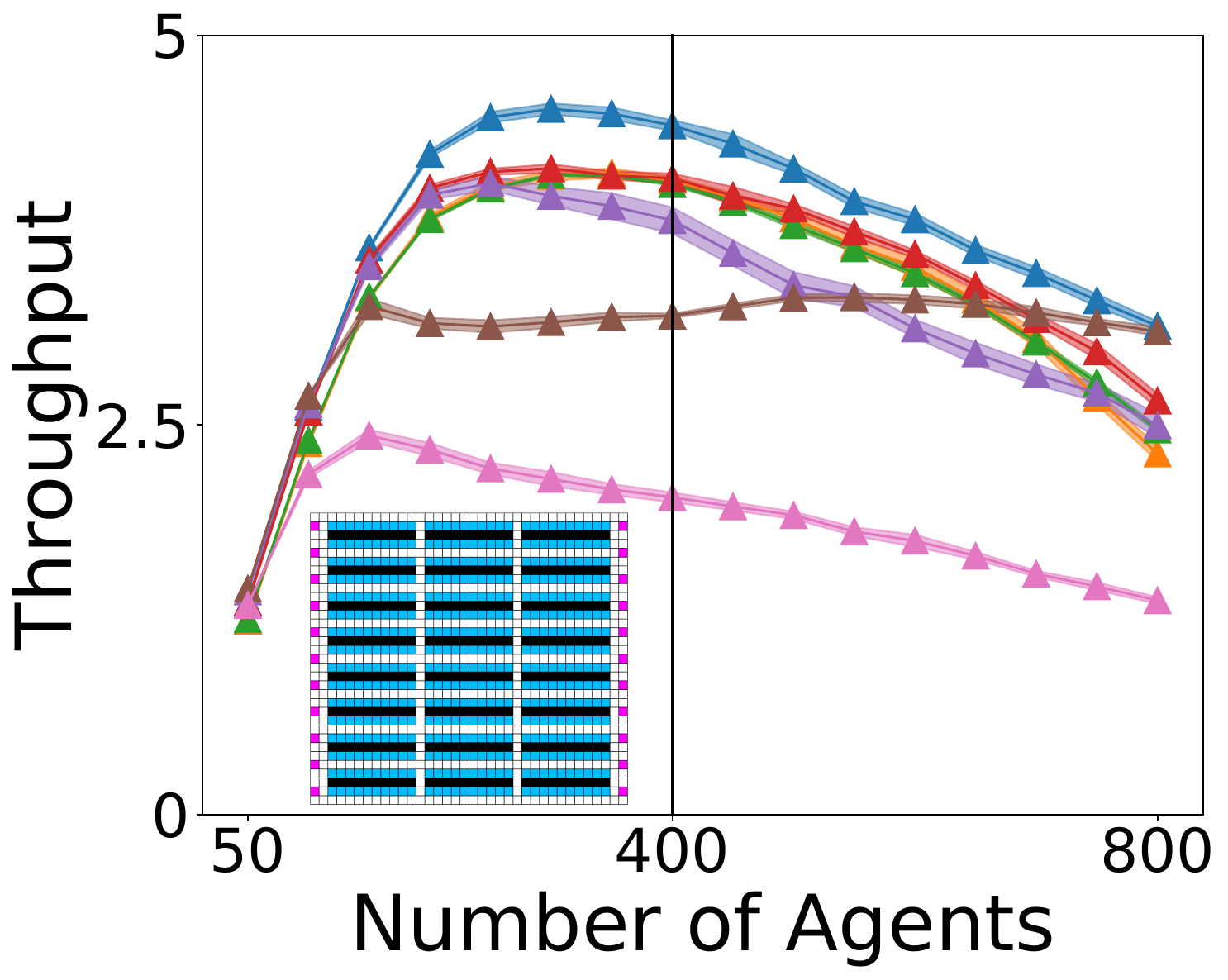}
    \end{subfigure}\hfill
    \begin{subfigure}[t]{0.19\textwidth}
        \centering
        \includegraphics[width=\linewidth]{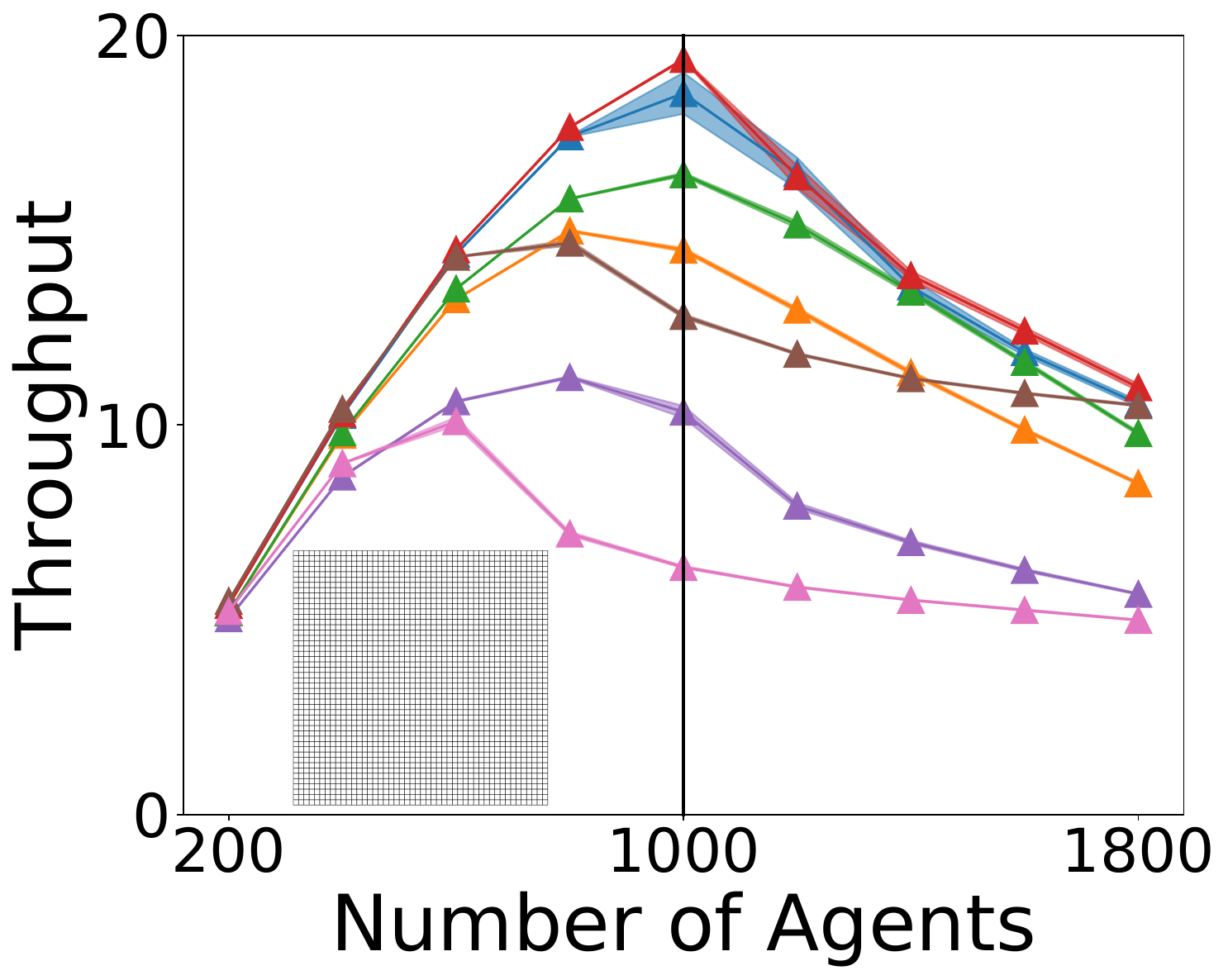}
    \end{subfigure}\hfill
    \begin{subfigure}[t]{0.19\textwidth}
        \centering
        \includegraphics[width=\linewidth]{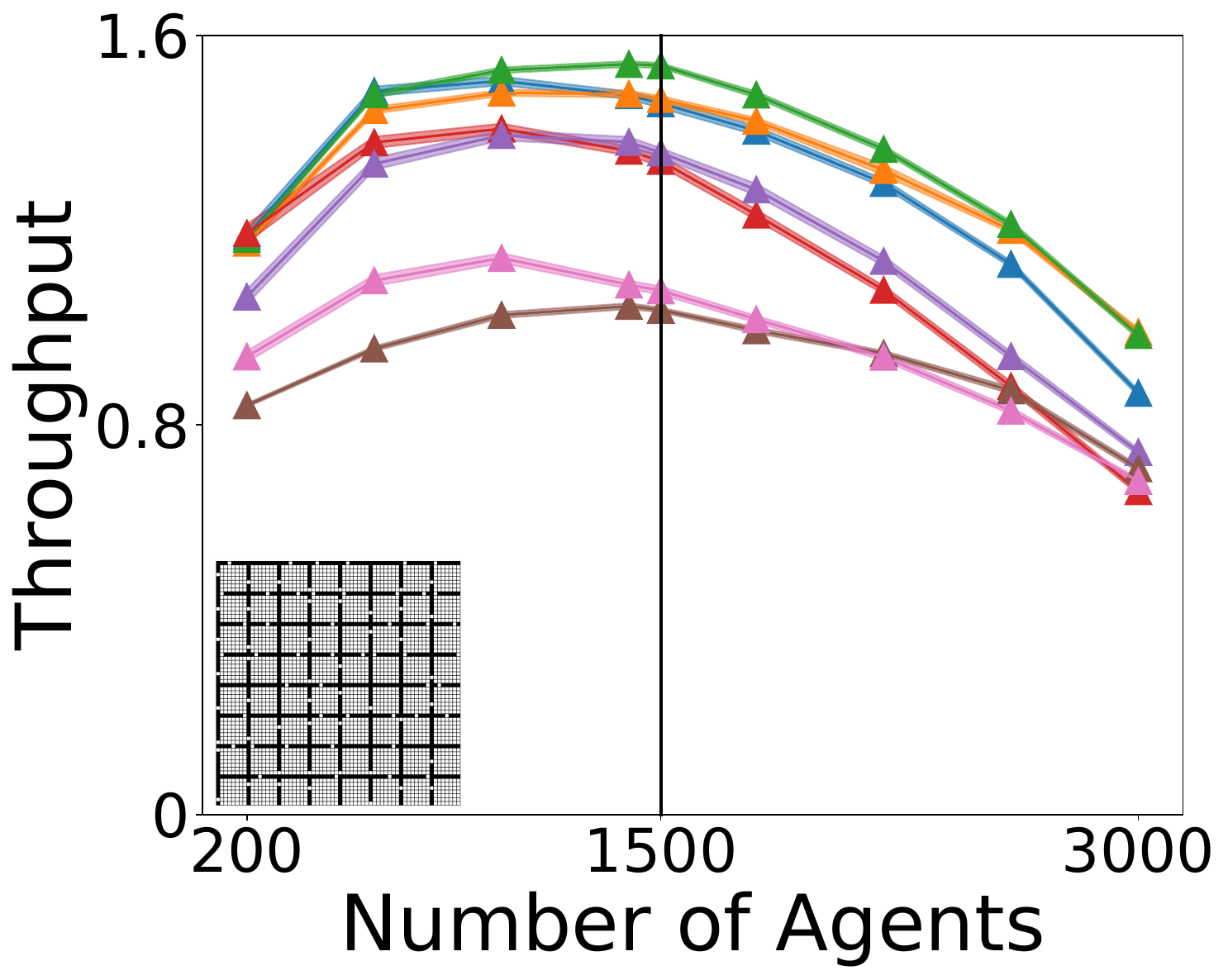}
    \end{subfigure}\hfill
    \begin{subfigure}[t]{0.19\textwidth}
        \centering
        \includegraphics[width=\linewidth]{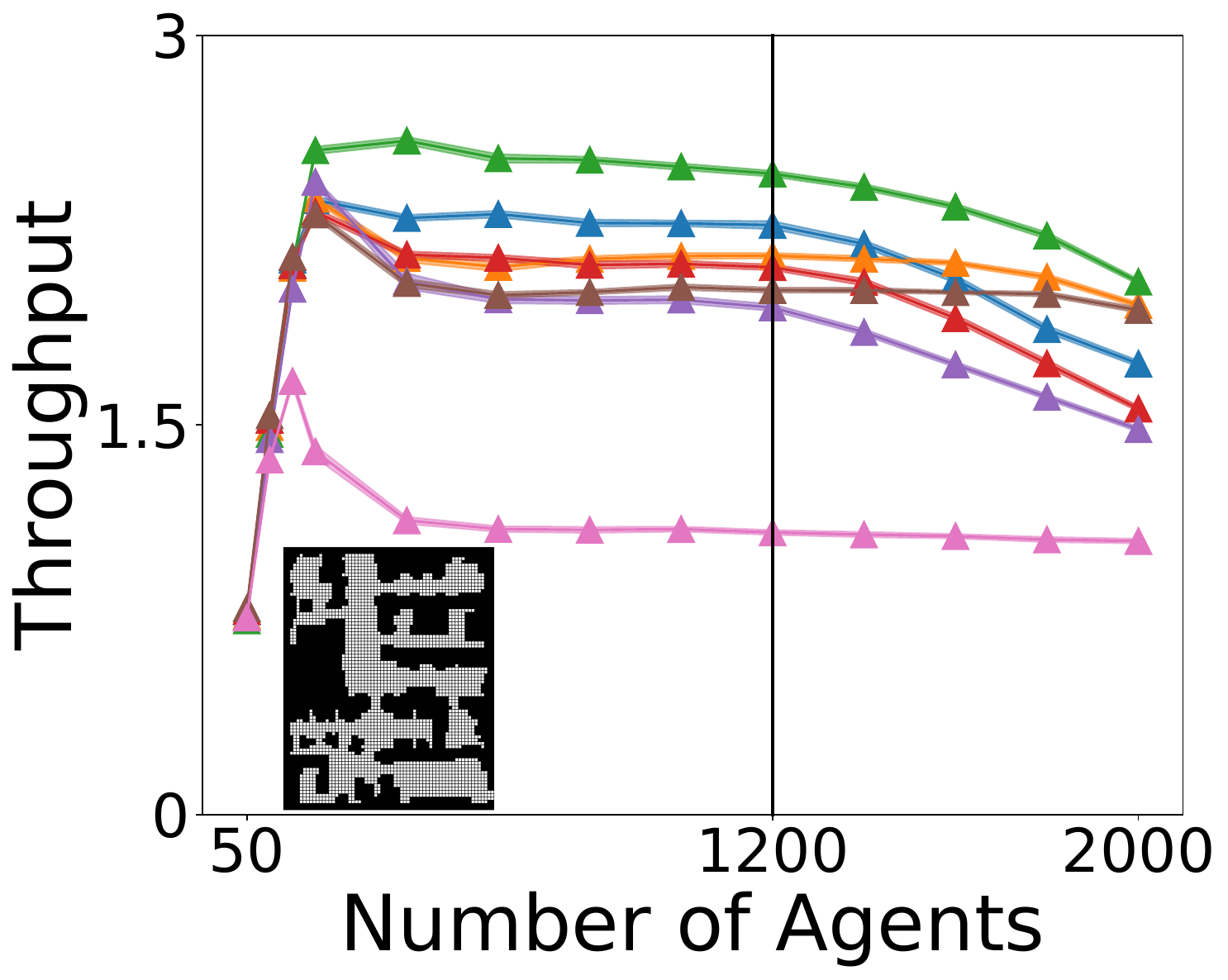}
    \end{subfigure}\par\vspace{3pt}

    % -------- Row 2 --------
    \begin{subfigure}[t]{0.19\textwidth}
        \centering
        \includegraphics[width=\linewidth]{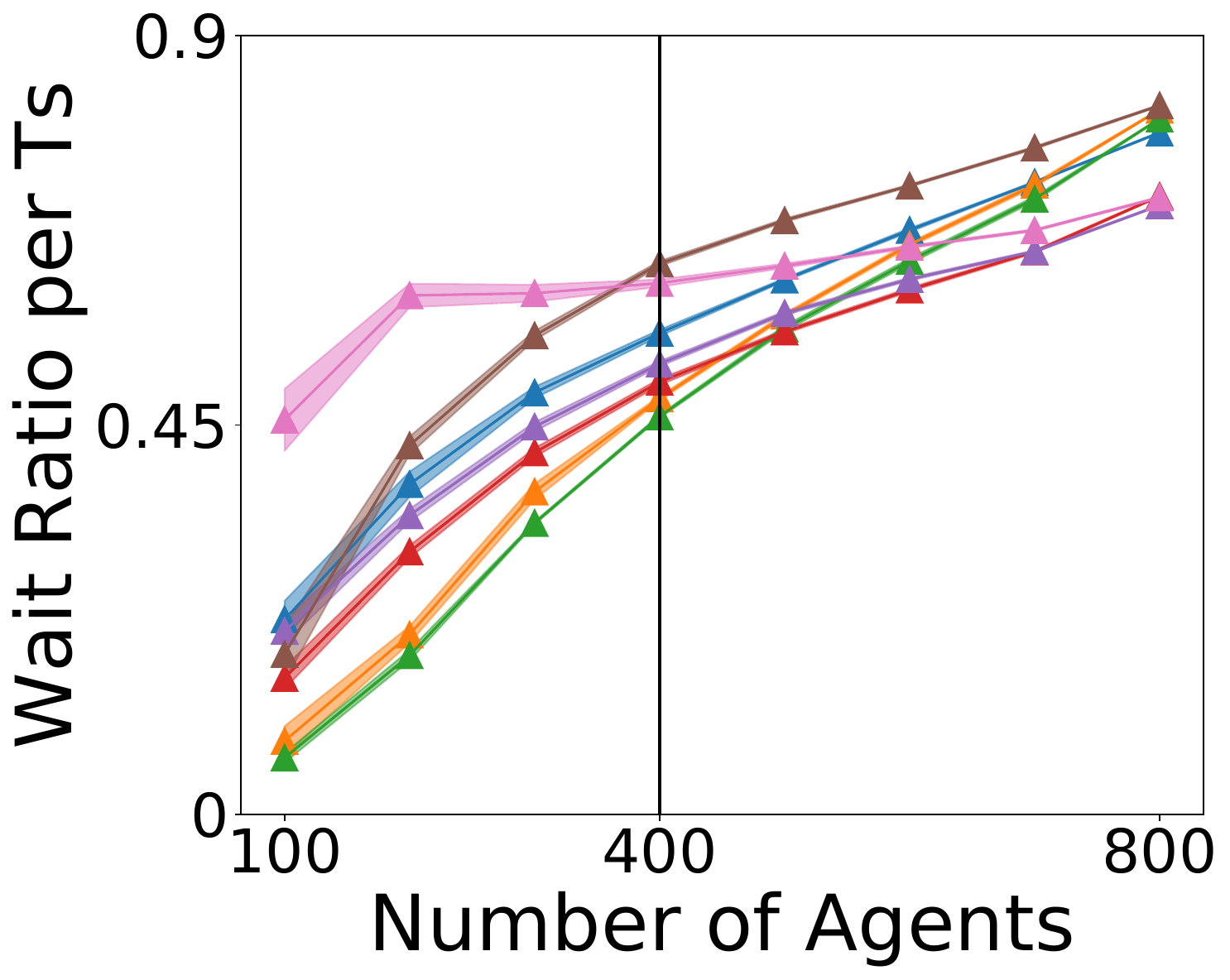}
    \end{subfigure}\hfill
    \begin{subfigure}[t]{0.19\textwidth}
        \centering
        \includegraphics[width=\linewidth]{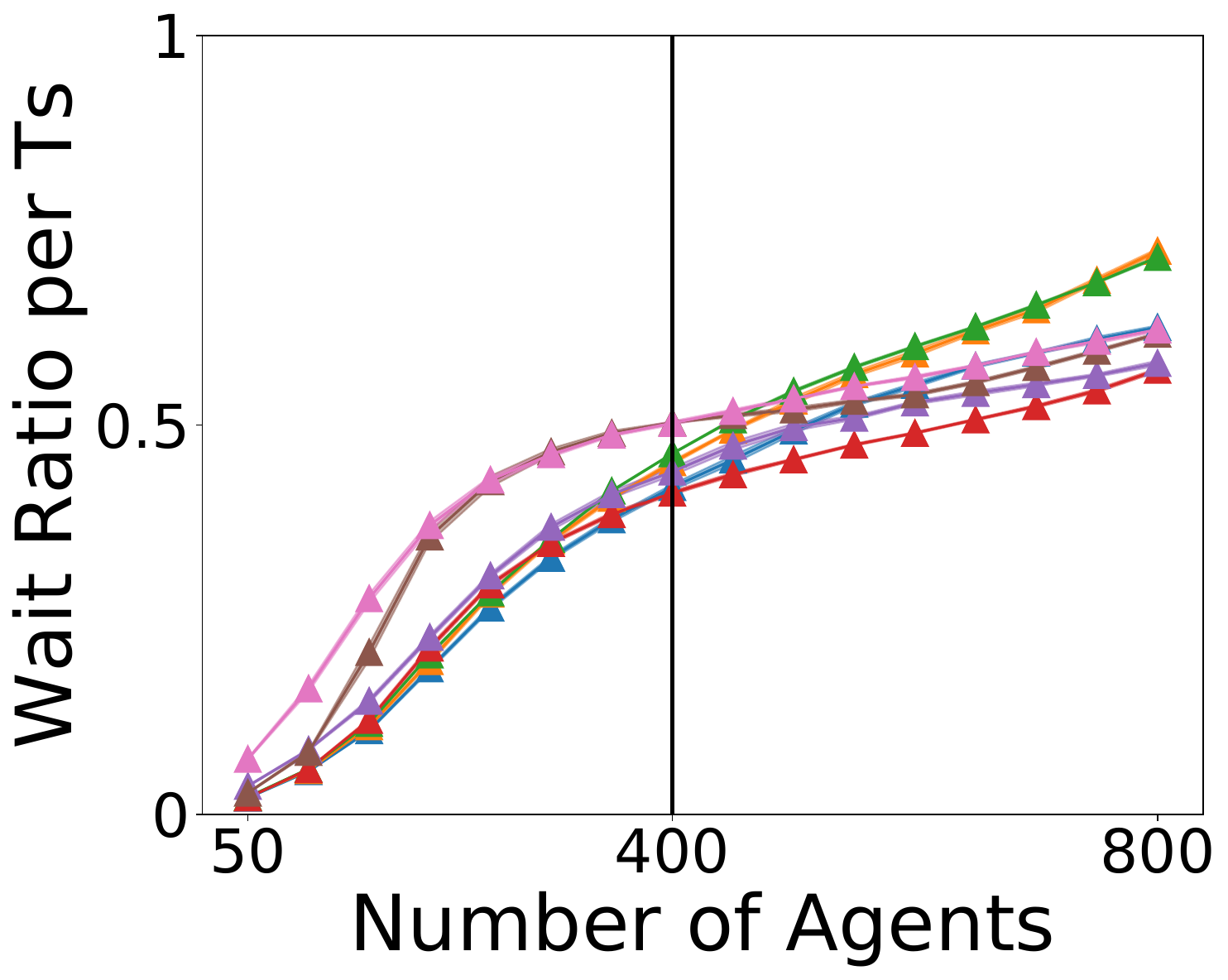}
    \end{subfigure}\hfill
    \begin{subfigure}[t]{0.19\textwidth}
        \centering
        \includegraphics[width=\linewidth]{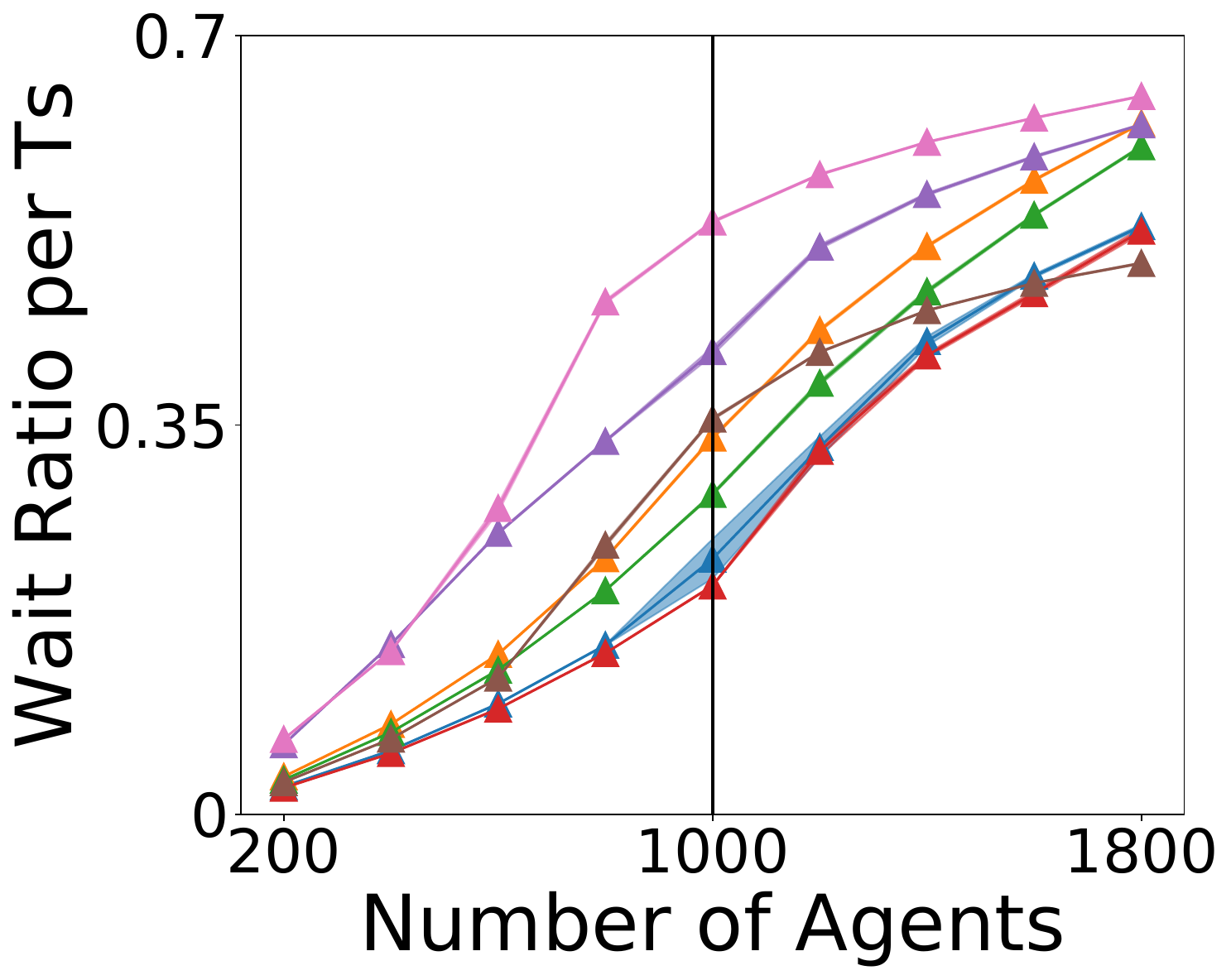}
    \end{subfigure}\hfill
    \begin{subfigure}[t]{0.19\textwidth}
        \centering
        \includegraphics[width=\linewidth]{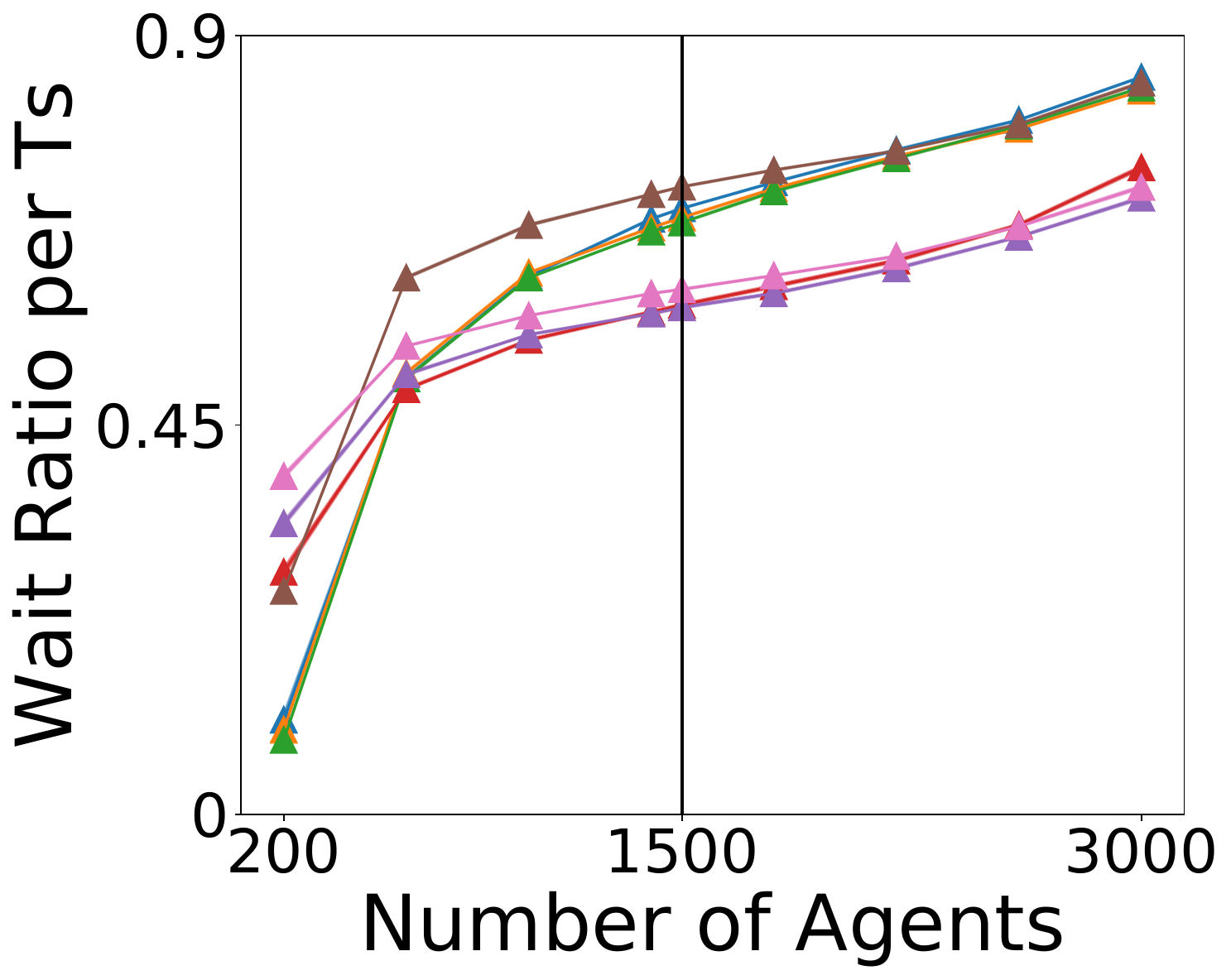}
    \end{subfigure}\hfill
    \begin{subfigure}[t]{0.19\textwidth}
        \centering
        \includegraphics[width=\linewidth]{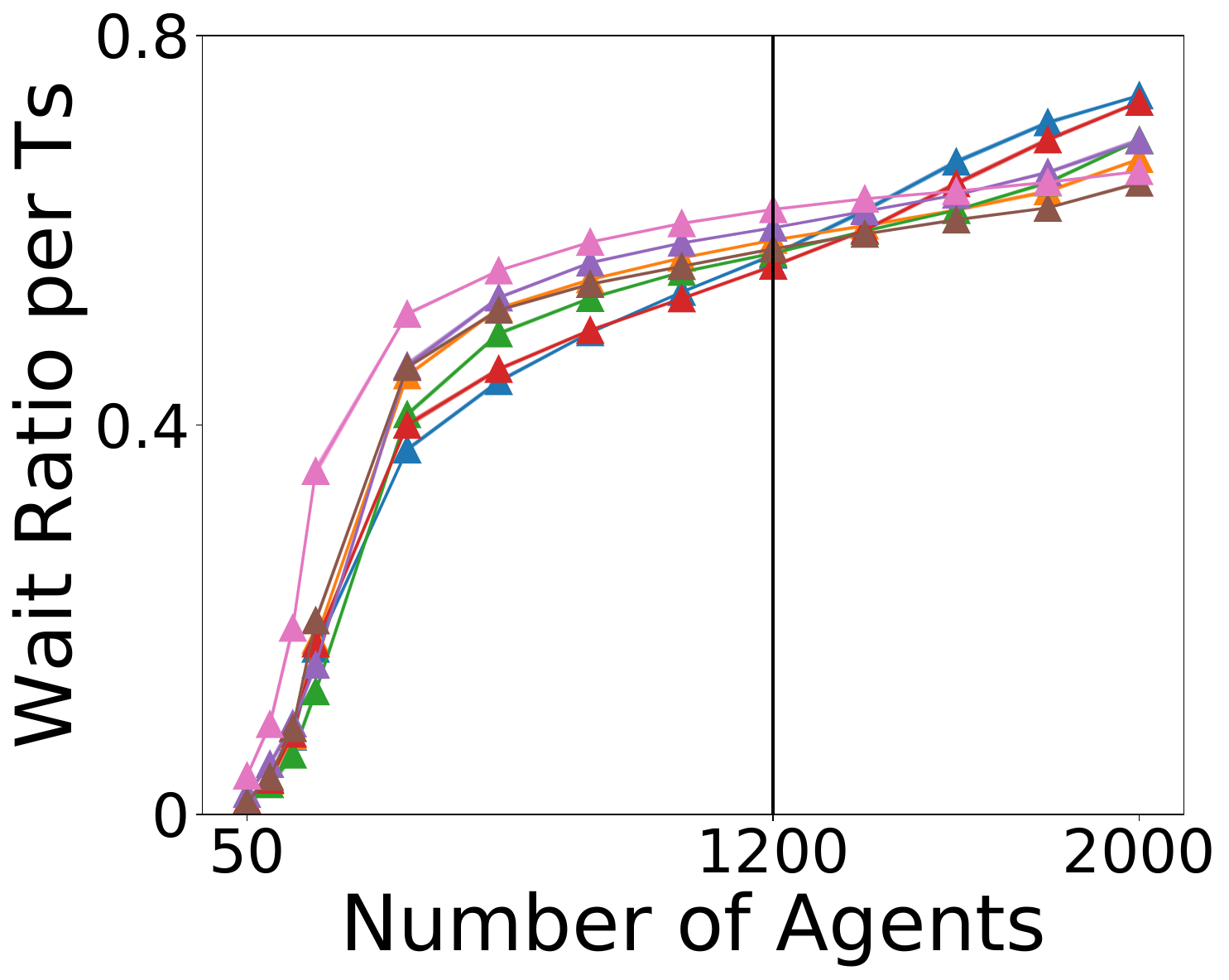}
    \end{subfigure}\par\vspace{3pt}

    % -------- Row 3 --------
    \begin{subfigure}[t]{0.19\textwidth}
        \centering
        \includegraphics[width=\linewidth]{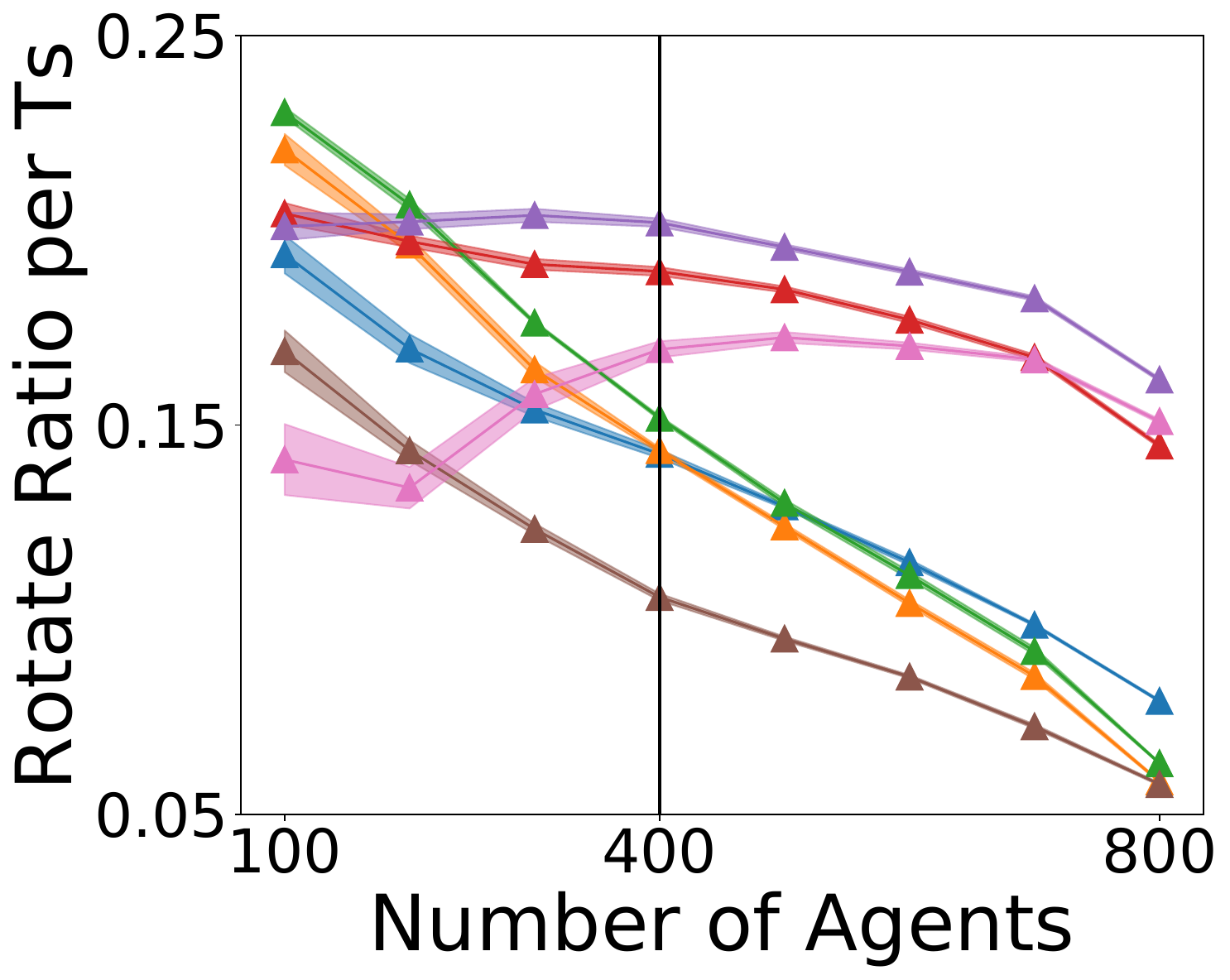}
    \end{subfigure}\hfill
    \begin{subfigure}[t]{0.19\textwidth}
        \centering
        \includegraphics[width=\linewidth]{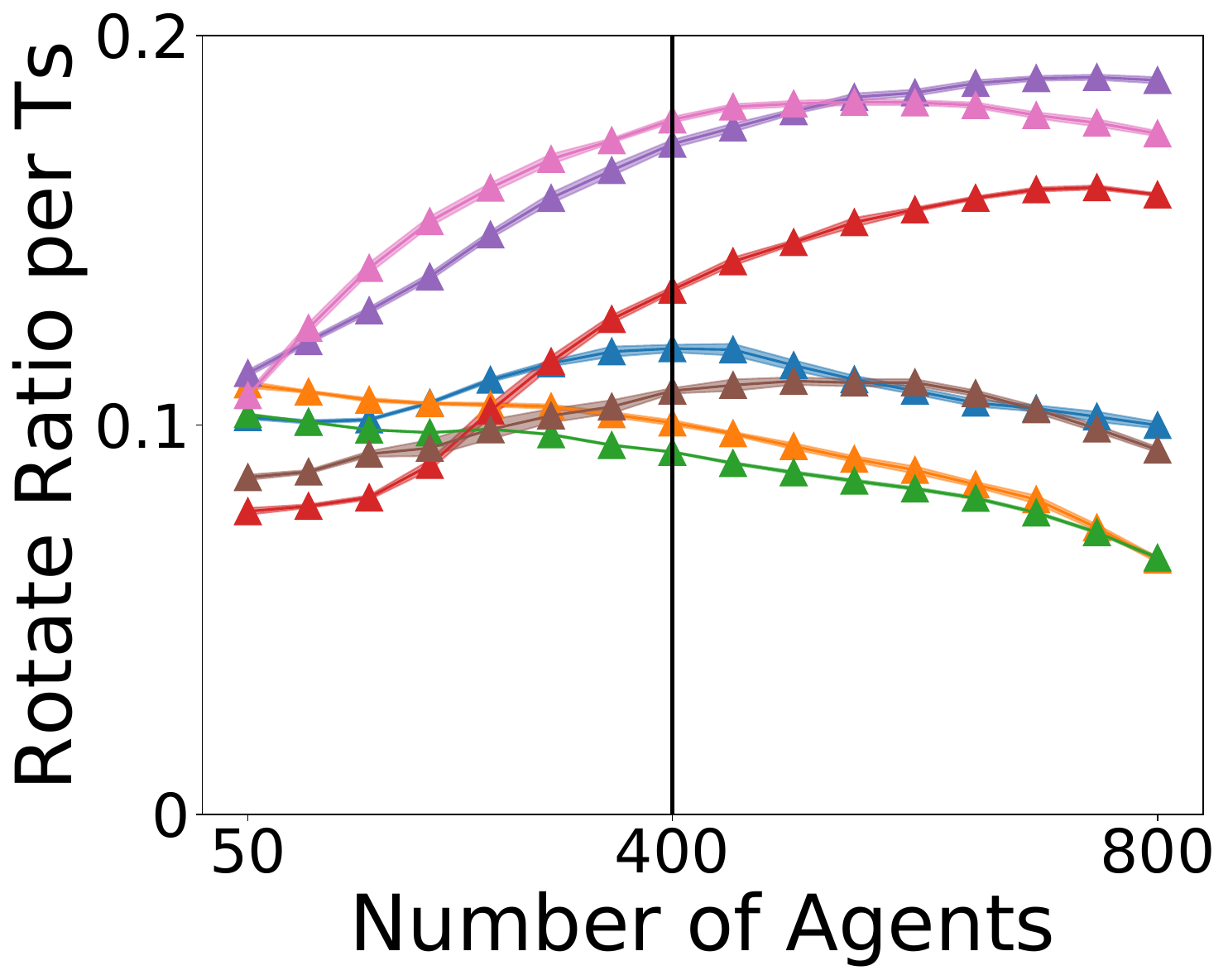}
    \end{subfigure}\hfill
    \begin{subfigure}[t]{0.19\textwidth}
        \centering
        \includegraphics[width=\linewidth]{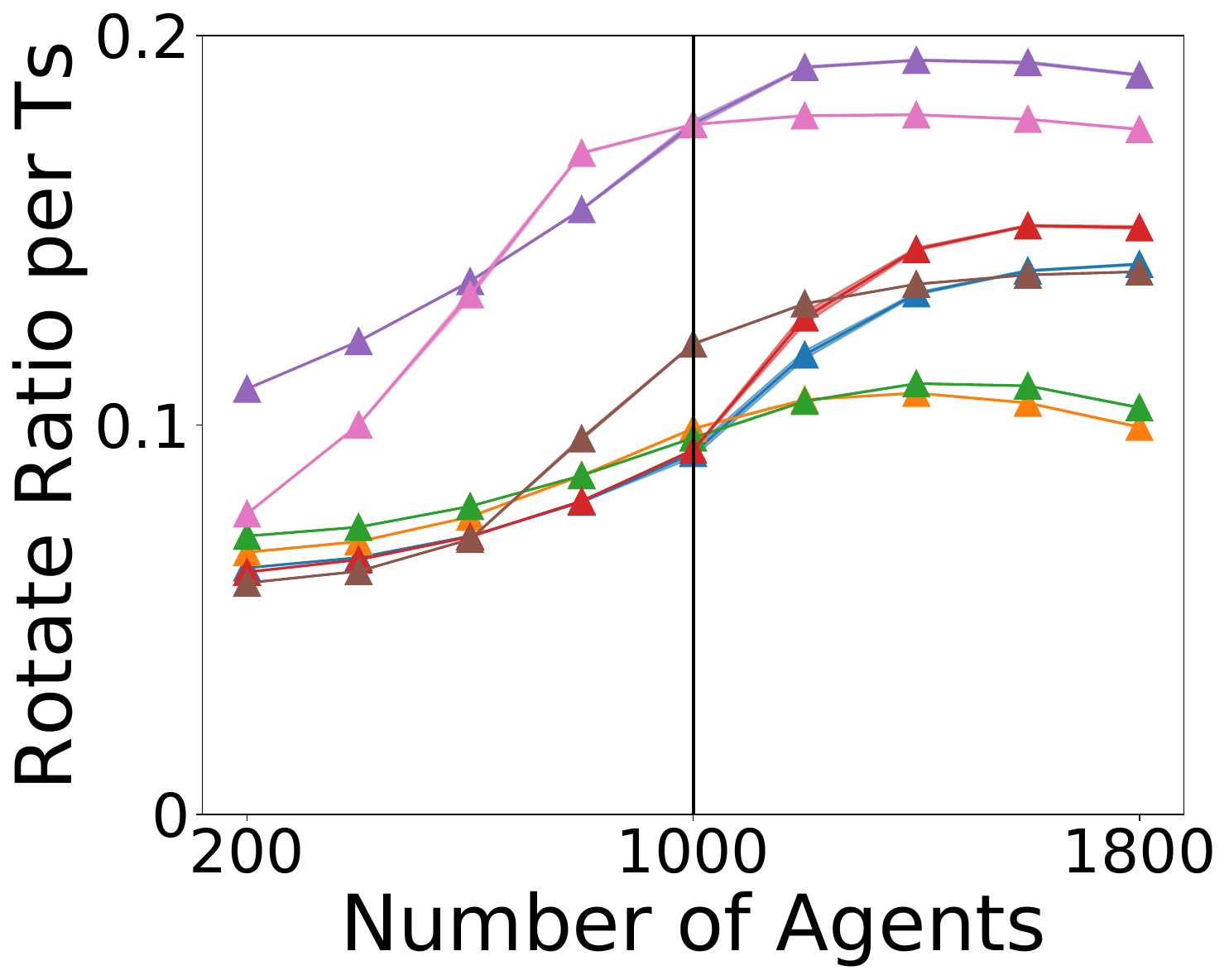}
    \end{subfigure}\hfill
    \begin{subfigure}[t]{0.19\textwidth}
        \centering
        \includegraphics[width=\linewidth]{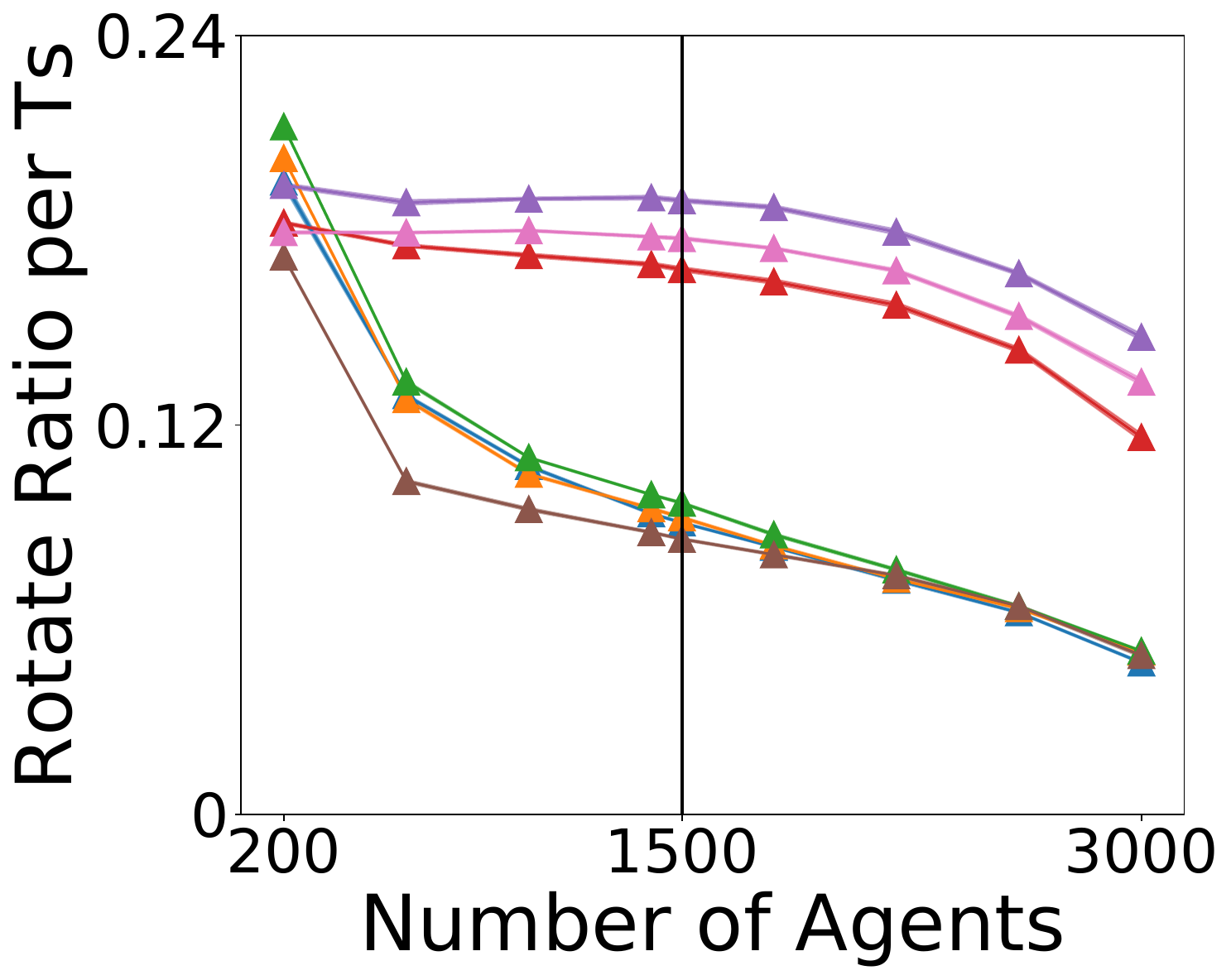}
    \end{subfigure}\hfill
    \begin{subfigure}[t]{0.19\textwidth}
        \centering
        \includegraphics[width=\linewidth]{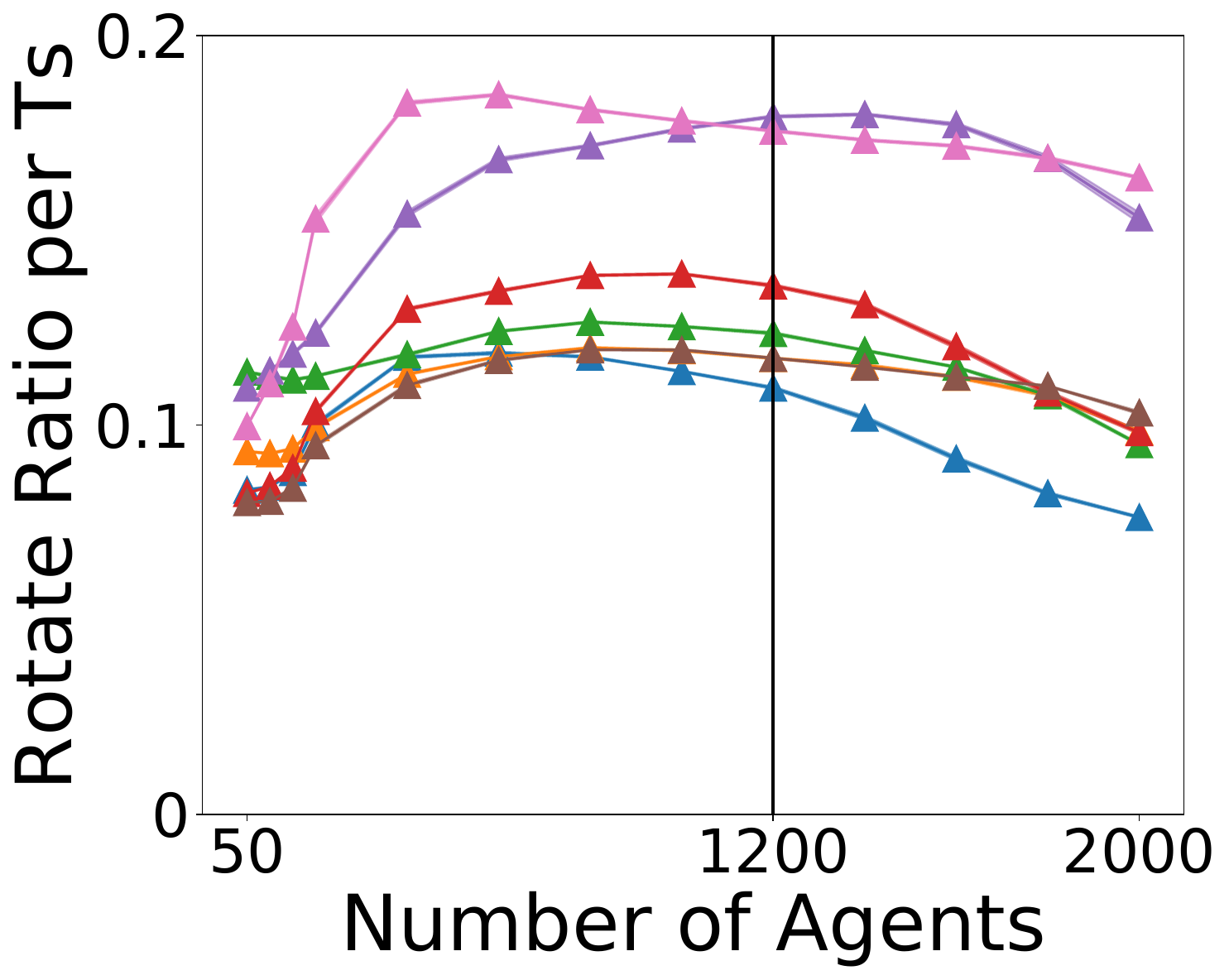}
    \end{subfigure}
    \par\vspace{-\abovecaptionskip}
    \subcaptionbox{Setup 4\label{fig:pibt-r-random-32-32-20}}
        {\phantom{\rule{0.19\textwidth}{0pt}}}\hfill
    \subcaptionbox{Setup 5\label{fig:pibt-r-warehouse-33-36}}
        {\phantom{\rule{0.19\textwidth}{0pt}}}\hfill
    \subcaptionbox{Setup 6\label{fig:pibt-r-empty-48-48}}
        {\phantom{\rule{0.19\textwidth}{0pt}}}\hfill
    \subcaptionbox{Setup 7\label{fig:pibt-r-room-64-64-8}}
        {\phantom{\rule{0.19\textwidth}{0pt}}}\hfill
    \subcaptionbox{Setup 8\label{fig:pibt-r-den312d}}
        {\phantom{\rule{0.19\textwidth}{0pt}}}

    \caption{Throughput, wait action ratio, and rotation action ratio with different numbers of agents for PIBT in setups 4 to 8. }
    \label{fig:major-result-pibt-r}
\end{figure*}

\subsubsection{Proposed Methods vs. Baselines}
We run all optimized (mixed) guidance graphs with various numbers of agents, each for 50 simulations, and show the results in \Cref{fig:major-result-pibt-r,fig:major-result-rhcr-r}. We show throughput and the ratio of agents that wait and rotate at each timestep (Ts). The solid lines are the average, and the shaded areas are 95\% confidence intervals. The black vertical lines indicate $N_a$.

With PIBT, one of our proposed methods always achieves the highest throughput, but the best method varies across setups. Interestingly, maximizing the number of unidirectional edges with the two-phase MGGO-DS reduces the ratio of rotation actions, potentially at the cost of more waiting. This is because unidirectional edges can prevent agents from taking suboptimal detours, especially with PIBT. This helps the two-phase MGGO-DS to achieve the highest throughput in \randomSmall, \roomLarge, and \denSmall. In \emptyMid, however, agents have more space to coordinate, so using edge-direction-aware GGO-PU to optimize the edge weights outperforms all MGGO methods. Nevertheless, our edge-direction-aware GGO.PU outperforms GGO-DS on most maps, highlighting the importance of incorporating edge-direction-related information into GGO. In \warehouselargeW, the QD-based Joint MGGO-PU achieves the best performance, while the other two proposed methods achieve similar performance, indicating the need for joint optimization.

With RHCR in setup 2, the QD-based Joint MGGO-PU achieves the highest throughput in \warehouselargeW, and unidirectional edges are helpful in terms of reducing the rotation action ratio.
However, for setups 1 and 3, optimizing edge directions is less helpful, and we do not observe significant differences on the ratio of rotation actions in these setups. Equipped with an optimal low-level planner, RHCR is less inclined to have agents traversing the high-cost edges than PIBT. Therefore, optimizing edge weights is good enough in most cases for RHCR. Nevertheless, providing edge-direction-related information to GGO continues to be helpful, allowing edge-direction-aware GGO-PU to achieve the best throughput in setups 1 and 3. We show addition numerical results at $N_a$ in \Cref{appen:numerical-result}.

\subsubsection{Ablations}

\begin{table}[!t]
    \centering
    \resizebox{1\linewidth}{!}{
        \begin{tabular}{ccrrr}
\toprule
MGGO & Throughput & CPU Runtime \\
\midrule
\textbf{QD + Joint MGGO-PU} & $\textbf{4.42} \pm \textbf{0.02}$ & $2.76 \pm 0.03$\\
CMA-ES + Joint MGGO-PU  & $4.26 \pm 0.01$ & $3.35 \pm 0.01$\\
QD + Joint MGGO-PU (Indept Edge-Dir Rep) & $4.18 \pm 0.01$ & $\textbf{2.57} \pm \textbf{0.01}$\\
\midrule
\textbf{MGGO-DS (Phase 1)} & $\textbf{4.07} \pm \textbf{0.01}$ & $\textbf{2.40} \pm \textbf{0.02}$\\
MGGO-DS (Phase 1, Bi-Dir) & $3.93 \pm 0.01$ & $3.03 \pm 0.02$\\
MGGO-DS (Phase 1, k-node Only) & $3.86 \pm 0.01$ & $3.01 \pm 0.02$\\
MGGO-DS (Phase 1, k-node/k-edge Only) & $3.90 \pm 0.01$ & $2.76 \pm 0.02$\\
\bottomrule
        \end{tabular}
    }
    \caption{Numerical results of ablation experiments.
    We run each graph for 50 times and report the results in the format of $x \pm y$, where $x$ and $y$ are the average and standard error, respectively.
   }
    \label{tab:numerical-result-ablation}
\end{table}

% Definitely:
% For Joint opt: QD + obj = tp + edge_sim vs. CMA-ES + obj = tp
% model architecture: each node predict direction going right and down (instead of all four directions)
\paragraph{Variants of Joint MGGO-PU.}
We compare our proposed QD-based Joint MGGO-PU with two variants to justify our design choices: (1) CMA-ES + Joint MGGO-PU, where we use CMA-ES to directly optimize throughput to show the necessity of using QD and regularizing with edge similarity scores, and (2) QD + Joint MGGO-PU (Indept Edge-Dir Rep), where we use the independent edge direction representation from \Cref{sec:qd-mggo} in both input and output of the update model to justify our choice of representation. We use the same $N_e$, $N_{eval}$, and $T$ in setup 5.
We present the result with $N_a = 400$ agents in \Cref{tab:numerical-result-ablation}, showing that our proposed Joint MGGO-PU method achieves the highest throughput. 

% Include bi-directed edges or not during optimization for GDO
\paragraph{Variants of $(1+\lambda)$ EA in MGGO-DS Phase One.}
We compare our proposed $(1+\lambda)$ EA with three variants: (1) MGGO-DS (Phase 1, Bi-Dir), where the decision variables can assign an edge to be bidirected, (2) MGGO-DS (Phase 1, $k$-vertices Only), where only the $k$-vertices mutation operator is used, and (3) MGGO-DS (Phase 1, $k$-vertices/$k$-edges Only), where only $k$-vertices and $k$-edges mutation operators are used. \Cref{tab:numerical-result-ablation} shows that including bidirected edges does not help optimize the edge directions. In addition, using diverse mutation operators helps with optimizations. We show more ablation results in \Cref{appen:qd-vs-cma-es-mggo,appen:gdo-ablation}, and example optimized mixed guidance graphs in \Cref{appen:example-mgg}.

% Maybe:
% PIU with iter=1 + observation = NoG + CC vs. PIU with iter=5 + observation = NoG (old setup)
% Have results with PIBT, but not with PIBT_R

\section{Conclusion}

We extend GGO to MGGO by optimizing both edge directions and weights, and introduce \ERS, a fast algorithm to enforce strong connectivity via edge reversals. We also incorporate edge-direction-related information to GGO. Empirically, our proposed methods outperforms the baselines, establishing a new state-of-the-art for guidance in LMAPF.
We discover that edge directions are more beneficial for rule-based planners such as PIBT, primarily by reducing rotation actions while offering more limited gains for search-based methods such as RHCR. 
However, MGGO does not consistently outperform all GGO variants, likely because of reduced sample efficiency due to increased complexity. Future work may improve sample efficiency via surrogate modeling~\cite{Zhang2021DeepSA}.

\section*{Acknowledgments}
This work is in part supported by the National Science Foundation (NSF) under grant numbers \#$2328671$ and \#$2441629$, as well as a gift from Amazon.
This work used Bridge-$2$ at Pittsburgh Supercomputing Center (PSC)~\cite{PSCBridgeTwo2021} through allocation CIS$220115$ from the Advanced Cyberinfrastructure Coordination Ecosystem: Services \& Support (ACCESS) program, which is supported by NSF under grant numbers \#$2138259$, \#$2138286$, \#$2138307$, \#$2137603$, and \#$2138296$.

\bibliographystyle{named}
\bibliography{ijcai26}

\clearpage

\appendix

\section{Theoretical and Empirical Results of \ERS} 

\subsection{Proof} \label{appen:efs-proof}

% \begin{lemma} \label{lemma:d-is-dag}
%     If $G_{in}$ is not strongly connected and we build a condensation graph $\mathcal{D(V,E)}$ following \ERS in \Cref{alg:edge-dir-flip-search}, then $\mathcal{D}$ is a Directed Acyclic Graph (DAG).
% \end{lemma}

% \begin{proof}
%     For the sake of contradiction, assume that $\mathcal{D}$ is not a DAG. Then this is at least one directed cycle in $\mathcal{D}$. Since all meta vertices on the directed cycle can reach each other, the vertices in the SCCs of $G_{in}$ corresponding to the mete vertices on the directed cycle can form a new SCC. This contradicts the fact that all SCCs are found using the Tarjan algorithm. Therefore $\mathcal{D}$ is a DAG.
% \end{proof}

\begin{lemma} \label{lemma:D-weak-connect}
    The condensation graph $\mathcal{D}$ is weakly connected.
\end{lemma}

\begin{proof}
    According to \citet{cormen2009introalgo}, $\mathcal{D}$ is a DAG. Therefore, $\mathcal{D}$ is trivially weakly connected.
\end{proof}

\begin{theorem} \label{theorem:have-rev-edges}
    If, from a meta vertex $\eta \in \mathcal{V}$ s.t. $deg^-(\eta) = 0$, we follow \ERS to find the set of edges $E_{rev}$, then $|E_{rev}| \geq 1$.
\end{theorem}

\begin{proof}
    For the sake of contradiction, assume that $|E_{rev}| = 0$. Then $|E_{out}| \leq 1$.
    
    If $|E_{out}| = 0$, there are no edges between $\eta$ and other meta vertices, contradicting \Cref{lemma:D-weak-connect}.
    
    If $|E_{out}| = 1$, there is only 1 vertex $v \in V_{in}$ in the SCC that corresponds to $\eta$ in $G_{in}$. Let the SCC be $C_\eta$. Then the edge $e = (v,u) \in E_{in}$ that goes from $v$ to a vertex $u \not\in C_\eta$ must be a bridge. This is because removing $e$ will disconnect $v$ from $G_{in}$. Since all bridges are assumed to be bidirected in $G_{in}$, which means that they cannot be directed edges between different SCCs, this contradicts the fact that $e$ is a bridge.
    Therefore, $|E_{out}| \geq 2$, and thus $|E_{rev}| \geq 1$.
\end{proof}

% \begin{theorem}
%     \ERS is a complete algorithm that is guaranteed to return a strongly connected graph.
% \end{theorem}

% \begin{proof}
%     According to \Cref{theorem:have-rev-edges}, we always reverse some edges in each iteration of \ERS if $G_{in}$ is not yet strongly connected. Such a reversal will always increase the number of meta vertices in $\mathcal{D}$ that have both incoming and outgoing edges. Eventually, all meta vertices in $\mathcal{D}$ have both incoming and outgoing edges. 
% \end{proof}

\subsection{Empirical Evaluation} \label{appen:efs-exp}

\subsubsection{Experiment Setup}

We compare \ERS with a MILP solver that reverses edges to ensure strong connectivity of a graph while optimizing the number of reversed edges. We formulate the MILP based on a prior work~\cite{duhamel_strong_2024}. To compare the MILP and our \ERS, we generate directed four-connected empty grid graphs of different sizes and run both algorithms to repair them for strong connectivity. For each graph, we generate 10 instances and give a 1 minute time limit to both methods. We measure (1) the number of reversed edges, (2) the CPU runtime, and (3) the success rate, defined as the number of graphs that each algorithm successfully solved without timeout.

\subsubsection{Experiment Result}

\Cref{fig:ers-vs-milp} shows the result. Although \ERS has a much worse solution quality compared to MILP.(that is, \ERS reverses more edges than MILP), it has a runtime significantly faster than MILP. This is especially important for an evolutionary-based solver such as CMA-ES for CMA-MAE because we run these methods for each generated mixed guidance graph. Using \ERS can make the optimization process much faster.
In addition, with a 1 minute time limit, the success rate of MILP drops to about 60\% in graphs with more than 3,000 vertices, while \ERS maintains a success rate of 100\%.

\section{Additional Experiments}

\subsection{Implementation and Compute Resources} \label{appen:imple-compute}

\subsubsection{Implementation}
We implement CMA-ES and CMA-MAE in Python with Pyribs~\cite{pyribs}, the update model in Pytorch~\cite{Paszke2019PyTorchAI}, and \ERS in Python. We implement PIBT and RHCR in C++ based on prior implementations~\cite{Li2020LifelongMP,Jiang2024Competition}. For RHCR, we use PBS~\cite{MaAAAI19} as the MAPF solver and SIPP~\cite{PhillipsICRA11} as the single-agent solver. To speed up optimization, we follow \citet{zhang2024ggo} to stop RHCR when the agents get congested (i.e. more than half of the agents wait in place).

\begin{figure}[t]
    \centering
    \includegraphics[width=0.5\linewidth]{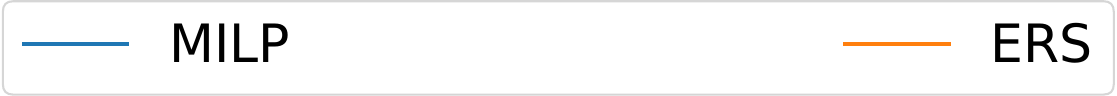}\\
    \begin{subfigure}[t]{0.16\textwidth}
        \includegraphics[width=\linewidth]{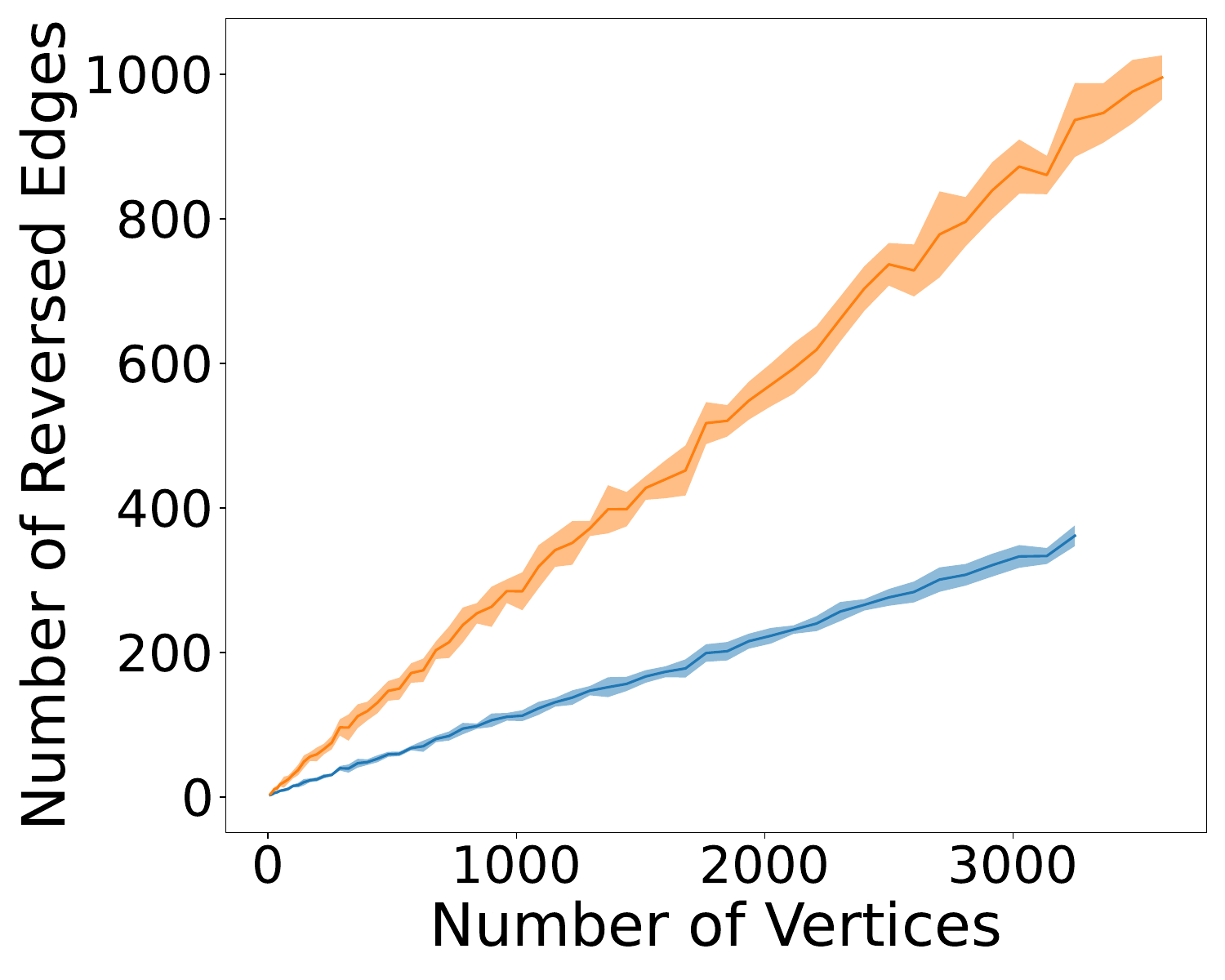}
    \end{subfigure}\hfill
    \begin{subfigure}[t]{0.16\textwidth}
        \includegraphics[width=\linewidth]{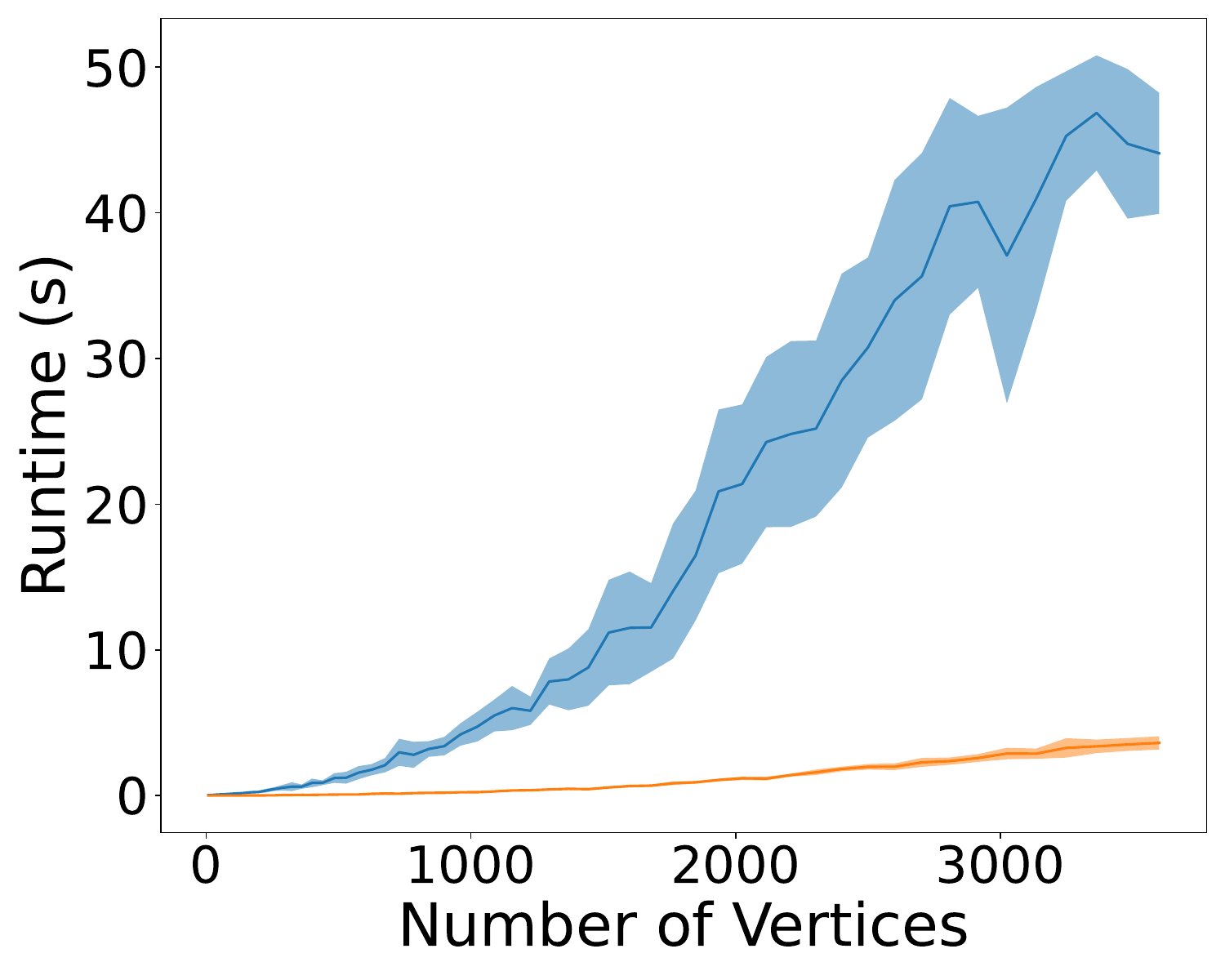}
    \end{subfigure}\hfill
    \begin{subfigure}[t]{0.16\textwidth}
        \includegraphics[width=\linewidth]{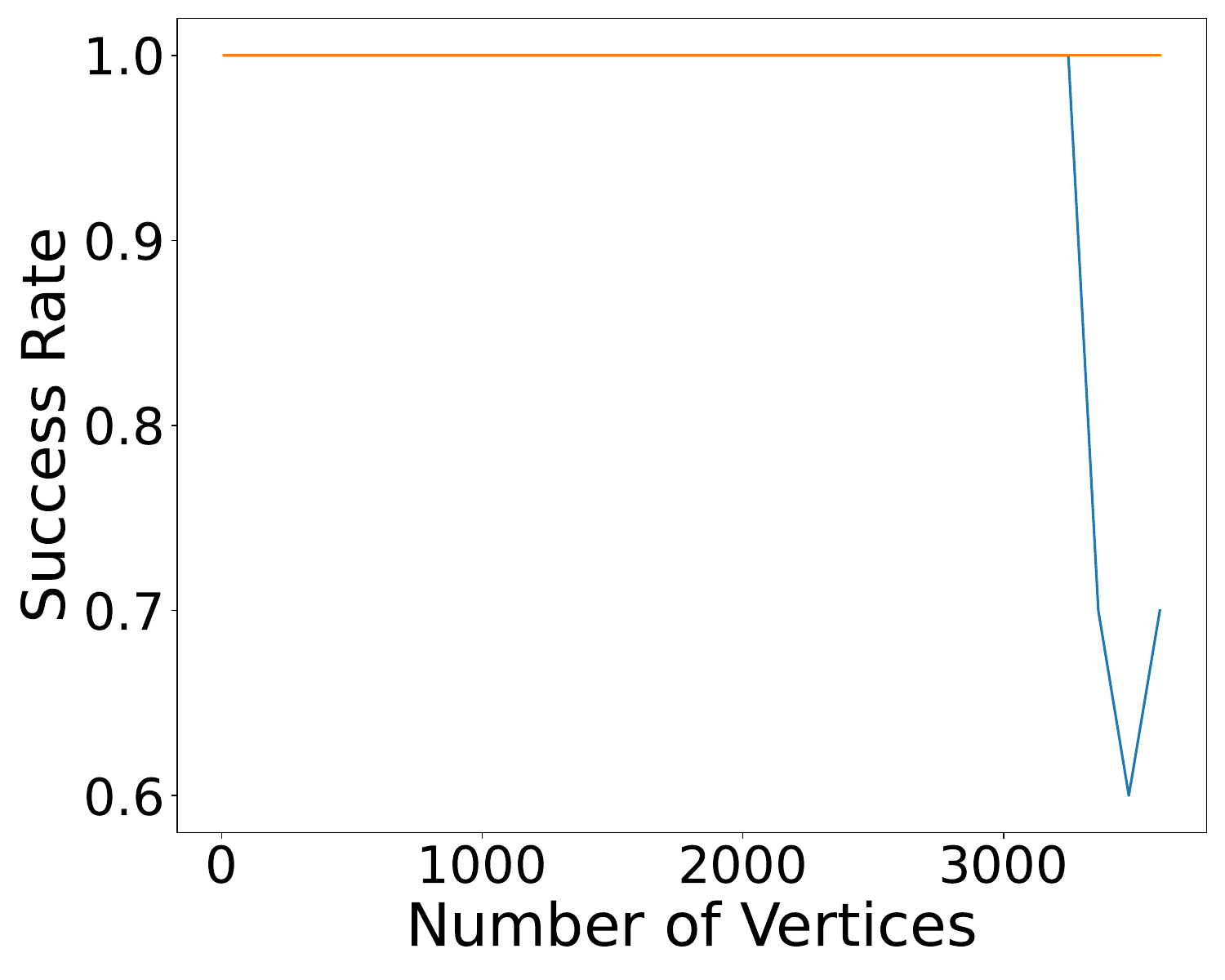}
    \end{subfigure}
    \caption{ERS vs. MILP.}
    \label{fig:ers-vs-milp}
\end{figure}

\begin{figure}[t]
    \centering
    \includegraphics[width=1\linewidth]{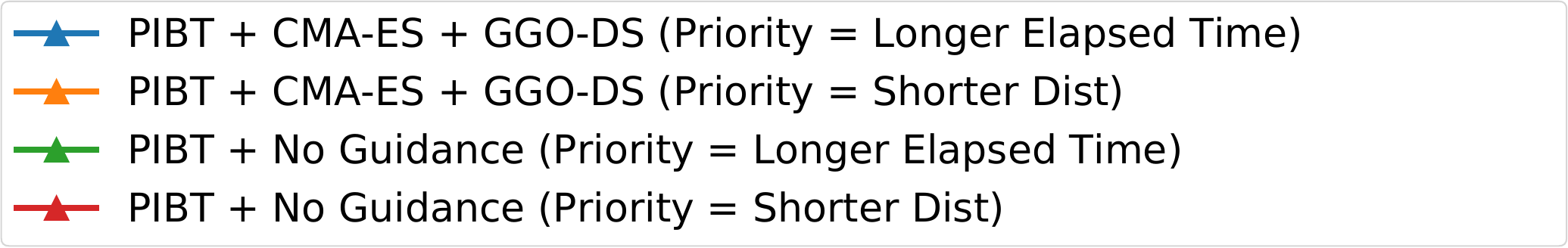}\\
    \begin{subfigure}[t]{0.16\textwidth}
        \includegraphics[width=\linewidth]{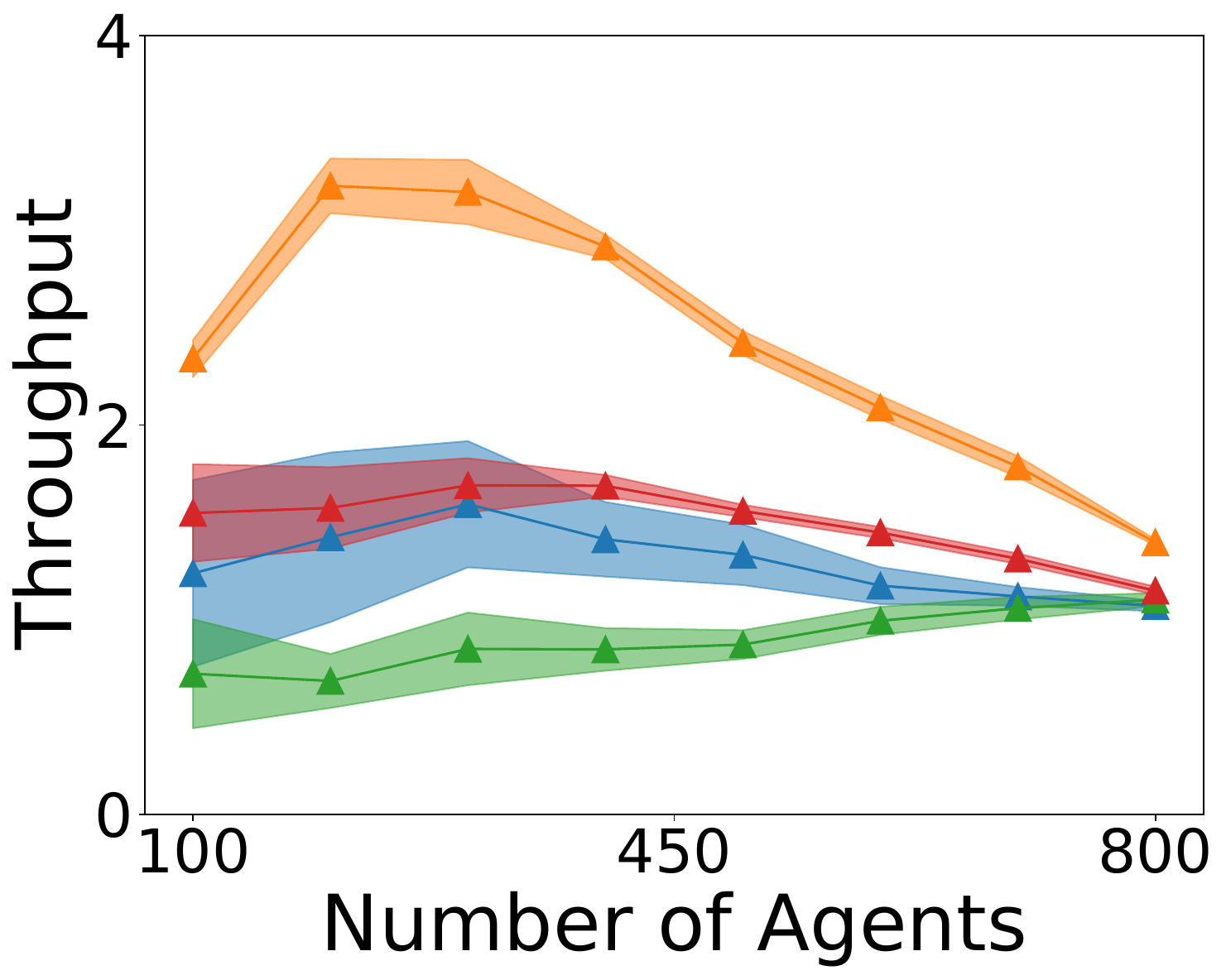}
    \caption{Setup 4}
    \end{subfigure}\hfill
    \begin{subfigure}[t]{0.16\textwidth}
        \includegraphics[width=\linewidth]{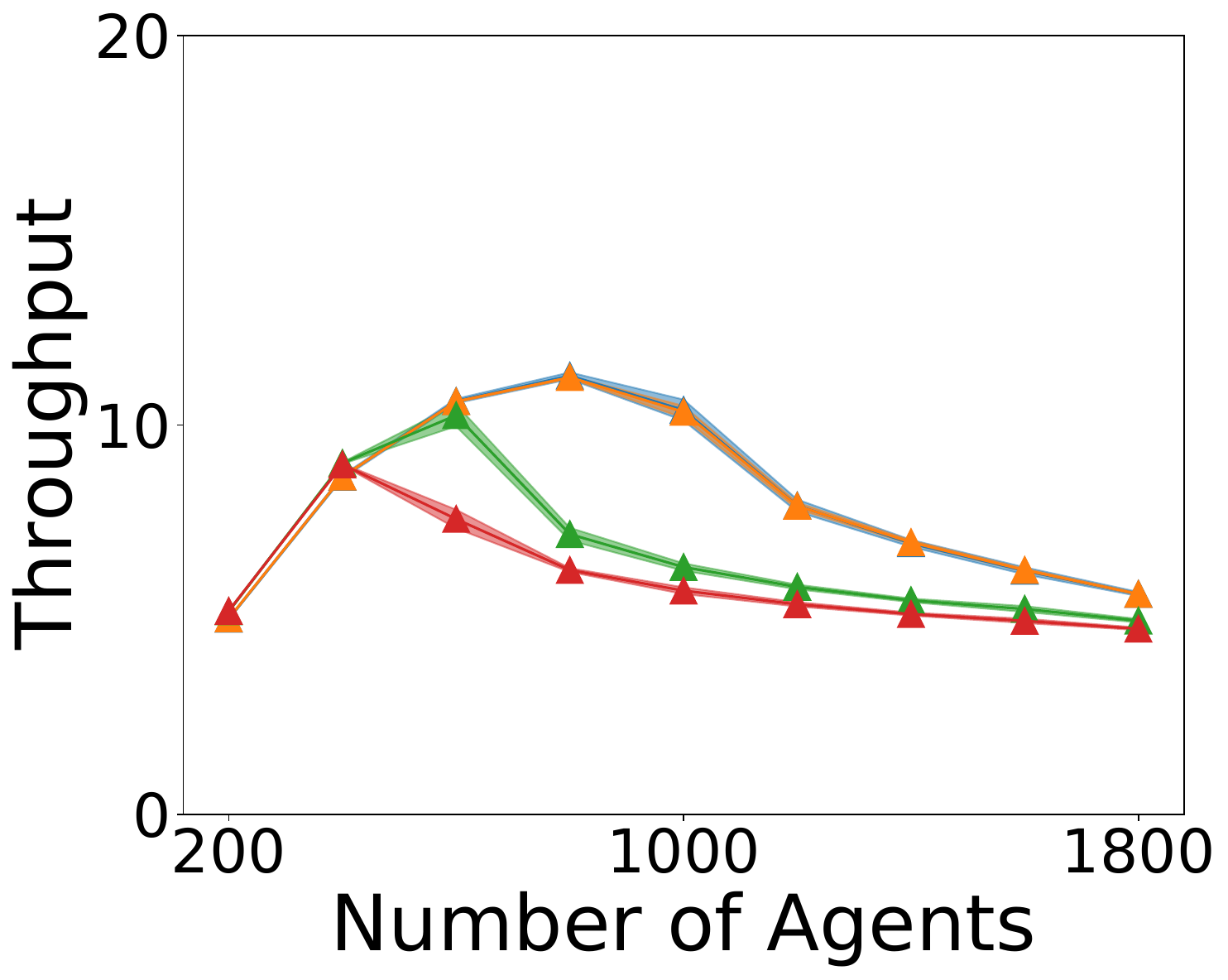}
        \caption{Setup 6}
    \end{subfigure}\hfill
    \begin{subfigure}[t]{0.16\textwidth}
        \includegraphics[width=\linewidth]{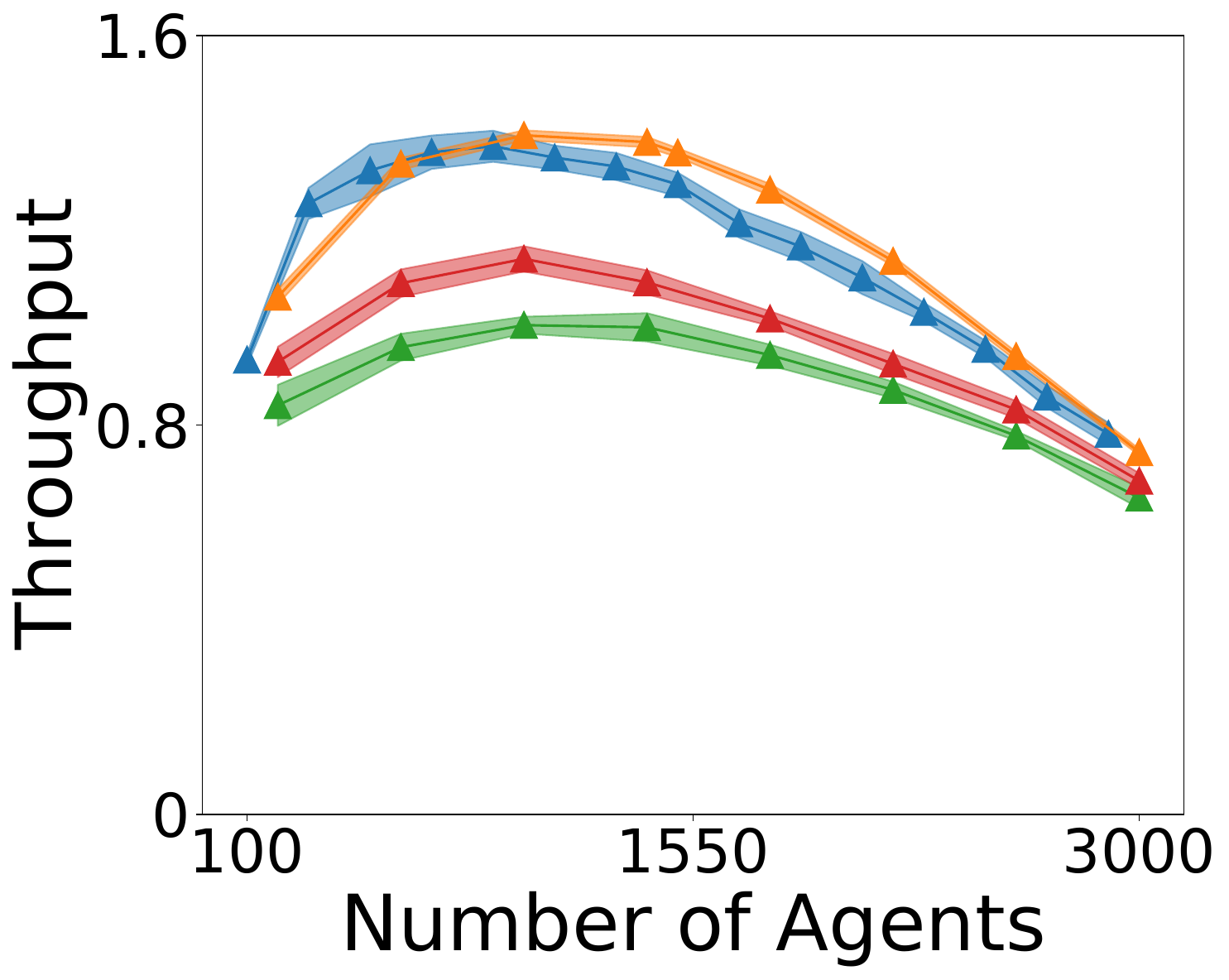}
        \caption{Setup 7}
    \end{subfigure}
    \caption{Comparison of different agent priority strategies of PIBT before and after GGO.}
    \label{fig:pibt-r-order-ablation}
\end{figure}

\paragraph{Agents Priorities in PIBT} We report an interesting side discovery regarding PIBT with the rotational motion model. In PIBT, agents move in a priority order. With the pebble motion model, \citet{okumura2019priority} establishes that priorities are determined based on the number of timesteps elapsed from the point when the agents reached their last goal. Agents that arrive at their last goal earlier shall move first.
On the other hand, \citet{Jiang2024Competition} suggests that with the rotational motion model, it is better if the priorities are determined by the distance of the agents to their next goal. Agents closer to their next goal shall move first. We discover that, in non-biconnected maps, using the latter can significantly reduce the occurrence of deadlocks and livelocks, improving average throughput and reducing the variance in evaluation. We compare the two priority strategies before and after running GGO-DS~\cite{zhang2024ggo} in setups 4, 6, and 7, two of which are non-biconnected (\randomSmall and \roomLarge) and one of which is biconnected (\emptyMid). 
We present the results in \Cref{fig:pibt-r-order-ablation}. It is clear that ordering agents by shorter distances to goals demonstrates much better performance in non-biconnected maps (setups 4 and 7). With the \emptyMid map in setup 6, ordering with longer elapsed time demonstrates better performance, but no significant gap is discovered after GGO.

\subsubsection{Compute Resources}
We run our experiments on five machines: (1) a local machine with a 64-core AMD Ryzen Threadripper 3990X CPU, 192 GB of RAM, (2) a local machine with a 64-core AMD Ryzen Threadripper 3990X CPU, 128 GB of RAM, (3) a local machine with a 64-core AMD Ryzen Threadripper 7980X CPU, 256 GB of RAM, (4) a local machine with a 16-core AMD Ryzen 9950X CPU, 64 GB of RAM, and (5) an HPC with numerous 64-core AMD EPYC 7742 CPUs, each with 256 GB of RAM~\cite{PSCBridgeTwo2021}. We measure all CPU runtime on machine (4).

\subsection{Additional Numerical Results} \label{appen:numerical-result}

\begin{table}[!t]
    \centering
    \resizebox{1\linewidth}{!}{
        \begin{tabular}{ccrrr}
\toprule
Setup & MAPF + (M)GGO & SR & Throughput & CPU Runtime \\
\midrule
\multirow{7}{*}{1} & \textbf{CMA-ES + Edge-Dir-Aware GGO-PU} & 94\% & $\textbf{4.93} \pm \textbf{0.03}$ & $29.47 \pm 0.81$\\
  & CMA-ES + GGO-DS & 92\% & $4.69 \pm 0.04$ & $33.77 \pm 0.85$\\
  & Directed Crisscross & 0\% & N/A & N/A\\
  & No Guidance & 0\% & N/A & N/A\\
  & QD + Joint MGGO-PU & 92\% & $4.08 \pm 0.03$ & $32.59 \pm 0.94$\\
  & Two-Phase MGGO-DS & \textbf{100\%} & $4.20 \pm 0.00$ & $\textbf{18.63} \pm \textbf{0.08}$\\
  & Two-Phase MGGO-DS (Phase 1) & \textbf{100\%} & $4.08 \pm 0.01$ & $23.78 \pm 0.16$\\
\midrule
\multirow{7}{*}{2} & CMA-ES + Edge-Dir-Aware GGO-PU & \textbf{100\%} & $4.75 \pm 0.01$ & $40.00 \pm 0.20$\\
  & CMA-ES + GGO-DS & 26\% & $2.07 \pm 0.03$ & $56.98 \pm 1.33$\\
  & Directed Crisscross & 0\% & N/A & N/A\\
  & No Guidance & 0\% & N/A & N/A\\
  & \textbf{QD + Joint MGGO-PU} & \textbf{100\%} & $\textbf{5.11} \pm \textbf{0.01}$ & $\textbf{29.70} \pm \textbf{0.19}$\\
  & Two-Phase MGGO-DS & 82\% & $3.58 \pm 0.01$ & $32.41 \pm 0.21$\\
  & Two-Phase MGGO-DS (Phase 1) & 90\% & $3.64 \pm 0.01$ & $41.43 \pm 0.24$\\
\midrule
\multirow{7}{*}{3} & \textbf{CMA-ES + Edge-Dir-Aware GGO-PU} & \textbf{100\%} & $\textbf{20.81} \pm \textbf{0.01}$ & $203.48 \pm 0.78$\\
  & CMA-ES + GGO-DS & 96\% & $18.44 \pm 0.02$ & $312.43 \pm 1.66$\\
  & Directed Crisscross & 4\% & $3.77 \pm 0.08$ & $226.38 \pm 0.07$\\
  & No Guidance & 0\% & N/A & N/A\\
  & QD + Joint MGGO-PU & \textbf{100\%} & $20.73 \pm 0.01$ & $189.15 \pm 0.50$\\
  & Two-Phase MGGO-DS & \textbf{100\%} & $19.84 \pm 0.01$ & $\textbf{181.16} \pm \textbf{0.53}$\\
  & Two-Phase MGGO-DS (Phase 1) & 80\% & $16.56 \pm 0.02$ & $248.92 \pm 1.39$\\
\midrule
\multirow{7}{*}{4} & CMA-ES + Edge-Dir-Aware GGO-PU & \textbf{100\%} & $3.00 \pm 0.02$ & $2.31 \pm 0.01$\\
  & CMA-ES + GGO-DS & \textbf{100\%} & $2.85 \pm 0.02$ & $2.25 \pm 0.01$\\
  & Directed Crisscross & \textbf{100\%} & $1.95 \pm 0.02$ & $2.14 \pm 0.01$\\
  & No Guidance & \textbf{100\%} & $1.65 \pm 0.02$ & $2.19 \pm 0.01$\\
  & QD + Joint MGGO-PU & \textbf{100\%} & $2.70 \pm 0.02$ & $\textbf{2.13} \pm \textbf{0.01}$\\
  & \textbf{Two-Phase MGGO-DS} & \textbf{100\%} & $\textbf{3.80} \pm \textbf{0.01}$ & $2.13 \pm 0.01$\\
  & Two-Phase MGGO-DS (Phase 1) & \textbf{100\%} & $3.56 \pm 0.01$ & $2.14 \pm 0.01$\\
\midrule
\multirow{7}{*}{5} & CMA-ES + Edge-Dir-Aware GGO-PU & \textbf{100\%} & $4.08 \pm 0.01$ & $2.84 \pm 0.02$\\
  & CMA-ES + GGO-DS & \textbf{100\%} & $3.81 \pm 0.04$ & $2.48 \pm 0.02$\\
  & Directed Crisscross & \textbf{100\%} & $3.20 \pm 0.01$ & $2.42 \pm 0.01$\\
  & No Guidance & \textbf{100\%} & $2.04 \pm 0.01$ & $2.42 \pm 0.02$\\
  & \textbf{QD + Joint MGGO-PU} & \textbf{100\%} & $\textbf{4.42} \pm \textbf{0.02}$ & $2.76 \pm 0.03$\\
  & Two-Phase MGGO-DS & \textbf{100\%} & $4.04 \pm 0.00$ & $2.56 \pm 0.01$\\
  & Two-Phase MGGO-DS (Phase 1) & \textbf{100\%} & $4.07 \pm 0.01$ & $\textbf{2.40} \pm \textbf{0.02}$\\
\midrule
\multirow{7}{*}{6} & \textbf{CMA-ES + Edge-Dir-Aware GGO-PU} & \textbf{100\%} & $\textbf{19.39} \pm \textbf{0.01}$ & $14.10 \pm 0.03$\\
  & CMA-ES + GGO-DS & \textbf{100\%} & $10.34 \pm 0.08$ & $12.57 \pm 0.03$\\
  & Directed Crisscross & \textbf{100\%} & $12.79 \pm 0.03$ & $12.76 \pm 0.03$\\
  & No Guidance & \textbf{100\%} & $6.35 \pm 0.02$ & $\textbf{11.84} \pm \textbf{0.03}$\\
  & QD + Joint MGGO-PU & \textbf{100\%} & $18.51 \pm 0.26$ & $14.12 \pm 0.03$\\
  & Two-Phase MGGO-DS & \textbf{100\%} & $16.43 \pm 0.02$ & $12.99 \pm 0.04$\\
  & Two-Phase MGGO-DS (Phase 1) & \textbf{100\%} & $14.49 \pm 0.02$ & $13.02 \pm 0.03$\\
\midrule
\multirow{7}{*}{7} & CMA-ES + Edge-Dir-Aware GGO-PU & \textbf{100\%} & $1.34 \pm 0.01$ & $26.53 \pm 0.06$\\
  & CMA-ES + GGO-DS & \textbf{100\%} & $1.36 \pm 0.00$ & $25.21 \pm 0.06$\\
  & Directed Crisscross & \textbf{100\%} & $1.04 \pm 0.00$ & $\textbf{22.42} \pm \textbf{0.05}$\\
  & No Guidance & \textbf{100\%} & $1.08 \pm 0.00$ & $24.81 \pm 0.06$\\
  & QD + Joint MGGO-PU & \textbf{100\%} & $1.46 \pm 0.00$ & $23.01 \pm 0.04$\\
  & \textbf{Two-Phase MGGO-DS} & \textbf{100\%} & $\textbf{1.54} \pm \textbf{0.00}$ & $22.73 \pm 0.05$\\
  & Two-Phase MGGO-DS (Phase 1) & \textbf{100\%} & $1.47 \pm 0.00$ & $22.51 \pm 0.04$\\
\midrule
\multirow{7}{*}{8} & CMA-ES + Edge-Dir-Aware GGO-PU & \textbf{100\%} & $2.11 \pm 0.01$ & $15.45 \pm 0.02$\\
  & CMA-ES + GGO-DS & \textbf{100\%} & $1.95 \pm 0.01$ & $14.23 \pm 0.03$\\
  & Directed Crisscross & \textbf{100\%} & $2.02 \pm 0.00$ & $14.24 \pm 0.02$\\
  & No Guidance & \textbf{100\%} & $1.09 \pm 0.00$ & $\textbf{13.88} \pm \textbf{0.03}$\\
  & QD + Joint MGGO-PU & \textbf{100\%} & $2.27 \pm 0.01$ & $15.16 \pm 0.02$\\
  & Two-Phase MGGO (Phase 1 Only) & \textbf{100\%} & $2.15 \pm 0.01$ & $14.18 \pm 0.02$\\
  & \textbf{Two-Phase MGGO-DS} & \textbf{100\%} & $\textbf{2.47} \pm \textbf{0.01}$ & $14.35 \pm 0.03$\\
\bottomrule
        \end{tabular}
    }
    \caption{Success rates (\textit{SR}), throughput, and CPU runtimes of the simulations on different mixed guidance graphs.
   % \textit{SR} refers to the success rate. 
   For RHCR, the success rate is the percentage of simulations that end without congestion. For PIBT, it is the percentage of simulations that end without timeout.
   We measure the throughput and CPU runtime over only successful simulations.}
    \label{tab:numerical-result}
\end{table}

\Cref{tab:numerical-result} shows the numerical results.
We run RHCR and PIBT on all graphs with the same $N_a$ and $T$ in \Cref{tab:exp-setup} for 50 times and report the numerical results in the format of $x \pm y$, where $x$ and $y$ are the average and standard error, respectively. The numerical results follow the same trend with those in \Cref{fig:major-result-pibt-r,fig:major-result-rhcr-r}.

\subsection{Variant of Joint MGGO-PU} \label{appen:qd-vs-cma-es-mggo}

\begin{figure}[!t]
    \centering
    \includegraphics[width=1\linewidth]{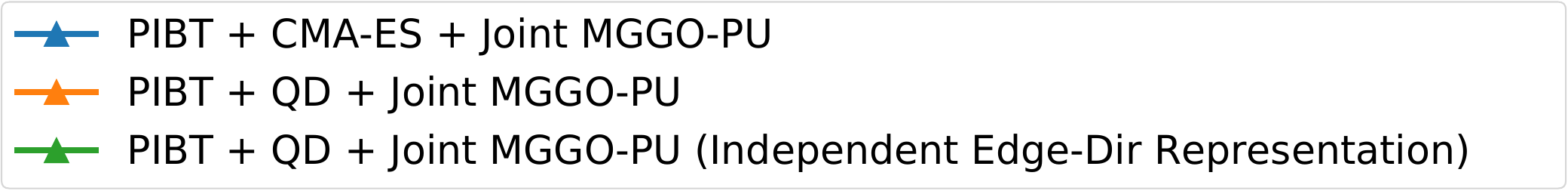}\\
    \begin{subfigure}[t]{0.16\textwidth}
        \includegraphics[width=\linewidth]{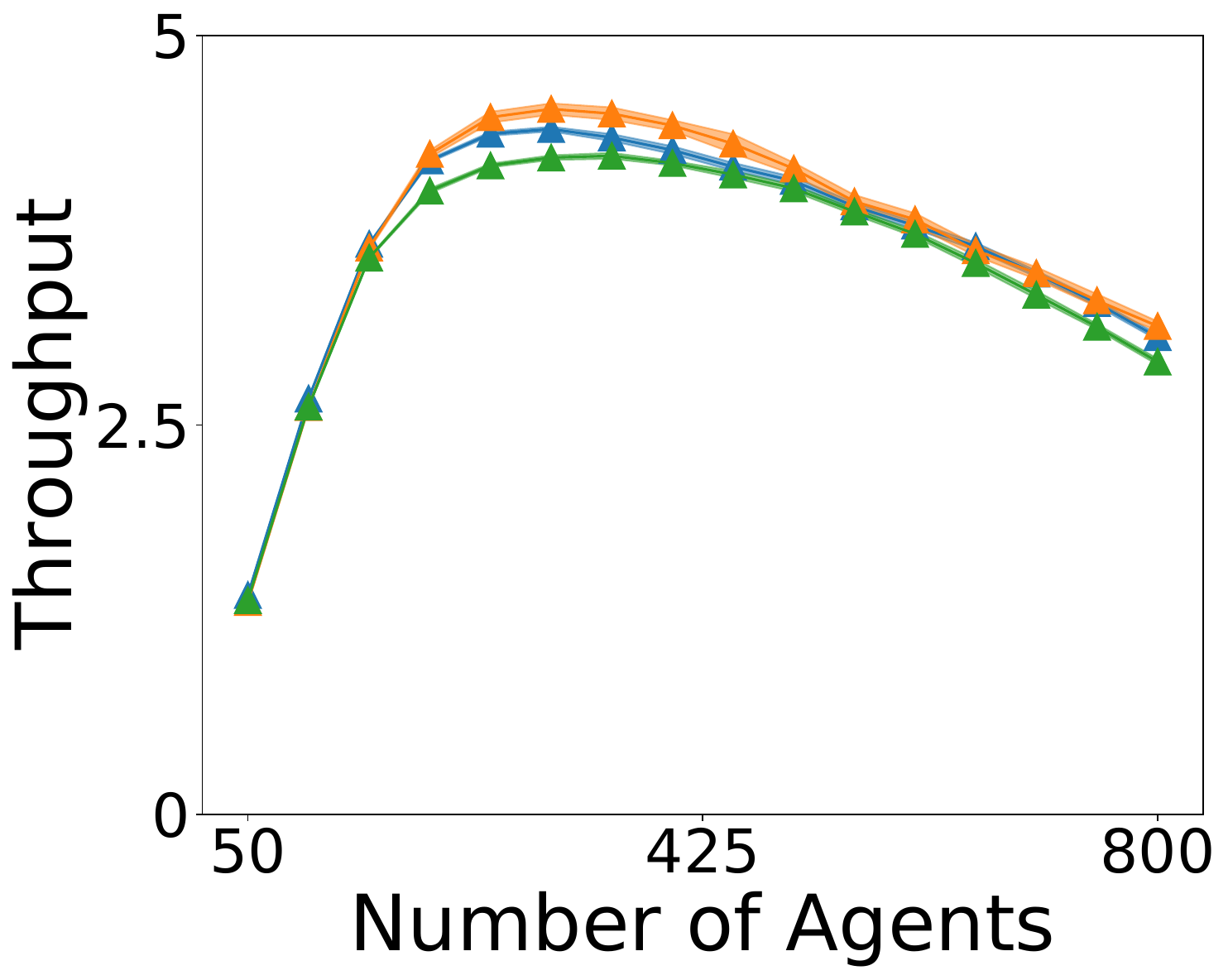}
    \end{subfigure}\hfill
    \begin{subfigure}[t]{0.16\textwidth}
        \includegraphics[width=\linewidth]{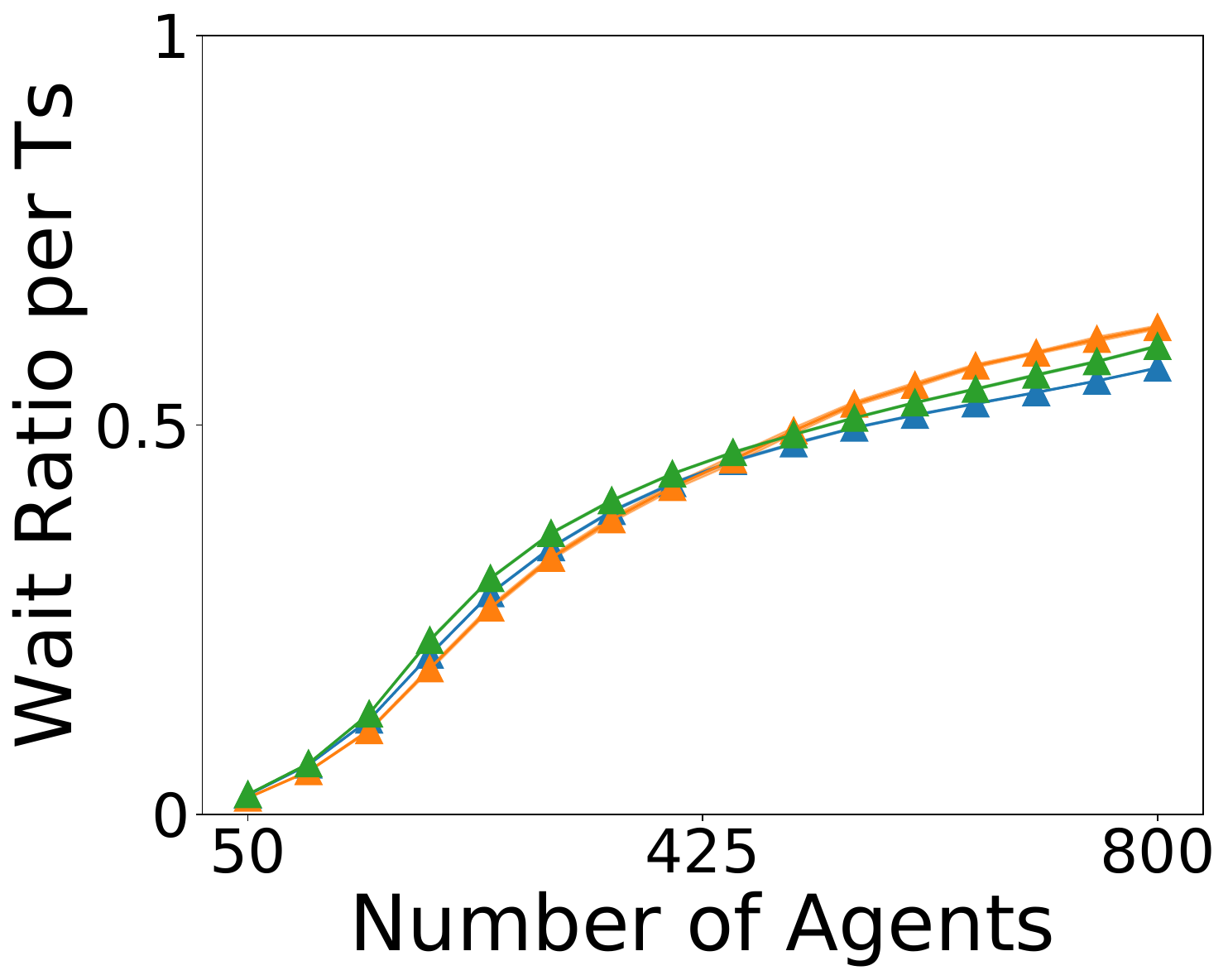}
    \end{subfigure}\hfill
    \begin{subfigure}[t]{0.16\textwidth}
        \includegraphics[width=\linewidth]{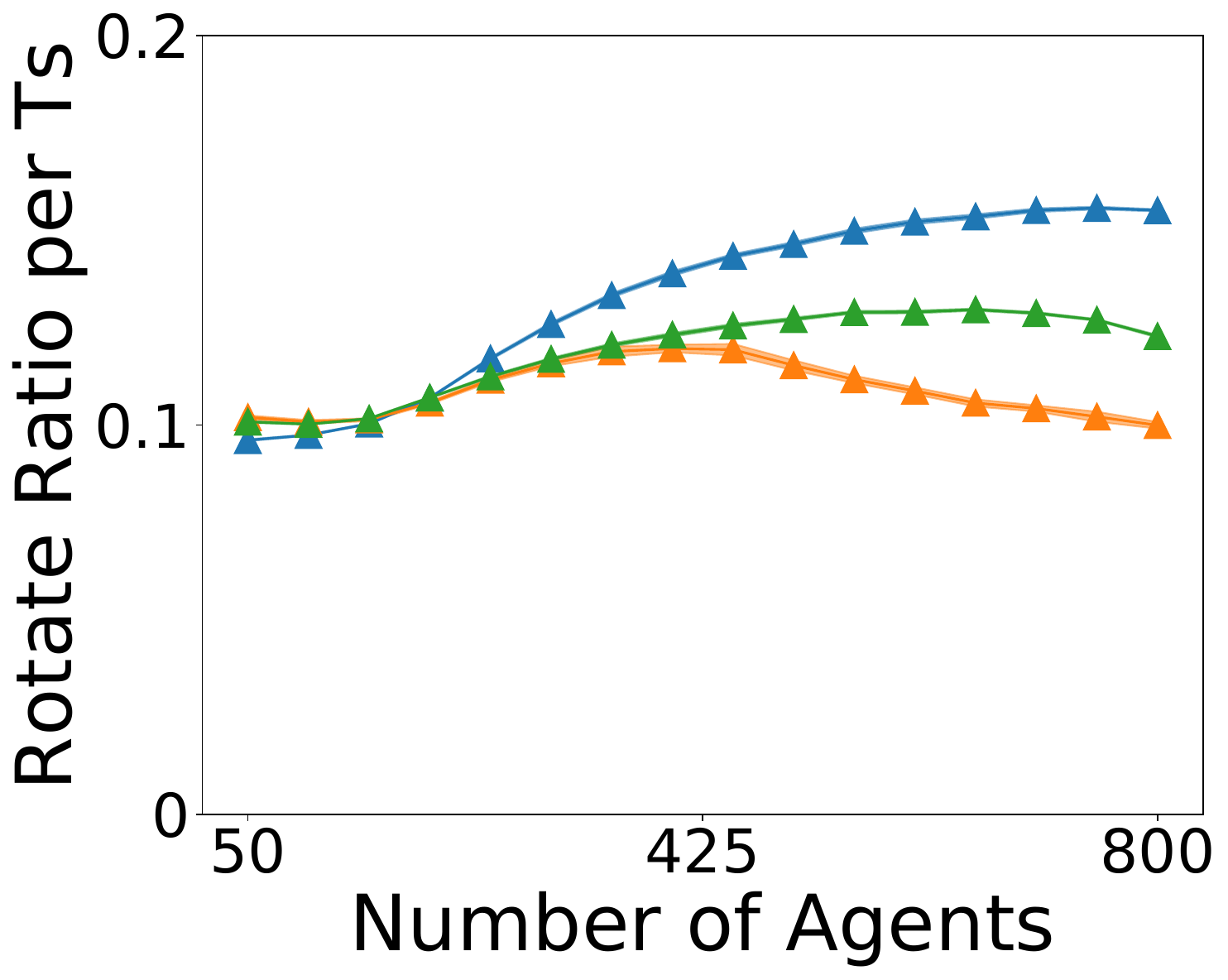}
    \end{subfigure}
    \caption{Throughput, wait action ratio, and rotation action ratio of Joint MGGO-PU ablations experiments.}
    \label{fig:cma-es-qd-mggo-pu}
\end{figure}

We run PIBT with different numbers of agents, each 50 times, and report the average as well as the 95\% confidence interval in \Cref{fig:cma-es-qd-mggo-pu}. Interestingly, QD is able to reduce the ratio of rotation actions during the simulation. This is because the mixed guidance graph found by QD has about 70\% of directed edges, while the one optimized by CMA-ES only has 14\%.

\subsection{Variants of EA in MGGO-DS Phase One} \label{appen:gdo-ablation}

\begin{figure}[!t]
    \centering
    \includegraphics[width=0.95\linewidth]{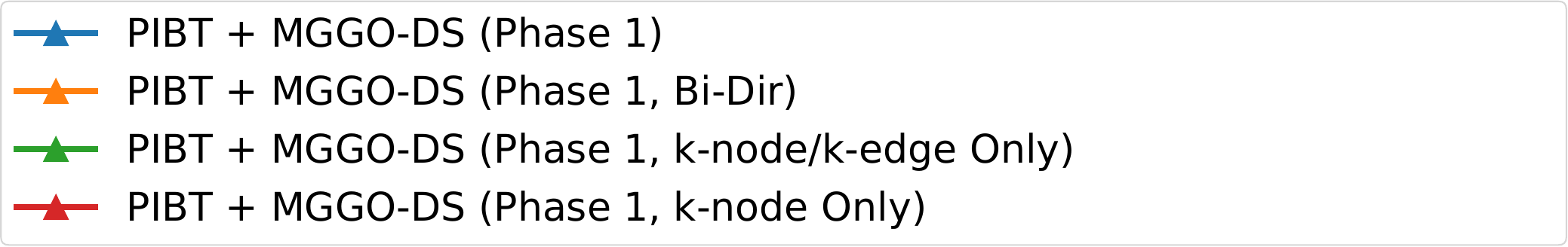}\\
    \begin{subfigure}[t]{0.16\textwidth}
        \includegraphics[width=\linewidth]{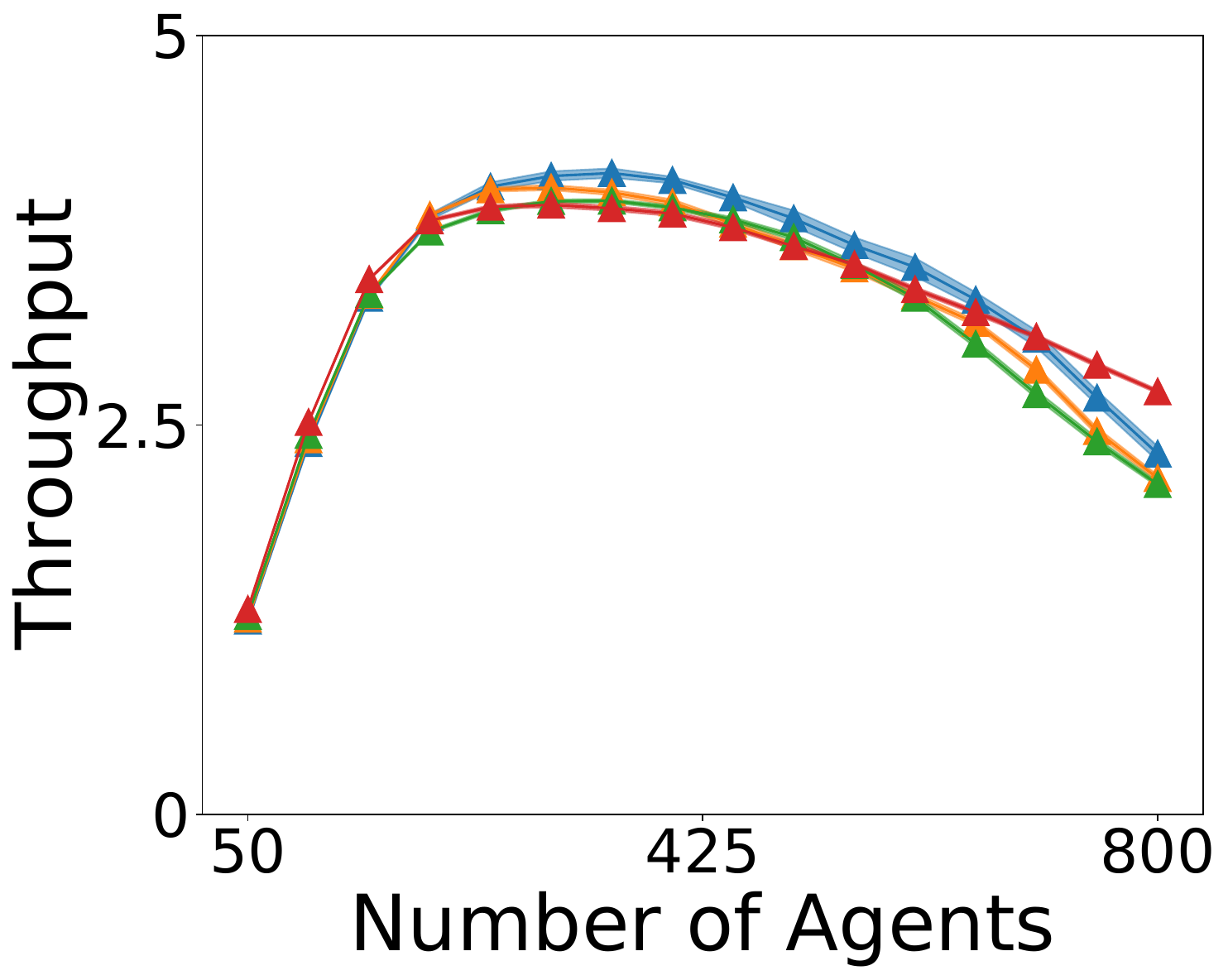}
    \end{subfigure}\hfill
    \begin{subfigure}[t]{0.16\textwidth}
        \includegraphics[width=\linewidth]{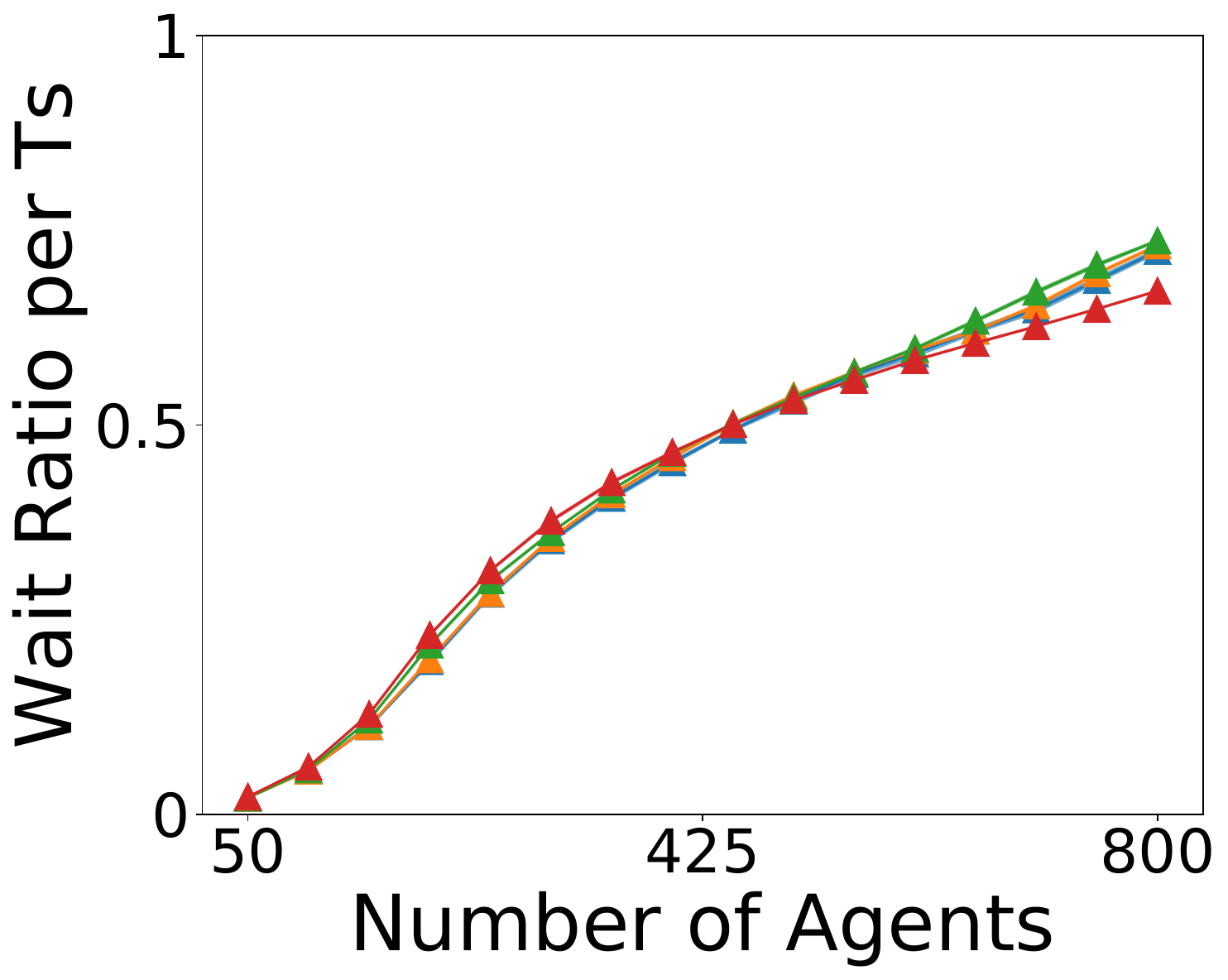}
    \end{subfigure}\hfill
    \begin{subfigure}[t]{0.16\textwidth}
        \includegraphics[width=\linewidth]{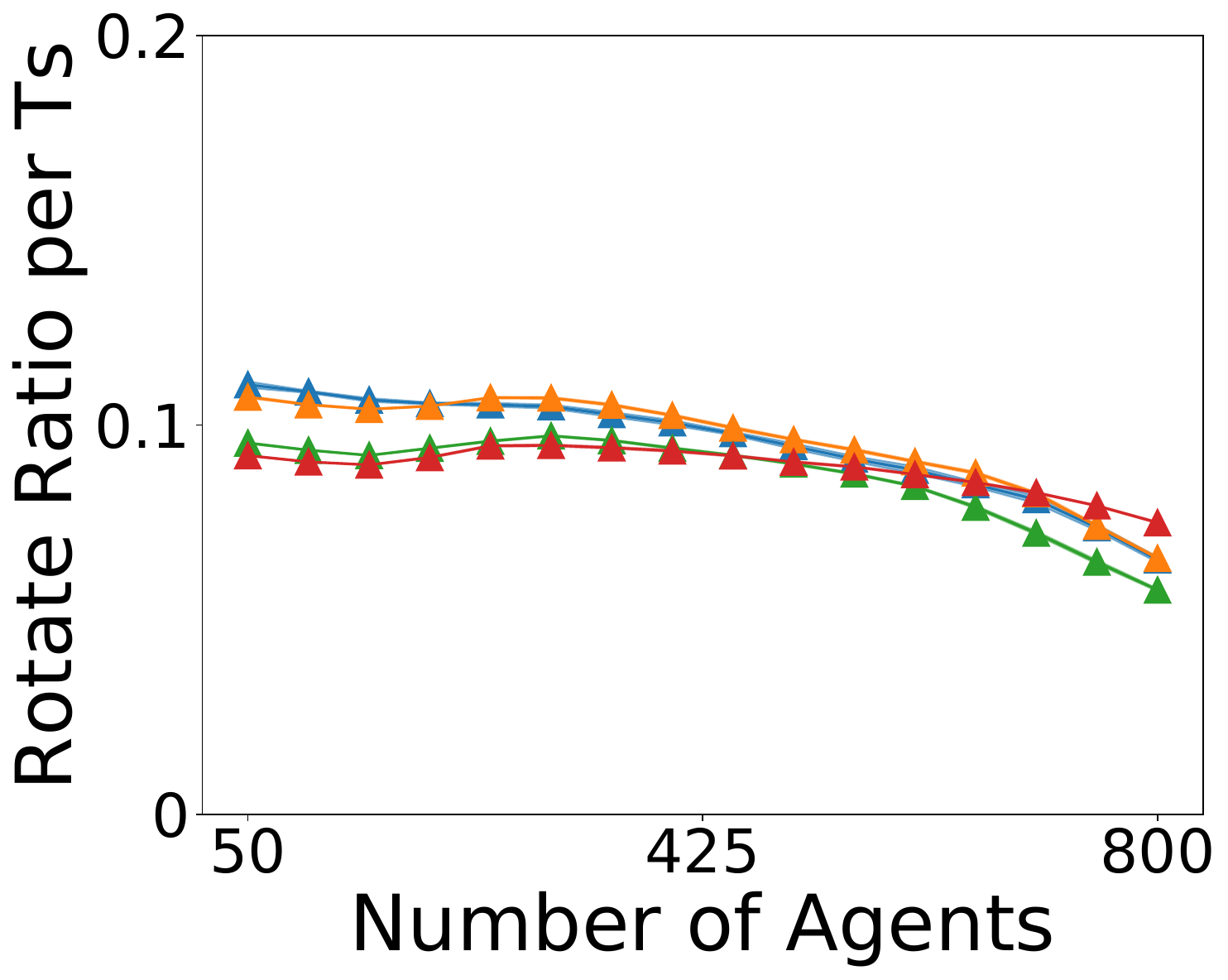}
    \end{subfigure}
    \caption{Throughput, wait action ratio, and rotation action ratio of $(1+\lambda)$ EA ablations experiments.}
    \label{fig:gdo-ablation}
\end{figure}

We run PIBT with different numbers of agents, each 50 times, and report the average as well as the 95\% confidence interval in \Cref{fig:gdo-ablation}. The comparison in throughput aligns with our findings in \Cref{sec:exp}. For the ratio of waiting and rotating actions, we do not observe significant differences, as all optimized mixed guidance graphs have edge weights 1 and the number of directed edges are maximized.

\section{Optimized Mixed Guidance Graphs} \label{appen:example-mgg}

\begin{figure*}
    \centering
    \begin{subfigure}[t]{0.48\textwidth}
        \includegraphics[width=\linewidth]{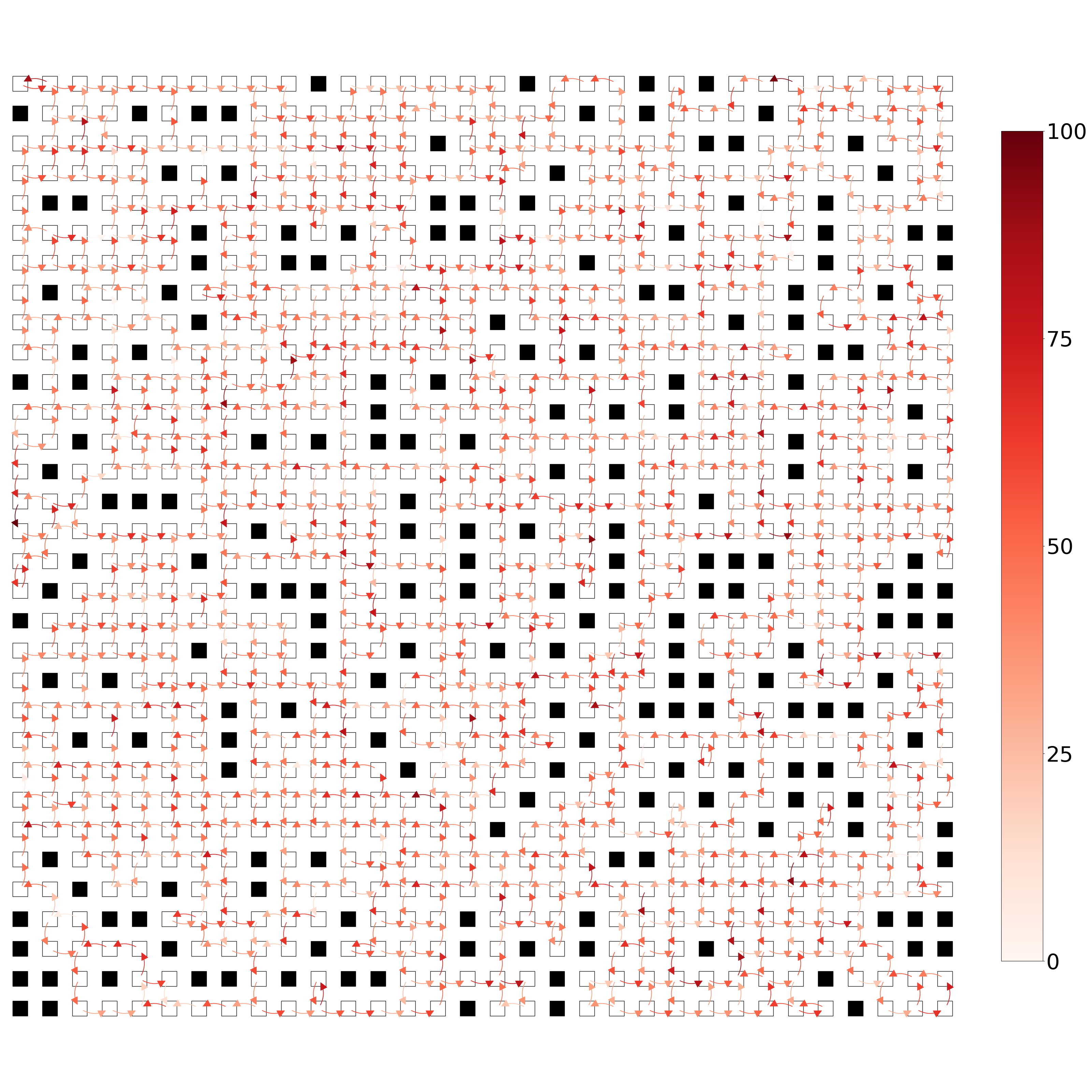}
    \caption{Setup 4: \randomSmall}
    \end{subfigure}\hfill
    \begin{subfigure}[t]{0.52\textwidth}
        \includegraphics[width=\linewidth]{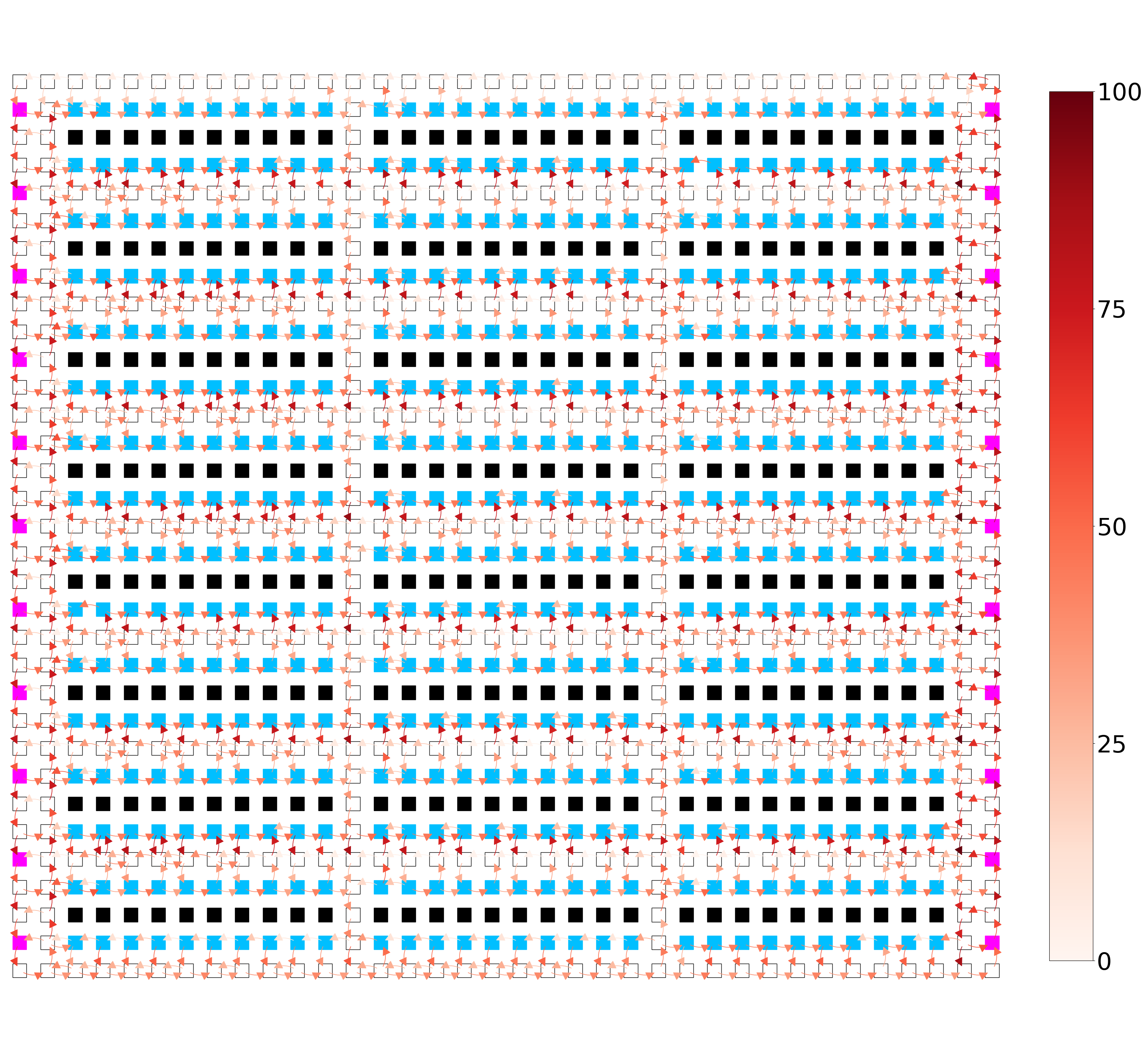}
    \caption{Setup 5: \warehouselargeW}
    \end{subfigure}
    \caption{Example optimized mixed guidance graphs. The arrows indicate edge directions and the color indicates edge weights.}
    \label{fig:example-mgg}
\end{figure*}

\Cref{fig:example-mgg} shows example optimized mixed guidance graphs. Due to the size of the maps, we are only able to show the two smallest maps used in the experiments.

% \subsection{PIU vs. PU} \label{appen:pu-vs-piu}

% \subsection{Representation of Edge Directions} \label{appen:mgg-represent}

% \subsection{Incorporate Bidirected Edges to Phase-One MGGO}

% \subsection{Scalability of Edge Direction Flip Search}

\end{document}